\documentclass[journal]{IEEEtran}
\usepackage{epsfig,graphicx,subfigure,psfrag,amsmath,cases}
\usepackage{latexsym,amssymb,amsmath,epsfig,subfigure,algorithm,mathtools}
\usepackage{algorithmic}
\usepackage{color}
\usepackage{url}
\usepackage{scrtime}
\usepackage{array}
\author{\authorblockN{ Derrick Wing Kwan Ng,~\IEEEmembership{Member,~IEEE,} Ernest S. Lo,~\IEEEmembership{Member,~IEEE,} \\
and Robert Schober,~\IEEEmembership{Fellow,~IEEE}}\thanks{This paper has been presented  in part at IEEE PIMRC 2013 \cite{CN:Kwan_PIMRC2013}. Derrick Wing Kwan Ng and Robert Schober are with the Institute for Digital Communications (IDC),
Friedrich-Alexander-University Erlangen-N\"urnberg (FAU), Germany (email:\{kwan, schober\}@lnt.de). The authors are also with the University of British Columbia, Vancouver, Canada. \newline
Ernest S. Lo is with the  Centre Tecnol\`{o}gic de Telecomunicacions de Catalunya - Hong Kong (CTTC-HK) (email: ernest.lo@cttc.hk). Derrick Wing Kwan Ng was supported by the Qatar National Research Fund (QNRF), under project NPRP 5-401-2-161. Robert Schober was supported by the AvH Professorship Program of the Alexander von Humboldt Foundation.}
    \\
}

\title{Multi-Objective Resource Allocation for Secure  Communication in Cognitive Radio Networks with Wireless Information and Power Transfer}

\newtheorem{Thm}{Theorem}
\newtheorem{Lem}{Lemma}
\newtheorem{Cor}{Corollary}

\newtheorem{Prob}{Problem}
\newtheorem{T-Prob}{Transformed Problem}
\newtheorem{proposition}{Proposition}

\newtheorem{Remark}{Remark}

\DeclareMathOperator{\diag}{\mathrm{diag}}
\DeclareMathOperator{\bigo}{\cal O}
\DeclareMathOperator{\Tr}{\mathrm{Tr}}
\DeclareMathOperator{\Rank}{\mathrm{Rank}}
\DeclareMathOperator{\maxo}{\mathrm{maximize}}
\DeclareMathOperator{\mino}{\mathrm{minimize}}
\newcommand{\abs}[1]{\lvert#1\rvert}
\newcommand{\norm}[1]{\lVert#1\rVert}
\newcolumntype{L}{>{\arraybackslash\raggedright}m{3.5cm}}
\newcommand{\textoverline}[1]{$\overline{\mbox{#1}}$}

\definecolor{kwan}{rgb}{0,0,1}
\makeatletter

\newcommand{\Rmnum}[1]{\expandafter\@slowromancap\romannumeral #1@}
\makeatother
\begin{document}
\maketitle

\begin{abstract}
 In this paper, we study resource allocation for multiuser multiple-input single-output (MISO) secondary communication systems with multiple system design objectives. We consider  cognitive radio (CR) networks where the secondary receivers  are able to harvest energy from the radio frequency when they are idle. The secondary system provides
simultaneous wireless power and secure information transfer to the secondary receivers.  { We propose a multi-objective optimization framework for the design of a Pareto optimal resource allocation algorithm  based on the weighted Tchebycheff approach. In particular, the  algorithm design incorporates three important system design objectives: total transmit power minimization, energy harvesting efficiency maximization, and interference power leakage-to-transmit power ratio minimization.  The proposed framework
 takes into account a quality of service (QoS) requirement regarding communication secrecy in the secondary system  and the imperfection of the  channel state information (CSI) of potential eavesdroppers (idle secondary receivers and primary receivers) at the secondary transmitter. The proposed framework  includes total harvested power  maximization and interference power leakage minimization as special cases.}   The adopted multi-objective optimization problem is non-convex and is recast as a convex optimization problem via  semidefinite programming (SDP) relaxation. It is shown that the global optimal solution of the original problem can be constructed by exploiting both the primal and the dual optimal solutions of the SDP relaxed problem. Besides, two suboptimal resource allocation schemes for the case when the solution of the dual problem  is unavailable for constructing the optimal solution are proposed.  Numerical results not only demonstrate the close-to-optimal performance of the proposed suboptimal schemes, but also unveil
an interesting trade-off between the considered conflicting system design  objectives.
\end{abstract}

%\newpage \setcounter{page}{1}
\section{Introduction}
\label{sect1}
 \IEEEPARstart{T}{he} explosive growth of the demand for  ubiquitous, secure, and high data rate
wireless communication services  has led to a tremendous solicitation of limited radio resources such as bandwidth and energy. In practice, fixed spectrum allocation has been implemented for resource sharing in traditional wireless communication systems. Although interference can be avoided by assigning different wireless services to different licensed frequency bands, such a fixed spectrum allocation strategy may result in spectrum under utilization. In fact, the Federal Communications Commission (FCC) has reported that $70$ percent of the allocated spectrum in the United States is not fully utilized, cf. \cite{Report:CR}.  As a result, cognitive radio (CR) has emerged as one of the most
promising solutions to improve  spectrum efficiency \cite{JR:CR_overview}. In particular,  CR enables a secondary system to access the spectrum of a primary  system  as long as the
interference from the secondary system does not  severely degrade the quality of service
(QoS) of the primary system  \cite{Report:CR}\nocite{JR:CR_tutorials,JR:CR_overview,JR:CR_WSN,JR:CR_WSN2,JR:CR_sensing1,JR:CR_sensing2,CN:CR_beamforming_1,JR:CR_beamforming_1}--\cite{JR:CR_beamforming_2}.
CR is not only applicable to traditional cellular networks, but also has the potential to improve the performance of wireless sensor networks \cite{JR:CR_WSN,JR:CR_WSN2}.
 In \cite{JR:CR_sensing1} and \cite{JR:CR_sensing2}, cooperative spectrum sensing  and the sensing-throughput trade-off were studied  for single antenna systems, respectively. In \cite{CN:CR_beamforming_1}, joint beamforming and power control was studied for transmit power minimization in  multiple-transmit-antenna CR downlink systems.  In \cite{JR:CR_beamforming_1} and \cite{JR:CR_beamforming_2}, by taking into account the imperfectness of channel state information, robust beamforming designs were proposed for  CR networks with single and multiple secondary users, respectively.  Furthermore, a detailed  performance analysis of transmit antenna selection  in multiple-antenna networks was presented in  \cite{JR:antenna_selection} for multi-relay networks and then extended to CR relay networks in \cite{JR:antenna_selection2}. However, since the transmit precoding strategies in  \cite{JR:antenna_selection}  and \cite{JR:antenna_selection2} are not optimized, they do not fully exploit the available degrees of freedom in the network for maximizing the system performance.

Although the current spectrum scarcity  may be partially overcome by CR technology, wireless communication devices, such as wireless sensors, are often powered by batteries  with limited
energy storage capacity. This constitutes another major bottleneck for providing communication services and extending
the lifetime of networks. On the other hand, energy harvesting is envisioned to provide a perpetual energy source to facilitate  self-sustainability  of power-constrained communication devices \cite{Powercast}\nocite{JR:harvesting_single_user,JR:harvesting_single_user,Krikidis2014,Ding2014,JR:SWIPT_mag,CN:WIPT_fundamental,
CN:Shannon_meets_tesla,JR:WIPT_bruno,JR:WIP_receiver,JR:Kai_bin}--\cite{JR:WIPT_fullpaper}.  In addition to conventional renewable energy sources such as biomass, wind, and solar,     wireless  power transfer  has emerged as a new option for prolonging the lifetime of  battery-powered wireless devices.
Specifically, the transmitter can transfer energy to the receivers via  electromagnetic waves in radio frequency
(RF).  Nowadays, energy harvesting circuits are able
to harvest microwatt to milliwatt of power over the range
of several meters for a transmit power of $1$ Watt and a
carrier frequency of less than $1$ GHz \cite{Powercast}. Thus, RF energy
can be a viable energy source for devices with low-power
consumption, e.g. wireless sensors \cite{Krikidis2014,Ding2014}.  The integration of RF energy harvesting capabilities into communication systems  provides the possibility of simultaneous wireless information and power transfer (SWIPT) \cite{Krikidis2014}--\cite{JR:WIPT_fullpaper}. As a result, in addition to the traditional QoS constraints such as communication reliability, efficient energy transfer is  expected to play an important role as a new QoS requirement. This new requirement introduces a paradigm shift in the design of both  resource allocation algorithms and  transceiver signal processing.  In \cite{CN:WIPT_fundamental} and \cite{CN:Shannon_meets_tesla}, the fundamental trade-off between the maximum achievable data rate and energy transfer was studied for a noisy single-user communication channel and a pair of  noisy coupled-inductor circuits, respectively.
  Then, in \cite{JR:WIPT_bruno}, the authors extended the trade-off study to a two-user multiple-antenna transceiver system. In \cite{JR:WIP_receiver}, the authors proposed  \emph{separated receivers} for SWIPT to facilitate low-complexity receiver design;  these receivers can be built by using  off-the-shelf components. In \cite{JR:Kai_bin}, different resource allocation algorithms were designed for broadband far field wireless systems with  SWIPT. In  \cite{JR:WIPT_fullpaper},  the authors showed that the energy efficiency of a communication system can be improved by  RF energy harvesting at the receivers. Nevertheless,  resource allocation algorithms maximizing the energy harvesting efficiency of  SWIPT CR systems have not been reported in the literature yet. Besides,
  two conflicting system design objectives arise naturally for a CR network providing SWIPT service to the secondary receivers  in practice. On the one hand, the secondary  transmitter should transmit with high power to facilitate energy transfer to the energy harvesting receivers. On the other hand, the secondary transmitter should transmit with low power to cause minimal interference  at the primary receivers. Thus,  considering these  conflicting system design objectives,  the single objective resource allocation algorithms proposed in \cite{CN:WIPT_fundamental}--\cite{JR:WIPT_fullpaper} may not be applicable in SWIPT CR networks.  Furthermore, transmitting with high signal power may also cause substantial  information leakage and high vulnerability to eavesdropping.

Recently, physical (PHY) layer security
 has attracted much attention in the research community  for preventing eavesdropping
\cite{JR:EE-sec}--\nocite{JR:EE_secrecy,JR:ken_artifical_noise,JR:CR_phy,JR:Kwan_secure_imperfect,
JR:rui_zhang,JR:Kwan_SEC_DAS,CN:kwan_vicky,CN:massive_MIMO_security_SWIPT}\cite{JR:PHY_CR}.  In \cite{JR:EE-sec}, the authors
proposed a beamforming scheme for maximization of the energy efficiency of secure communication systems. In \cite{JR:EE_secrecy} and \cite{JR:ken_artifical_noise},  the spatial degrees of freedom
offered by multiple antennas were used to degrade the channel of the eavesdroppers deliberately   via artificial noise transmission.  Thereby,
 communication secrecy was guaranteed at the expense of allocating  a large portion of the transmit power to artificial noise generation.  In \cite{JR:CR_phy}, the authors addressed the power allocation problem in CR secondary systems with PHY layer security provisioning.  However, the resource allocation algorithm designs in \cite{JR:EE-sec}--\cite{JR:CR_phy} cannot be directly extended to the case of
 of RF energy harvesting due to the differences in the underlying system models. On the other hand, \cite{CN:Kwan_PIMRC2013} and  \cite{JR:rui_zhang} studied different resource allocation algorithms for providing secure communication in systems with
separated information and energy harvesting receivers. Yet, the assumption of having perfect
channel state information (CSI) of the energy harvesting receivers in \cite{CN:Kwan_PIMRC2013} and  \cite{JR:rui_zhang} may  be too optimistic if the energy harvesting receivers do not interact with the transmitter periodically. In \cite{JR:Kwan_secure_imperfect}, the case
where the transmitter has only imperfect CSI of the energy harvesting receivers was considered
and a robust beamforming design was proposed to minimize the total transmit power of a
system with simultaneous energy and secure information transfer. In \cite{JR:Kwan_SEC_DAS}, the authors studied resource allocation algorithm design for secure information and renewable green
energy transfer to mobile receivers in distributed antenna
communication systems. In \cite{CN:kwan_vicky} and  \cite{CN:massive_MIMO_security_SWIPT}, beamforming algorithm  design and secrecy outage capacity was studied  for multiple-antenna potential eavesdropper  and passive eavesdroppers, respectively.
 However, the beamforming algorithms developed in  \cite{JR:Kwan_secure_imperfect}--\cite{CN:kwan_vicky} may not be applicable in CR networks. Furthermore, in \cite{JR:PHY_CR}, the secrecy
outage probability of CR networks was investigated in the presence of a passive eavesdropper.

Form the above discussions, we conclude that for
CR communication systems providing simultaneous wireless energy transfer and secure communication
services, conflicting system design objectives such as total transmit power minimization,
energy harvesting efficiency maximization, and interference power leakage-to-transmit power
ratio minimization play an important role for resource allocation. However, the problem formulations
in \cite{CN:WIPT_fundamental}--\cite{JR:PHY_CR} focus on a single system design objective and cannot be used to study the trade-off between the aforementioned conflicting design goals.  In this paper, we address the above issues and the contributions of the paper are summarized as follows:
\begin{itemize}
\item Different from our previous work in \cite{JR:Kwan_secure_imperfect}, in this paper, we propose a new non-convex multi-objective optimization problem with the aim to jointly minimize the total transmit power, maximize the energy harvesting efficiency, and minimize the interference power leakage-to-transmit power ratio for CR networks with SWIPT. The problem formulation takes into account the  imperfectness of the CSI of potential eavesdroppers (idle secondary receivers) and primary receivers in secondary multiuser multiple-input single-output (MISO) systems with RF energy harvesting receivers.  The solution of the optimization problem leads to a set of Pareto optimal resource allocation policies.
    \item The considered non-convex optimization problem  is recast as a convex optimization problem via  semidefinite programming (SDP) relaxation. We show that the global optimal solution of the original problem can be constructed by exploiting both the primal and the dual optimal solutions of the SDP relaxed problem.
        \item The obtained solution structure  is also applicable to the multi-objective optimization of the  total harvested power, the interference power leakage, and the total transmit power.
    \item Two suboptimal resource allocation schemes are proposed for the case when the solution of the dual problem of the SDP relaxed problem is unavailable for construction of the optimal solution.
\end{itemize}
 Our results unveil a non-trivial trade-off between the considered system design objectives which can be summarized as follows: (1) A resource allocation policy  minimizing the total transmit power also leads to a low total interference power leakage in general; (2) energy harvesting efficiency maximization and transmit power minimization  are conflicting system design objectives; (3) the maximum energy harvesting efficiency is achieved at the expense of high interference power leakage and high transmit power.

\section{System Model}
\label{sect:OFDMA_AF_network_model}
In this section,  we first introduce the notation used in this paper. Then, we present the adopted CR downlink channel model for  secure communication with SWIPT.
\subsection{Notation}
 We use boldface capital and lower case letters to denote matrices and vectors, respectively. For a square-matrix $\mathbf{S}$,
$\Tr(\mathbf{S})$ denotes the  trace of matrix
$\mathbf{S}$.  $\mathbf{S}\succ \mathbf{0}$ and $\mathbf{S}\succeq \mathbf{0}$ indicate that
$\mathbf{S}$ is a positive definite and a positive  semidefinite matrix, respectively. $(\mathbf{S})^H$
and $\Rank(\mathbf{S})$ denote the conjugate transpose and the
rank of matrix $\mathbf{S}$, respectively. $\mathbf{I}_{N}$
denotes an $N\times N$ identity matrix.  $\mathbb{C}^{N\times M}$ and $\mathbb{R}^{N\times M}$ denote the space of $N\times M$ matrices with complex and real entries, respectively. $\mathbb{H}^N$ represents the set of all $N$-by-$N$ complex Hermitian matrices.  $\abs{\cdot}$ and $\norm{\cdot}$ denote the absolute value of a complex scalar and the
Euclidean norm of a matrix/vector, respectively. $\diag(x_1, \cdots, x_K)$ denotes a diagonal matrix with the diagonal elements given by $\{x_1, \cdots, x_K\}$.    $\mathrm{Re}(\cdot)$ extracts the real part of a complex-valued input. The distribution of a circularly symmetric complex Gaussian (CSCG)
vector with mean vector $\mathbf{x}$ and covariance matrix
$\mathbf{\Sigma}$  is denoted by ${\cal
CN}(\mathbf{x},\mathbf{\Sigma})$, and $\sim$ means ``distributed
as".   ${\cal E}\{\cdot\}$ represents  statistical expectation. For a real valued  continuous function $f(\cdot)$,
 $\nabla_{\mathbf{X}} f(\mathbf{X})$ denotes the gradient of $f(\cdot)$ with respect to matrix $\mathbf{X}$. $[x]^+$ stands for $\max\{0,x\}$.

\subsection{Downlink Channel Model}
{We consider a CR secondary network for short distance downlink communication.  There are one secondary transmitter equipped with $N_{\mathrm{T}}>1$ antennas,  $K$ secondary receivers, one primary transmitter\footnote{We note that the considered system model can be extended to include multiple primary transmitters at the expense of a more involved notation.}, and $J$ primary receivers. The primary transmitter, primary receivers, and secondary receivers are single-antenna devices that share the same spectrum, cf.  Figure  \ref{fig:system_model}. We assume $N_{\mathrm{T}}>J$ to enable efficient communication in the CR secondary network. }  The secondary transmitter  provides  SWIPT services to the secondary receivers while the primary transmitter provides broadcast services to the primary receivers. In practice, the CR secondary operator may rent spectrum from the primary operator under the condition that the interference leakage from the secondary system to the primary system is properly controlled.   We assume that the secondary receivers are ultra-low power devices, such as wireless sensors\footnote{\label{label:fn}  The power consumption of typical  sensor micro-controllers, such as the Texas Instruments micro-controller: MSP430F2274 \cite{MCU}, is in the order of microwatt in the idle mode.
As a result,  wireless power transfer is a viable option for the energy supply of wireless sensors.}, which either harvest energy or decode information from the received radio signals in each time instant, but are not able to perform both concurrently due to hardware limitations \cite{JR:WIP_receiver,JR:WIPT_fullpaper}. In each scheduling slot, the secondary transmitter not only conveys
information to a given secondary receiver, but also  transfers energy\footnote{We adopt the normalized energy unit  Joule-per-second in this paper.  Therefore,
the terms ``power" and ``energy" are used interchangeably.} to the remaining $K-1$  idle secondary receivers to extend their lifetimes. \begin{figure}[t]
 \centering
\includegraphics[width=3.5in]{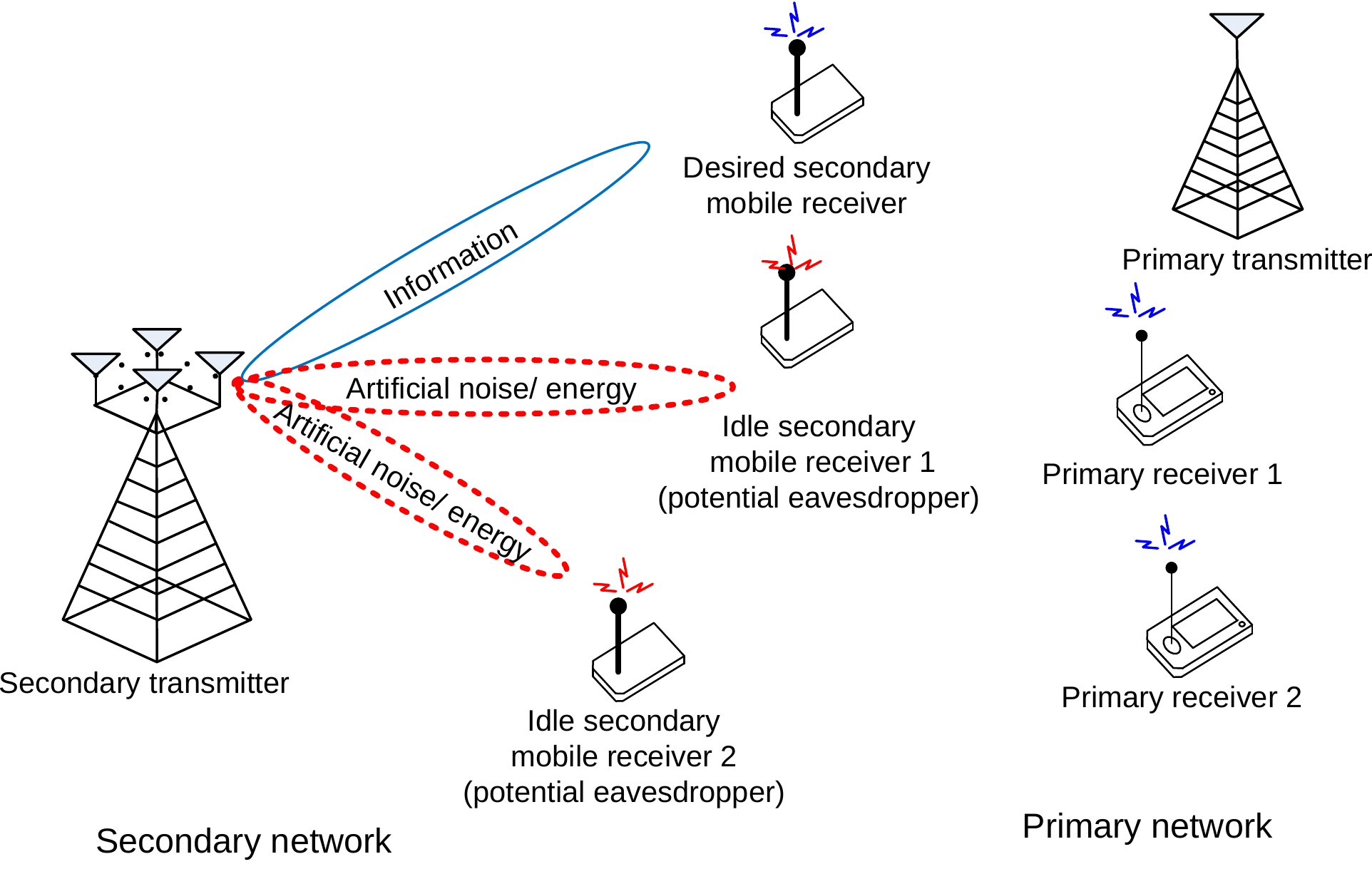}
 \caption{A CR network where $K=3$ secondary receivers (1 active and 2 idle receivers) share the same spectrum with $J=2$ primary receivers. The secondary transmitter conveys  information and  transfers power(/energy) to the $K$ secondary receivers  simultaneously. The red dotted ellipsoids  illustrate the dual use of artificial noise for providing security and facilitating efficient energy transfer to the secondary receivers.} \label{fig:system_model}\vspace*{-4mm}
\end{figure}
We note that only one secondary receiver is selected for information transfer to reduce the multiple access interference leakage to the primary receivers \cite{JR:CR_phy}. On the other hand,
the information signal of the desired secondary receiver is overheard by both the $K-1$  idle secondary receivers and the $J$ primary receivers. Hence, if the  idle secondary receivers and the primary receivers are malicious, they may  eavesdrop the signal of the selected secondary receiver, which has to be taken into account for resource allocation design to provide communication secrecy in the secondary network.  Thus, for guaranteeing communication security,  the secondary transmitter has to employ a resource allocation algorithm that accounts for this unfavourable  scenario and treat both idle secondary receivers and primary receivers as potential eavesdroppers. We assume a frequency flat slow fading channel. The received signals at the desired secondary receiver,  idle secondary receiver $k\in\{1,\ldots,K-1\}$, and   primary receiver $j\in\{1,\ldots,J\}$ are given by, respectively,
\begin{eqnarray}
\hspace*{-4mm}y\hspace*{-3mm}&=&\hspace*{-3mm}\mathbf{h}^{H} \mathbf{x}+q \sqrt{P^{\mathrm{PU}}} d+z, \\
\hspace*{-4mm}y_{k}^\mathrm{Idle}\hspace*{-3mm}&=&\hspace*{-3mm}\mathbf{g}_{k}^{H} \mathbf{x}+ f_k\sqrt{P^{\mathrm{PU}}}d+ z_k ,\,  \forall k\in\{1,\ldots,K-1\},\,\mbox{and}\\
\hspace*{-4mm}y_{j}^{\mathrm{PU}}\hspace*{-3mm}&=&\hspace*{-3mm}\mathbf{l}_{j}^{H} \mathbf{x}+t_j\sqrt{P^{\mathrm{PU}}} d + z_{\mathrm{PU}_j} ,\,\,  \forall j\in\{1,\ldots,J\}.
\end{eqnarray}
Here, $\mathbf{x}\in\mathbb{C}^{ N_T \times 1}$  denotes the symbol vector  transmitted by the secondary transmitter.
$\mathbf{h}^{H}\in\mathbb{C}^{1\times N_T}$, $\mathbf{g}_{k}^{H} \in\mathbb{C}^{1\times N_T}$, and  $\mathbf{l}^{H}_j\in\mathbb{C}^{1\times N_T}$   are the channel
vectors between the  secondary transmitter and the desired secondary receiver, idle receiver (potential eavesdropper) $k$, and primary receiver (potential eavesdropper) $j$, respectively.  $P^{\mathrm{PU}}$ and $d\in\mathbb{C}^{1\times 1}$ are the transmit power of the primary transmitter and the information signal intended for the primary receivers, respectively. $q \in\mathbb{C}^{1\times 1}$, $f_k \in\mathbb{C}^{1\times 1}$, and $t_j \in\mathbb{C}^{1\times 1}$ are the communication channels between the primary transmitter and desired secondary receiver, idle secondary receiver $k$, and primary receiver $j$, respectively.    $z_{\mathrm{PU}_j}$  includes the joint effects of  the thermal noise and the signal processing noise at primary receiver $j$ and is modelled as additive white Gaussian
noise (AWGN) with zero mean and variance\footnote{We assume that the noise characteristics are
identical for all primary receivers due to similar hardware architectures. } $\sigma_{\mathrm{PU}}^2$. $z$ and  $z_k$ include the joint effects  of thermal noise and signal processing
noise at the desired secondary receiver and idle secondary receiver $k$, respectively, and are modeled as AWGN.  Besides, the equivalent noises at the desired and idle secondary receivers, which capture the joint effect of  the received interference from the primary transmitter, i.e., $q \sqrt{P^{\mathrm{PU}}} d$ and $f_k\sqrt{P^{\mathrm{PU}}}d$,   thermal noise, and signal processing noise, are also modeled as AWGN  with zero mean and variances $\sigma_{\mathrm{z}}^2$ and $\sigma_{\mathrm{z}_k}^2$, respectively.
  \begin{Remark}
In this paper, we assume that  the primary
network is a legacy system  and the primary transmitter does not actively participate in
transmit power control. Furthermore, we assume that the primary transmitter transmits a Gaussian signal and we focus on quasi-static fading channels such that all channel gains remain constant within the coherence time of the secondary system.  These assumptions justify modelling the interference from the primary transmitter to the secondary receivers as  additive white Gaussian noise with different powers for different secondary receivers. This model has been commonly adopted in the literature for resource allocation algorithm design \cite{JR:CR_beamforming_1,JR:CR_phy,CN:karama}.
\end{Remark}

To guarantee secure communication and to facilitate an efficient power transfer in the secondary system, artificial noise
is generated at the secondary transmitter and is transmitted concurrently with the information signal. In particular, the transmit signal vector
\begin{eqnarray}
\mathbf{x}=\underbrace{\mathbf{w}s}_{\mbox{desired signal}}+\underbrace{\mathbf{v}}_{\mbox{artificial noise}}
\end{eqnarray}
is adopted at the secondary transmitter, where $s\in\mathbb{C}^{1\times 1}$ and $\mathbf{w}\in\mathbb{C}^{N_{\mathrm{T}}\times 1}$  are the information bearing signal for the desired receiver and the corresponding  beamforming vector, respectively. We assume without loss of generality that  ${\cal E}\{\abs{s}^2\}=1$. $\mathbf{v}\in\mathbb{C}^{N_{\mathrm{T}}\times 1}$ is the artificial noise vector generated by the secondary transmitter to combat the potential eavesdroppers. Specifically, $\mathbf{v}$ is modeled as a complex Gaussian random vector with mean $\mathbf{0}$ and covariance matrix
$\mathbf{V}\in \mathbb{H}^{N_{\mathrm{T}}}, \mathbf{V}\succeq \mathbf{0}$. We note that $\mathbf{w}$ and $\mathbf{V}$ have to be optimized such that the transmit signal of the secondary transmitter does not interfere severely with the primary users.

\section{Resource Allocation Problem Formulation}\label{sect:forumlation}
In this section, we define different quality of service (QoS) measures for  the secondary CR network for providing wireless power transfer and secure communication to the secondary receivers while protecting the primary receivers. Then, we
formulate three resource allocation problems reflecting three different system design objectives. For convenience, we define the following matrices:
$\mathbf{H}=\mathbf{h}\mathbf{h}^H$, $\mathbf{G}_k=\mathbf{g}_k\mathbf{g}_k^H, k\in\{1,\ldots,K-1\}$, and $\mathbf{L}_j=\mathbf{l}_j\mathbf{l}_j^H, j\in\{1,\ldots,J\}$.
\subsection{System Achievable Rate and Secrecy Rate}
\label{subsect:Instaneous_Mutual_information}
%%%%%%%%%%%%%%%%%%%%%%%%%%%%%%%%%%%%%%%%%%%%%%%%%%%
Given perfect CSI at the
receiver, the achievable rate (bit/s/Hz) between the secondary transmitter and the desired secondary receiver
is given by
\begin{eqnarray}\label{eqn:cap}
C=\log_2\Big(1+\Gamma\Big)\,\,\,\,
\mbox{and}\,\,\,\,\Gamma=\frac{\mathbf{w}^H\mathbf{H}\mathbf{w}}
{\Tr(\mathbf{H}\mathbf{V})+\sigma_{\mathrm{z}}^2} ,
\end{eqnarray}
where $\Gamma$ is the received signal-to-interference-plus-noise ratio (SINR) at the desired secondary receiver. On the other hand, the achievable rates between the secondary transmitter and idle secondary receiver $k\in \{1,\ldots,K-1\}$ and primary receiver $j\in \{1,\ldots,J\}$ are given  by
\begin{eqnarray}\label{eqn:cap-eavesdropper}
\hspace*{-3mm}C^{\mathrm{Idle}}_k\hspace*{-3mm}&=&\hspace*{-3mm}\log_2\Big(1+\Gamma^{\mathrm{Idle}}_k\Big),\,
\,\,\,\,\Gamma^{\mathrm{Idle}}_k= \frac{\mathbf{w}^H\mathbf{G}_k\mathbf{w}}{\Tr(\mathbf{G}_k\mathbf{V})+\sigma_{\mathrm{z}_k}^2},\,\mbox{and} \\
\hspace*{-3mm}C^{\mathrm{PU}}_j\hspace*{-3mm}&=&\hspace*{-3mm}\log_2\Big(1+\Gamma^{\mathrm{PU}}_j\Big),\,
\,\,\,\,\Gamma^{\mathrm{PU}}_j= \frac{\mathbf{w}^H\mathbf{L}_j\mathbf{w}}{\Tr(\mathbf{L}_k\mathbf{V})+\sigma_{\mathrm{PU}}^2},   \end{eqnarray} respectively,
where  $\Gamma^{\mathrm{Idle}}_k$ and  $\Gamma^{\mathrm{PU}}_j$ are the received SINRs at idle secondary receiver $k$ and primary receiver $j$, respectively. Since both the idle secondary receivers and the primary receivers are potential eavesdroppers,   the maximum achievable secrecy rate between the secondary transmitter
and the desired receiver is given by
\begin{eqnarray}\label{eqn:secrecy_cap}
C_\mathrm{sec}=\Big[C - \underset{\underset{j\in\{1,\ldots,J\}}{k\in\{1,\ldots,K-1\}}}{\max} \{C^{\mathrm{Idle}}_k,C^{\mathrm{PU}}_j\}\Big]^+.
\end{eqnarray}
In the literature, the secrecy rate, i.e.,  (\ref{eqn:secrecy_cap}),  is commonly adopted as a QoS requirement for system design to ensure secure communication \cite{JR:EE_secrecy,JR:ken_artifical_noise}. In particular, $C_{\mathrm{sec}}$  quantifies
the maximum achievable data rate at which a transmitter can reliably send  secret
information to the intended receiver such that the eavesdroppers are unable to decode the received signal \cite{Report:Wire_tap} even if the  eavesdroppers have unbounded computational capability\footnote{
We note that, in practice,  the  malicious secondary idle  receivers and primary receivers do not have to decode the eavesdropped information in real time. They can act as information collectors to sample the received signals and store them  for future decoding by other energy unlimited and powerful computational devices.}.

\subsection{Energy Harvesting Efficiency}
In the considered CR system, the  secondary receivers harvest energy from the RF when they are idle to extend their lifetimes\footnote{ In fact, nowadays many sensors  are equipped with hybrid energy harvesters for harvesting energy from different energy sources such as solar and thermal-energy \cite{CN:hybrid_energy_source,JR:hybrid_energy_source}. Thus, the harvested energy from the radio frequency may be used
as a supplement for supporting the energy consumption of the secondary receivers.}.  The energy harvesting efficiency plays an important role in the system design of such secondary networks and has to be considered in the
problem formulation. To this end, we define the energy harvesting efficiency in the secondary system as the ratio of the total power harvested  at the idle secondary receivers and the total power  radiated by the secondary transmitter. The total amount of energy harvested by the $K-1$ idle secondary receivers is modeled as
\begin{eqnarray}\label{eqn:harvested_power}
\mathrm{HP}(\mathbf{w},\mathbf{V})=\sum_{k=1}^{K-1}\eta_k\Big(\mathbf{w}^H\mathbf{G}_k\mathbf{w}
+\Tr(\mathbf{G}_k\mathbf{V})\Big),
\end{eqnarray}
  where $\eta_k$ is a constant, $0\le\eta_k\le 1,\forall k$, which represents the RF energy conversion efficiency of idle secondary receiver $k$ in converting
the received radio signal to electrical energy. { We note that the power received at the secondary receivers  from the primary transmitter and the AWGN power are neglected in (\ref{eqn:harvested_power}) as we focus on the worst-case scenario for  robust energy harvesting system design.}

On the other hand, the power radiated by the transmitter can be expressed
as
\begin{eqnarray}
 \label{eqn:power_consumption}\mathrm{TP}(\mathbf{w},\mathbf{V})=\norm{\mathbf{w}}^2+\Tr(\mathbf{V}).
\end{eqnarray}
Thus, the energy harvesting efficiency of the considered secondary CR system is
given by
\begin{eqnarray}\label{eqn:ehe}
\eta_{\mathrm{eff}}(\mathbf{w},\mathbf{V})=\frac{\mathrm{HP}(\mathbf{w},\mathbf{V})}{\mathrm{TP}(\mathbf{w},\mathbf{V})}.
\end{eqnarray}

\subsection{Interference Power Leakage-to-Transmit Power Ratio}
 In the considered CR network, the secondary receivers and the primary receivers share the same spectrum resource. However, the primary receivers are licensed users and thus the secondary transmitter is required to ensure the QoS of the primary receivers via a careful resource allocation design.  Strong interference may impair the primary network when the secondary transmitter increases its transmit power for providing SWIPT services to the secondary receivers.   As a result,  the interference power leakage-to-transmit power ratio (IPTR) is an important performance measure for designing the secondary CR network and should be captured in the resource allocation algorithm design. To this end, we first define the total interference power received by the $J$ primary receivers as
\begin{eqnarray}
\mathrm{IP}(\mathbf{w},\mathbf{V})=\sum_{j=1}^{J}\Big(\mathbf{w}^H\mathbf{L}_j\mathbf{w}
+\Tr(\mathbf{L}_j\mathbf{V})\Big).
\end{eqnarray}
Thus, the IPTR of the considered secondary CR network is
defined as
\begin{eqnarray}\label{eqn:Interference Power Leakage-to-Transmit Power Ratio}
\mathrm{IP}_{\mathrm{ratio}}(\mathbf{w},\mathbf{V})=\frac{\mathrm{IP}(\mathbf{w},\mathbf{V})}{\mathrm{TP}(\mathbf{w},\mathbf{V})}.
\end{eqnarray}

\subsection{Channel State Information}
%%%%%%%%%%%%%%%%%%%%%%%%%%%%%%%%%%%%%%%%%%%%%%%%%%%%%
 In this paper, we focus on a Time Division
Duplex (TDD)  communication system with slowly time-varying channels.  In practice,  handshaking\footnote{The legitimate receivers can either take turns to send the handshaking signals or  transmit simultaneously with orthogonal pilot sequences. } is performed between the secondary transmitter and the secondary receivers at the beginning of each scheduling slot. This allows the secondary transmitter to  obtain the statuses and  the QoS  requirements  of the secondary receivers. As a result, by exploiting the channel reciprocity, the downlink CSI of the secondary transmitter to the secondary receivers can be obtained by measuring the uplink training sequences embedded in the handshaking signals. Thus, we assume that the secondary-transmitter-to-secondary-receiver fading gains, $\mathbf{h}$ and $\mathbf{g}_{k},\forall k\in\{1,\ldots,K-1\}$, can be reliably estimated at the secondary transmitter at the beginning of each scheduling slot with negligible estimation error.  Then, during the transmission, the desired  secondary receiver is required to send positive acknowledgement (ACK) packets to inform the secondary transmitter of successful reception of the information packets. Hence, the transmitter is able to update  the CSI estimate of the desired receiver  frequently via the training sequences in each ACK packet. Therefore,  perfect CSI for the secondary-transmitter-to-desired-secondary-receiver  link, i.e.,  $\mathbf{h}$, is assumed  over the entire transmission period. However,   the remaining $K-1$ secondary receivers are idle and there is no interaction between them and the secondary transmitter after handshaking. As a result, the CSI of the idle secondary receivers becomes outdated during transmission.  To capture the impact of the CSI imperfection and to isolate specific channel estimation methods from  the resource allocation algorithm design,  we adopt a deterministic model \cite{JR:Robust_error_models1}--\nocite{JR:Robust_error_models2,JR:CSI-determinisitic-model}\cite{JR:CSI-determinisitic-model2} for  the resulting CSI uncertainty.  In particular, the  CSI of the link between the secondary transmitter
and idle secondary receiver $k$ is modeled as
\begin{eqnarray}\label{eqn:outdated_CSI}
\mathbf{g}_k&=&\mathbf{\hat g}_k + \Delta\mathbf{g}_k,\,   \forall k\in\{1,\ldots,K-1\}, \mbox{   and}\\
{\Omega }_k&\triangleq& \Big\{\Delta\mathbf{g}_k\in \mathbb{C}^{N_{\mathrm{T}}\times 1}  :\Delta\mathbf{g}_k^H \Delta\mathbf{g}_k \le \varepsilon_k^2\Big\},\forall k,\label{eqn:outdated_CSI-set}
\end{eqnarray}
where $\mathbf{\hat g}_k\in\mathbb{C}^{N_{\mathrm{T}}\times 1}$ is the CSI estimate available at the secondary transmitter at the beginning of a scheduling slot and $ \Delta\mathbf{g}_k$ represents the unknown channel uncertainty due to the time varying nature of the channel during transmission. The continuous set ${\Omega }_k$ in (\ref{eqn:outdated_CSI-set})  defines a Euclidean sphere and contains  all possible channel uncertainties. Specifically, the radius $\varepsilon_k$ represents the size of the sphere and defines the uncertainty region of the CSI of idle secondary receiver (potential eavesdropper) $k$. In practice, the value of $\varepsilon_k^2$ depends on the coherence time of the associated channel and the duration of transmission.

 Furthermore, to capture  the imperfectness of the CSI of the primary receiver channels at the secondary transmitter, we adopt the same CSI error model as for the idle secondary receivers.  In fact,   the primary receivers are not directly interacting with the secondary transmitter. Besides, the primary receivers may be silent for non-negligible periods of time due to bursty data communication. As a result,  the CSI of the primary receivers  can  be obtained only occasionally at the secondary transmitter when the primary receivers communicate with a primary transmitter.
 Hence, we model the CSI of the link between the secondary transmitter
and primary receiver $j$ as
 \begin{eqnarray}\label{eqn:outdated_CSI-primary}
\mathbf{l}_j&=&\mathbf{\hat l}_j + \Delta\mathbf{l}_j,\,   \forall j\in\{1,\ldots,J\}, \mbox{   and}\\
{\Psi}_j&\triangleq& \Big\{\Delta\mathbf{l}_j\in \mathbb{C}^{N_{\mathrm{T}}\times 1}  :\Delta\mathbf{l}_j^H \Delta\mathbf{l}_j \le \upsilon_j^2\Big\},\forall j,\label{eqn:outdated_CSI-set-primary}
\end{eqnarray}
where $\mathbf{\hat l}_j$ is the  estimate of the channel of primary receiver $j$  at the secondary transmitter and $ \Delta\mathbf{l}_j$ denotes the associated  channel uncertainty. ${\Psi}_j$  and  $\upsilon_j^2$  in (\ref{eqn:outdated_CSI-set-primary})  define the continuous set  of  all possible channel uncertainties and  the size of the uncertainty region of the estimated CSI of primary receiver $j$, respectively. We note that, in practice, the channel estimation qualities for primary receivers and secondary receivers  at the secondary transmitter may be different which leads to different values for $\varepsilon_k$ and $\upsilon_j$.
%%%%%%%%%%%%%%%%%%%%%%%%%%%%%%%%%%%%%%%%
\subsection{Optimization Problem Formulations}
\label{sect:cross-Layer_formulation}
%%%%%%%%%%%%%%%%%%%%%%%%%%%%%%%%%%%%%%%%%%%%%%%%%%%%%%%
We first propose three problem formulations for single-objective system design for  secure communication in the secondary CR network. In particular, each single-objective problem formulation considers  one aspect of the system design.   Then, we consider the three system design objectives jointly under the framework of multi-objective optimization.  In particular, the adopted multi-objective optimization enables the design of  a set of Pareto optimal resource allocation policies.
  The first problem formulation aims at maximizing the energy harvesting efficiency while providing  secure communication in the secondary CR network. The problem formulation is as follows:

\begin{Prob}{Energy Harvesting Efficiency Maximization:}\end{Prob}\vspace*{-1mm}
\begin{eqnarray}
\label{eqn:cross-layer}&&\hspace*{-0mm} \underset{\mathbf{V}\in \mathbb{H}^{N_{\mathrm{T}}},\mathbf{w}}{\maxo}\,\, \,\, \min_{\Delta\mathbf{g}_k\in {\Omega}_k}\,\, \eta_{\mathrm{eff}}(\mathbf{w},\mathbf{V})\nonumber\\
\notag \mbox{s.t.} &&\hspace*{-6mm}\mbox{C1: }\notag\frac{\mathbf{w}^H\mathbf{H}\mathbf{w}}{\Tr(\mathbf{H}
\mathbf{V})+\sigma_\mathrm{z}^2} \ge \Gamma_{\mathrm{req}}, \\
&&\hspace*{-6mm}\mbox{C2: }\hspace*{-2mm}\max_{\Delta\mathbf{g}_k\in {\Omega}_k}\frac{\mathbf{w}^H\mathbf{G}_k\mathbf{w}}
{\Tr(\mathbf{G}_k\mathbf{V})\hspace*{-0.5mm}+\hspace*{-0.5mm}\sigma_{\mathrm{z}_k}^2}\hspace*{-0.5mm} \le \hspace*{-0.5mm} \Gamma_{\mathrm{tol}_k},\forall k\hspace*{-0.5mm}\in\hspace*{-0.5mm}\{1,\ldots,K-1\},\notag\\
&&\hspace*{-6mm}\mbox{C3: }\max_{\Delta\mathbf{l}_j\in {\Psi}_j}\,\,\frac{\mathbf{w}^H\mathbf{L}_j\mathbf{w}}
{\Tr(\mathbf{L}_j\mathbf{V})+\sigma_{\mathrm{PU}}^2} \hspace*{-0.5mm}\le\hspace*{-0.5mm} \Gamma_{\mathrm{tol}_j}^{\mathrm{PU}},\forall j\in\{1,\ldots,J\},\notag\\
&&\hspace*{-6mm}\mbox{C4: } \norm{\mathbf{w}}^2  +\Tr(\mathbf{V})\le P_{\max}, \quad\quad \mbox{C5:}\,\, \mathbf{V}\succeq \mathbf{0}.
\end{eqnarray}
The system objective in (\ref{eqn:cross-layer}) is to maximize the worst case energy harvesting efficiency of the system for  channel estimation errors $\Delta \mathbf{g}_k$ belonging to set $\Omega_k$. Constant $\Gamma_{\mathrm{req}}$ in C1 specifies the minimum required  received SINR of the desired secondary receiver for information decoding.    $\Gamma_{\mathrm{tol}_k},\forall k\in\{1,\ldots,K-1\}$, and $\Gamma_{\mathrm{tol}_j}^{\mathrm{PU}},\forall j\in\{1,\ldots,J\}$, in C2 and C3, respectively,  are given system parameters which denote the maximum tolerable received SINRs at the potential eavesdroppers in the secondary network and the primary network, respectively. In practice, depending on the considered application,  the system operator chooses the  values of  $\Gamma_{\mathrm{req}}$,  $\Gamma_{\mathrm{tol}_k},\forall k\in\{1,\ldots,K-1\}$, and $\Gamma_{\mathrm{tol}_j}^{\mathrm{PU}},\forall j\in\{1,\ldots,J\}$,  such that   $\Gamma_{\mathrm{req}}\gg \Gamma_{\mathrm{tol}_k}>0$ and $\Gamma_{\mathrm{req}}\gg \Gamma_{\mathrm{tol}_j}^{\mathrm{PU}}>0$. In other words, the secrecy rate of the system is bounded below by $C_\mathrm{sec}\ge \log_2(1+\Gamma_{\mathrm{req}})-\log_2(1+\underset{k,j}{\max}\{\Gamma_{\mathrm{tol}_k},\Gamma_{\mathrm{tol}_j}^{\mathrm{PU}}\})> 0$. We note that although $\Gamma_{\mathrm{req}}$,   $\Gamma_{\mathrm{tol}_k}$, and $\Gamma_{\mathrm{tol}_j}^{\mathrm{PU}}$  in C1, C2, and C3, respectively,  are
not optimization variables in this paper, a balance between
secrecy rate and system achievable rate can be struck
by varying their values.  $P_{\max}$ in C4 specifies the maximum  transmit power in the power amplifier of the analog front-end of the secondary transmitter. C5 and $\mathbf{V}\in \mathbb{H}^{N_\mathrm{T}}$  are imposed since covariance matrix $\mathbf{V}$ has to be a  positive semidefinite Hermitian matrix.

To facilitate the presentation and without loss of generality, we rewrite Problem 1 in (\ref{eqn:cross-layer}) in the equivalent form \cite{book:convex}:
\begin{eqnarray}\label{eqn:cross-layer1-flip}
&&\hspace*{-20mm} \underset{\mathbf{V}\in \mathbb{H}^{N_{\mathrm{T}}},\mathbf{w}}{\mino}\,\,\,\, \max_{\Delta\mathbf{g}_k\in {\Omega}_k}\,\, -\eta_{\mathrm{eff}}(\mathbf{w},\mathbf{V})\nonumber\\
\hspace*{6mm}\mbox{s.t.} &&\hspace*{-3mm}\mbox{C1 -- C5}.
\end{eqnarray}

 The second system design objective is the minimization of the total transmit power of the secondary transmitter and   can be
mathematically formulated as:

\begin{Prob}{Total Transmit Power Minimization:} \label{Prob:min_tx}\end{Prob}\vspace*{-1mm}
\begin{eqnarray}\label{eqn:cross-layer2}
&&\hspace*{-20mm} \underset{\mathbf{V}\in \mathbb{H}^{N_{\mathrm{T}}},\mathbf{w}
}{\mino}\,\,\,\, \mathrm{TP}(\mathbf{w},\mathbf{V})\nonumber\\
\hspace*{3mm}\mbox{s.t.} &&\hspace*{-3mm}\mbox{C1 -- C5}.
\end{eqnarray}
Problem \ref{Prob:min_tx} yields the minimum total transmit power of the secondary transmitter while ensuring that the QoS requirement on secure communication is satisfied. We note that Problem \ref{Prob:min_tx} does not take into account the energy harvesting capability  of the idle secondary receivers and focuses only on the requirement of secure communication  via constraints C1, C2, and C3. Besides, although transmit power minimization has been studied in the literature in different contexts \cite{CN:Kwan_PIMRC2013,JR:power_minimziation_beamforming,JR:downlink_beamforming_CR}, combining Problem 2 with the new Problems 1 and 3 (see below) offers new insights for the design of  CR networks providing secure wireless information and power transfer to secondary receivers.

The third system design objective concerns the  minimization  of the worst case IPTR while providing secure communication in the secondary CR network. The problem formulation is given as:
\begin{Prob}{Interference Power Leakage-to-Transmit Power Ratio Minimization:} \label{Prob:min_leak}\end{Prob}\vspace*{-1mm}
\begin{eqnarray}\label{eqn:cross-layer3}
&&\hspace*{-20mm}\underset{\mathbf{V}\in \mathbb{H}^{N_{\mathrm{T}}},\mathbf{w}
}{\mino}\,\,\,\,\max_{\Delta\mathbf{l}_j\in {\Psi}_j}\,\,\notag \mathrm{IP}_{\mathrm{ratio}}(\mathbf{w},\mathbf{V})\nonumber\\
\hspace*{3mm}\mbox{s.t.} &&\hspace*{-3mm}\mbox{C1 -- C5}.
\end{eqnarray}

\begin{Remark}\label{remark1}
In   (\ref{eqn:cross-layer}) and (\ref{eqn:cross-layer3}), the maximization of the energy harvesting efficiency and the minimization of the IPTR are chosen as design objectives, respectively.  Alternative design objectives are the maximization of the total harvested power,  $\underset{\mathbf{V}\in \mathbb{H}^{N_{\mathrm{T}}},\mathbf{w}}{\maxo}\,\,  \underset{{\Delta\mathbf{g}_k\in {\Omega}_k}}{\min} \mathrm{HP}(\mathbf{w},\mathbf{V})$,  and the minimization of the total interference power leakage, $\underset{\mathbf{V}\in \mathbb{H}^{N_{\mathrm{T}}},\mathbf{w}
}{\mino}\,\,\,\,\underset{\Delta\mathbf{l}_j\in {\Psi}_j}{\max}\mathrm{IP}(\mathbf{w},\mathbf{V})$. We will show later that  the maximization  of the energy harvesting efficiency in (\ref{eqn:cross-layer}) and the minimization of the IPTR  in (\ref{eqn:cross-layer3}) subsume the total harvested power maximization  and the total interference power leakage minimization  as special cases, respectively.
Please refer to Remark \ref{eqn:from_fracitonal_to_non_fractional} for the solution of the total interference power leakage minimization and total harvested power maximization problems.
\end{Remark}

\begin{Remark}\label{remark12}
In fact, the optimization problem in   (\ref{eqn:cross-layer3})  can be extended to the  minimization of  the maximum received interference leakage per primary receiver. However, such problem formulation does not facilitate the study of the trade-off between interference leakage, energy harvesting, and total transmit power as the system performance is always limited by those primary users which have strong channels with respect to the secondary transmitter.
\end{Remark}

  In practice, the system design objectives in Problems $1$--$3$ are all desirable for the system operators of secondary CR networks in providing simultaneous power and secure information transfer.  Yet,  theses objectives  are usually conflicting with each other and each objective  focuses on only one aspect of the system. In the literature, multi-objective optimization has been proposed for studying the trade-off between conflicting system design objectives via the concept of Pareto optimality.   For facilitating the following exposition, we denote  the objective function and the optimal objective value for problem formulation $p\in\{1,2,3\}$ as $F_{p}(\mathbf{w},\mathbf{V})$ and $F_p^*$, respectively.  We define a resource allocation policy which is Pareto optimal as:

\emph{Definition  \cite{JR:MOOP}:}  A resource allocation policy, $\{\mathbf{w,V}\}$, is Pareto optimal if and only if there does not
exist another policy,  $\{\mathbf{w',V'}\}$,  such that $ F_i(\mathbf{w}',\mathbf{V}')\le F_i(\mathbf{w},\mathbf{V}), \forall i\in\{1,2,3\}$, and  $ F_j(\mathbf{w}',\mathbf{V}')< F_j(\mathbf{w},\mathbf{V})$ for at least one index $j\in\{1,2,3\}$.

The set of all Pareto optimal resource allocation polices is called the Pareto frontier or the Pareto optimal set.  In this paper, we adopt the
weighted Tchebycheff method \cite{JR:MOOP} for investigating the trade-off between objective functions 1, 2, and 3.  In particular, the weighted Tchebycheff method can provide the complete Pareto optimal set despite the non-convexity (if any) of the considered problems\footnote{In the literature, different scalarization methods have been proposed  for achieving the  points of the complete Pareto set for multi-objective optimization   \cite{JR:MOOP,book:MOOP2}. However, the weighted Tchebycheff method requires a lower computational complexity compared to other methods such as the weighted product method and the exponentially weighted criterion.}; it provides a necessary condition for Pareto optimality. The complete Pareto optimal set can be achieved by solving the following multi-objective problem:

\setcounter{Prob}{3}
\begin{Prob}{Multi-Objective Optimization -- Weighted Tchebycheff Method:} \label{Prob:multi}\end{Prob}\vspace*{-5mm}
\begin{eqnarray}
\label{eqn:cross-layer4}&&\hspace*{-25mm}\underset{\mathbf{V}\in \mathbb{H}^{N_{\mathrm{T}}},\mathbf{w}}{\mino}\,\,\max_{p\in\{\,1,\,2,\,3\}}\,\, \Bigg\{\lambda_p \Big(\frac{F_p\big(\mathbf{w},\mathbf{V}\big)-F_p^*}{\abs{F_p^*}}\Big)\Bigg\}\nonumber\\
 \mbox{s.t.} &&\hspace*{3mm}\mbox{C1 -- C5},
\end{eqnarray}
In fact,  by varying the values of $\lambda_p$, Problem 4 yields the
complete Pareto optimal  set \cite{JR:MOOP,book:MOOP2}. Besides, Problem \ref{Prob:multi} is a generalization of Problems 1, 2, and 3. In particular,  Problem \ref{Prob:multi} is equivalent\footnote{Here, ``equivalent" means that the considered problems share the same optimal resource allocation solution(s).  } to Problem $p$ when $\lambda_p=1$ and $\lambda_i=0, \forall i\ne p$.
{For instance, if the secondary energy harvesting receivers do not require wireless power transfer from the secondary transmitter, without loss of generality, we can set $\lambda_1=0$ in Problem 4 to study the tradeoff between the remaining two system design objectives.}
 In addition, this commonly adopted approach also provides a non-dimensional objective
function, i.e., the unit of the objective function is normalized.

\begin{Remark} Finding the Pareto optimal set of the multi-objective optimization problem  provides a set of Pareto optimal resource allocation policies. Then, depending  on the preference of the system operator, a proper resource allocation policy can be selected from the set for implementation. We note that the resource allocation algorithm in  \cite{JR:Kwan_secure_imperfect} cannot be directly applied to the problems considered in this paper since  it was designed for single-objective optimization, namely for the  for minimization of the total transmit power.
\end{Remark}

\begin{Remark} { Another possible problem formulation for the considered system model is to move some of the objective functions in (\ref{eqn:cross-layer}), (\ref{eqn:cross-layer2}), and (\ref{eqn:cross-layer3}) to the set of  constraints and constrain each of them by some constant. Then, by varying the constants, trade-offs between different objectives can be struck.    However, in general,  such a problem formulation does not reveal the Pareto optimal set due to the non-convexity of the problem.}
\end{Remark}

%%%%%%%%%%%%%%%%%%%%%%%%%%%%%%%%%%%%%%%%%%%%%%%%%%%%%%% %%%%%%%%%%%%%%%
\section{Solution of the Optimization Problems} \label{sect:solution}
The optimization problems in (\ref{eqn:cross-layer1-flip}), (\ref{eqn:cross-layer2}), and (\ref{eqn:cross-layer3}) are non-convex with respect to the optimization variables.  In particular, the non-convexity arises from objective function 1, objective function 3, and constraint C1. In order to obtain  tractable solutions for the problems, we recast Problems 1, 2, 3, and 4 as
convex optimization problems by semidefinite programming (SDP) relaxation  \cite{CN:SDP_relaxation1,JR:SDP_relaxation1}  and study the tightness of the adopted relaxation in this section.
%%%%%%%%%%%%%%%%%%%%%%%%%%%%%%%%%%%%%%%%%%%%%%%%%%%%%%%%%%%%%%%%%%%%%%%%%%%%%%%
\subsection{Semidefinite Programming Relaxation} \label{sect:solution_dual_decomposition}
%%%%%%%%%%%%%%%%%%%%%%%%%%%%%%%%%%%%%%%%%%%%%%%%%%%%%%%%%%%%%%%%%%%%%%%%%%%%%%%
To facilitate the SDP relaxation, we define
\begin{eqnarray}\label{eqn:change_of_variables}
\mathbf{W}=\mathbf{w}\mathbf{w}^H,\, \mathbf{W}=\frac{\overline{\mathbf{W}}}{\xi},  \mathbf{V}=\frac{\overline{\mathbf{V}}}{\xi},\, \xi=\frac{1}{\Tr({\mathbf{W}})+\Tr({\mathbf{V}})},
\end{eqnarray}
 and rewrite Problems 1 -- 4  in terms of new optimization variables $\overline{\mathbf{W}}$,  $\overline{\mathbf{V}}$, and $\xi$.

\begin{T-Prob}{Energy Harvesting Efficiency Maximization:}\end{T-Prob}\vspace*{-1mm}
\begin{eqnarray}
\label{eqn:cross-layer-t1}&&\hspace*{-10mm}\underset{\overline{\mathbf{V}},\overline{\mathbf{W}}\in \mathbb{H}^{N_{\mathrm{T}}},\xi
}{\mino}\,\, \max_{\Delta\mathbf{g}_k\in {\Omega}_k}-\sum_{k=1}^{K-1}\eta_k\Tr(\mathbf{G}_k(\overline{\mathbf{W}}+\overline{\mathbf{V}}))\nonumber\\
\notag \mbox{s.t.} &&\hspace*{-5mm}\mbox{\textoverline{C1}:} \notag\frac{\Tr(\mathbf{H}\overline{\mathbf{W}})}{\Tr(\mathbf{H}
\overline{\mathbf{V}})+\sigma_\mathrm{z}^2\xi} \ge \Gamma_{\mathrm{req}}, \\
&&\hspace*{-5mm}\mbox{\textoverline{C2}: }\notag\max_{\Delta\mathbf{g}_k\in {\Omega}_k}\frac{\Tr(\mathbf{G}_k\overline{\mathbf{W}})}
{\Tr(\mathbf{G}_k\overline{\mathbf{V}})+\sigma_{\mathrm{z}_k}^2\xi} \le \Gamma_{\mathrm{tol}_k},\forall k,\\
&&\hspace*{-5mm}\mbox{\textoverline{C3}: }\notag\max_{\Delta\mathbf{l}_j\in {\Psi}_j}\frac{\Tr(\mathbf{L}_j\overline{\mathbf{W}})}
{\Tr(\mathbf{L}_j\overline{\mathbf{V}})+\sigma_{\mathrm{PU}}^2\xi} \le \Gamma_{\mathrm{tol}_j}^{\mathrm{PU}},\forall j,\\
&&\hspace*{-5mm}\mbox{\textoverline{C4}: }\notag \Tr(\overline{\mathbf{W}})  +\Tr(\overline{\mathbf{V}})\le P_{\max}\xi, \\
%&&\hspace*{-5mm}\mbox{C6: }\notag \norm{\mathbf{w}}^2\varepsilon  +\Tr(\mathbf{V})\varepsilon+P_C\le P_{PG}, \\
&&\hspace*{-5mm}\mbox{\textoverline{C5}:}\,\, \mathbf{\overline{W}},\mathbf{\overline{V}}\succeq \mathbf{0},\hspace*{5mm}\mbox{\textoverline{C6}:}\,\, \xi \ge 0, \notag\\
&&\hspace*{-5mm}\mbox{\textoverline{C7}:}\,\, \Tr(\overline{\mathbf{W}})+\Tr(\overline{\mathbf{V}})= 1,\,\,\mbox{\textoverline{C8}:}\,\, \Rank(\mathbf{\overline{W}})=1,
\end{eqnarray}
where  $\overline{\mathbf{W}}\succeq \mathbf{0}$, $\overline{\mathbf{W}}\in \mathbb{H}^{N_{\mathrm{T}}}$, and $\Rank(\overline{\mathbf{W}})=1$ in (\ref{eqn:cross-layer-t1}) are imposed to guarantee that $\overline{\mathbf{W}}=\xi\mathbf{w}\mathbf{w}^H$ after optimizing $\mathbf{\overline W}$.

\begin{T-Prob}{Total Transmit Power Minimization:}\end{T-Prob}\vspace*{-1mm}
\begin{eqnarray}
\label{eqn:cross-layer-t2} \underset{\overline{\mathbf{V}},\mathbf{\overline{W}}\in \mathbb{H}^{N_{\mathrm{T}}}, \xi
}{\mino}&&\,\, \frac{1}{\xi}\nonumber\\
&&\hspace*{-2cm}\mbox{s.t. }\, \mbox{\textoverline{C1} -- \textoverline{C8}}.
\end{eqnarray}

\begin{T-Prob}{Interference Power Leakage-to-Transmit Power Ratio Minimization:}\end{T-Prob}\vspace*{-1mm}
\begin{eqnarray}
\label{eqn:cross-layer-t3} \underset{\overline{\mathbf{V}},\mathbf{\overline{W}}\in \mathbb{H}^{N_{\mathrm{T}}}, \xi
}{\mino}&&\,\, \hspace*{-0.5cm}\max_{\Delta\mathbf{l}_j\in {\Psi}_j}\sum_{j=1}^{J}\Tr(\mathbf{L}_j(\overline{\mathbf{W}}+\overline{\mathbf{V}}))\nonumber\\
 &&\hspace*{-0cm}\mbox{s.t. }\, \mbox{\textoverline{C1} -- \textoverline{C8}}.
\end{eqnarray}

\newpage
\begin{T-Prob}{Multi-Objective Optimization:}\end{T-Prob}\vspace*{-1mm}
\begin{eqnarray}
\label{eqn:cross-layer-t4}&&\hspace*{1cm}\underset{\overline{\mathbf{V}},\mathbf{\overline{W}}\in \mathbb{H}^{N_{\mathrm{T}}}, \xi,\tau
}{\mino}\,\,\,\tau\nonumber\\
 &&\hspace*{0.8cm}\notag \mbox{s.t. }\,
\mbox{\textoverline{C1} -- \textoverline{C8}},\\
&&\hspace*{-1.2cm}\mbox{\textoverline{C9}a: }\frac{\lambda_1}{\abs{F_1^*}} (\overline{F_1}-F_1^*)\le \tau,\,\,\,\mbox{\textoverline{C9}b: }\frac{\lambda_2}{\abs{F_2^*}} (\overline{F_2}-F_2^*)\le \tau,\notag\\
&&\hspace*{-1.2cm}\mbox{\textoverline{C9}c: }\frac{\lambda_3}{\abs{F_3^*}} (\overline{F_3}-F_3^*)\le \tau,
\end{eqnarray}
where $\overline{F_1}=\underset{\Delta\mathbf{g}_k\in {\Omega}_k}\min-\sum_{k=1}^{K-1}\varepsilon_k\Tr(\mathbf{G}_k(\overline{\mathbf{W}}+\overline{\mathbf{V}}))$, $\overline{F_2}= \frac{1}{\xi}$, $\overline{F_3}=  \underset{\Delta\mathbf{l}_j\in {\Psi}_j}\max\sum_{j=1}^{J}\Tr(\mathbf{L}_j(\overline{\mathbf{W}}+\overline{\mathbf{V}}))$,   $\tau$ is an auxiliary optimization variable, and (\ref{eqn:cross-layer-t4}) is the epigraph representation \cite{book:convex} of (\ref{eqn:cross-layer4}).
\begin{proposition}
The above transformed Problems (\ref{eqn:cross-layer-t1})--(\ref{eqn:cross-layer-t4}) are equivalent to the original problems in  (\ref{eqn:cross-layer1-flip})--(\ref{eqn:cross-layer4}), respectively. Specifically, we can recover the solutions of the original problems from the solutions of the transformed problems based on (\ref{eqn:change_of_variables}).
\end{proposition}

\begin{proof}
Please refer to Appendix A.
\end{proof}
Since  transformed Problem 4 is a generalization of transformed Problems 1, 2, and 3, we focus on the  methodology for solving   transformed Problem\footnote{In studying the solution structure of transformed Problem 4, we assume that the optimal objective values of transformed Problems 1--3 are given constants, i.e.,  $F^*_p,\forall p\in\{1,2,3\}$, are known. Once the structure of the  optimal resource allocation scheme  of transformed Problem 4 is obtained, it can be exploited to obtain the optimal solution of transformed Problems 1--3. } 4. In practice,  the considered problems may be infeasible when  the channels are in unfavourable  conditions and/or the QoS requirements are too stringent. However,  in the sequel, for studying the trade-off between different system design objectives and the design of different resource allocation schemes,  we assume that the problem is always feasible\footnote{We note that multiple optimal solutions may exist for the  considered problems and the proposed optimal  resource allocation scheme is able to find at least one of the global optimal solutions.}.

First, we address constraints $\mbox{\textoverline{C2}}$,  $\mbox{\textoverline{C3}}$, and $\mbox{\textoverline{C9}}$. We note that although these constraints  are convex with respect to the optimization variables, they are semi-infinite constraints  which are generally intractable. For facilitating the design of a tractable resource allocation algorithm, we introduce two auxiliary optimization variables $E_{k}^{\mathrm{SU}}$   and $I_{j}^{\mathrm{PU}}$ and rewrite transformed Problem 4 in (\ref{eqn:cross-layer-t4}) as
\begin{eqnarray}
\label{eqn:cross-layer-t31}&&\underset{\mathbf{\overline{W}},\overline{\mathbf{V}}\in \mathbb{H}^{N_{\mathrm{T}}}
,I_{j}^{\mathrm{PU}},E_{k}^{\mathrm{SU}}, \xi,\tau,}{\mino}\,\,\,\tau\nonumber\\
 &&\notag \mbox{s.t. }\,
\mbox{\textoverline{C1} -- \textoverline{C8}},\notag\\
&&\hspace*{-0.6cm}\mbox{\textoverline{C9}a: }\lambda_1 \Big(\sum_{k=1}^{K-1}E_{k}^{\mathrm{SU}}-F_1^*\Big)\le \tau\abs{F_1^*},\notag\\
&&\hspace*{-0.6cm}\mbox{\textoverline{C9}b: }\lambda_2(\overline{F_2}-F_2^*)\le \tau\abs{F_2^*}, \notag\\
&&\hspace*{-0.6cm}\mbox{\textoverline{C9}c: }\lambda_3 \Big(\sum_{j=1}^{J}I_{j}^{\mathrm{PU}}-F_3^*\Big)\le \tau\abs{F_3^*}\notag,\\
&&\hspace*{-0.6cm}\mbox{\textoverline{C10}: }E_{k}^{\mathrm{SU}}\ge \max_{\Delta\mathbf{g}_k\in {\Omega}_k}- \eta_k\Tr(\mathbf{G}_k(\overline{\mathbf{W}}+\overline{\mathbf{V}})),\forall k,\notag\\
&&\hspace*{-0.6cm}\mbox{\textoverline{C11}: }I_{j}^{\mathrm{PU}}\ge \max_{\Delta\mathbf{l}_j\in {\Psi}_j} \Tr(\mathbf{L}_j(\overline{\mathbf{W}}+\overline{\mathbf{V}})),\forall j.
\end{eqnarray}
In fact, the introduced auxiliary variables $E_{k}^{\mathrm{SU}}$ and $I_{j}^{\mathrm{PU}}$   decouple the original two nested semi-infinite constraints into two semi-infinite constraints and two  affine constraints, i.e., $\mbox{\textoverline{C10}}$,
$\mbox{\textoverline{C11}}$ and $\mbox{\textoverline{C9}a}$,  $\mbox{\textoverline{C9}c}$, respectively. It can be  verified that (\ref{eqn:cross-layer-t31}) is equivalent to  (\ref{eqn:cross-layer-t3}), i.e., constraints $\mbox{\textoverline{C10}}$ and $\mbox{\textoverline{C11}}$ are satisfied with equality for the optimal solution. Next, we transform constraints $\mbox{\textoverline{C2}}$, $\mbox{\textoverline{C3}}$, $\mbox{\textoverline{C10}}$, and $\mbox{\textoverline{C11}}$ into linear matrix inequalities (LMIs) using the following lemma:
\begin{Lem}[S-Procedure \cite{book:convex}] Let a function $f_m(\mathbf{x}),m\in\{1,2\},\mathbf{x}\in \mathbb{C}^{N\times 1},$ be defined as
\begin{eqnarray}\label{eqn:S-procedure}
f_m(\mathbf{x})=\mathbf{x}^H\mathbf{A}_m\mathbf{x}+2 \mathrm{Re} \{\mathbf{b}_m^H\mathbf{x}\}+c_m,
\end{eqnarray}
where $\mathbf{A}_m\in\mathbb{H}^N$, $\mathbf{b}_m\in\mathbb{C}^{N\times 1}$, and $c_m\in\mathbb{R}$. Then, the implication $f_1(\mathbf{x})\le 0\Rightarrow f_2(\mathbf{x})\le 0$  holds if and only if there exists a $\delta\ge 0$ such that
\begin{eqnarray}\delta
\begin{bmatrix}
       \mathbf{A}_1 & \mathbf{b}_1          \\
       \mathbf{b}_1^H & c_1           \\
           \end{bmatrix} -\begin{bmatrix}
       \mathbf{A}_2 & \mathbf{b}_2          \\
       \mathbf{b}_2^H & c_2           \\
           \end{bmatrix}          \succeq \mathbf{0},
\end{eqnarray}
provided that there exists a point $\mathbf{\hat{x}}$ such that $f_k(\mathbf{\hat{x}})<0$.
\end{Lem}

Now, we apply Lemma 1 to constraint $\mbox{\textoverline{C2}}$. In particular, we substitute
 $\mathbf{g}_k=\mathbf{\hat g}_k +\Delta\mathbf{g}_k$ into constraint $\mbox{\textoverline{C2}}$. Therefore, the implication,
\begin{eqnarray}
&&\hspace*{-0.5mm} \Delta\mathbf{g}_k^H \Delta\mathbf{g}_k\hspace*{-1mm}\le\hspace*{-1mm} \varepsilon_k^2\\
\Rightarrow\,\hspace*{-3mm}&&\hspace*{-1.8mm} \mbox{\textoverline{C2}: }0\hspace*{-1mm}\ge\hspace*{-1mm}  \max_{\Delta\mathbf{g}_k\in {\Omega}_k} \Delta\mathbf{g}_k^H\big(\frac{\overline{\mathbf{W}}}{\Gamma_{\mathrm{tol}_k}}-\overline{\mathbf{V}}\big)\Delta\mathbf{g}_k\notag\\
&& \hspace*{-6.8mm}+
2\mathrm{Re}\Big\{\mathbf{\hat g}_k^H\big(\hspace*{-0.5mm}\frac{\overline{\mathbf{W}}}{\Gamma_{\mathrm{tol}_k}}
\hspace*{-0.5mm}-\hspace*{-0.5mm}\overline{\mathbf{V}}\hspace*{-0.5mm}\big)\Delta\mathbf{ g}_k\Big\}+\mathbf{\hat g}_k^H\big(\hspace*{-0.5mm}\frac{\overline{\mathbf{W}}}{\Gamma_{\mathrm{tol}_k}}\hspace*{-0.5mm}-\hspace*{-0.5mm}
\overline{\mathbf{V}}\hspace*{-0.5mm}\big)\mathbf{\hat g}_k\hspace*{-0.5mm}-\hspace*{-0.5mm}\xi\sigma_{\mathrm{z}_k}^2,\forall k,\notag
\end{eqnarray}
holds if and only if there exists a $\delta_k\ge0$ such that the following  LMI constraint holds:
\begin{eqnarray}\label{eqn:LMI_C2}
\mbox{\textoverline{C2}: } &&\hspace*{-5mm}\mathbf{S}_{\mathrm{\overline{C2}}_k}(\overline{\mathbf{W}},\overline{\mathbf{V}}, \xi,\delta_k)\\ \notag
=&&\hspace*{-6mm}
\begin{bmatrix}
       \delta_k\mathbf{I}_{N_{\mathrm{T}}}+\overline{\mathbf{V}}-\frac{\overline{\mathbf{W}}}{\Gamma_{\mathrm{tol}_k}} & \hspace*{-6mm} (\overline{\mathbf{V}}-\frac{\overline{\mathbf{W}}}{\Gamma_{\mathrm{tol}_k}})\mathbf{\hat g}_k          \\
       \hspace*{-2mm}\mathbf{\hat g}_k^H (\overline{\mathbf{V}}-\frac{\overline{\mathbf{W}}}{\Gamma_{\mathrm{tol}_k}})    & \hspace*{-6mm} -\delta_k\varepsilon_k^2 +\xi\sigma_{\mathrm{z}_k}^2+  \mathbf{\hat g}_k^H (\overline{\mathbf{V}}-\frac{\overline{\mathbf{W}}}{\Gamma_{\mathrm{tol}_k}}) \mathbf{\hat g}_k        \\
           \end{bmatrix}\\
         =&&\hspace*{-6mm}\begin{bmatrix}
       \delta_k\mathbf{I}_{N_{\mathrm{T}}}+\overline{\mathbf{V}} & \overline{\mathbf{V}}\mathbf{\hat g}_k          \\
       \mathbf{\hat g}_k^H \overline{\mathbf{V}}    & -\delta_k\varepsilon_k^2 +\xi\sigma_{\mathrm{z}_k}^2+  \mathbf{\hat g}_k^H \overline{\mathbf{V}} \mathbf{\hat g}_k        \\
           \end{bmatrix} \notag\\
         - &&\hspace*{-6mm} \frac{1}{\Gamma_{\mathrm{tol}_k}} \mathbf{U}_{\mathbf{g}_k}^H\overline{\mathbf{W}}\mathbf{U}_{\mathbf{g}_k}\succeq \mathbf{0}, \forall k,\notag
\end{eqnarray}
for $\delta_k\ge 0, k\in\{1,\ldots,K-1\}$ where $\mathbf{U}_{\mathbf{g}_k}=\Big[\mathbf{I}_{N_{\mathrm{T}}}\quad\mathbf{\hat g}_k\Big]$.

Similarly, we rewrite constraints $\mbox{\textoverline{C3}}$, $\mbox{\textoverline{C10}}$, and $\mbox{\textoverline{C11}}$  in the form of (\ref{eqn:S-procedure}) which leads to
\begin{eqnarray}
 \mbox{\textoverline{C3}: }&&\hspace*{-6mm}0\hspace*{-1mm}\ge\hspace*{-1mm}  \max_{\Delta\mathbf{l}_j\in {\Psi}_j} \Delta\mathbf{l}_j^H\big(\frac{\overline{\mathbf{W}}}{\Gamma_{\mathrm{tol}_j}^{\mathrm{PU}}}
 \hspace*{-0.5mm}-\hspace*{-0.5mm}\overline{\mathbf{V}}\big)\Delta\mathbf{l}_j\\
&&\notag\hspace*{-15mm}+
2\mathrm{Re}\Big\{\hspace*{-0.5mm}\mathbf{\hat l}_j^H\big(\frac{\overline{\mathbf{W}}}{\Gamma_{\mathrm{tol}_j}^{\mathrm{PU}}}\hspace*{-0.5mm}-\hspace*{-0.5mm}\overline{\mathbf{V}}\big)\Delta\mathbf{ l}_j\hspace*{-0.5mm}\Big\}\hspace*{-0.5mm}+\hspace*{-0.5mm} \mathbf{\hat l}_j^H\big(\frac{\overline{\mathbf{W}}}{\Gamma_{\mathrm{tol}_j}^{\mathrm{PU}}}\hspace*{-0.5mm}-\hspace*{-0.5mm}\overline{\mathbf{V}}\big)\mathbf{\hat l}_j\hspace*{-0.5mm}-\hspace*{-0.5mm}\xi\sigma_{\mathrm{PU}}^2,\forall k,
\\
\mbox{\textoverline{C10}: }&&\hspace*{-6mm}0\hspace*{-1mm}\ge\hspace*{-1mm}  \max_{\Delta\mathbf{g}_k\in {\Omega}_k} -\eta_k\Big\{\Delta\mathbf{g}_k^H\big(\overline{\mathbf{W}}\hspace*{-0.5mm}+\hspace*{-0.5mm} \overline{\mathbf{V}}\big)\Delta\mathbf{g}_k\\
&&\notag\hspace*{-15mm}+
2\mathrm{Re}\Big\{\mathbf{\hat g}_k^H\big(\overline{\mathbf{W}}\hspace*{-0.5mm}+\hspace*{-0.5mm} \overline{\mathbf{V}}\big)\Delta\mathbf{ g}_k\Big\} \hspace*{-0.5mm}+\hspace*{-0.5mm} \mathbf{\hat g}_k^H\big(\overline{\mathbf{W}}\hspace*{-0.5mm}+\hspace*{-0.5mm} \overline{\mathbf{V}}\big)\mathbf{\hat g}_k\Big\}\hspace*{-0.5mm}-\hspace*{-0.5mm}E_{k}^{\mathrm{SU}},\forall k,\,\mbox{and}
\\
 \mbox{\textoverline{C11}: }&&\hspace*{-6mm}0\hspace*{-1mm}\ge\hspace*{-1mm}  \max_{\Delta\mathbf{l}_j\in {\Psi}_j} \Delta\mathbf{l}_j^H\big(\overline{\mathbf{W}}\hspace*{-0.5mm}+\hspace*{-0.5mm} \overline{\mathbf{V}}\big)\Delta\mathbf{l}_j\\ \notag
&&\hspace*{-15mm}+ 2\mathrm{Re}\Big\{\mathbf{\hat l}_j^H\big(\overline{\mathbf{W}}\hspace*{-0.5mm}+\hspace*{-0.5mm} \overline{\mathbf{V}}\big)\Delta\mathbf{ l}_j\Big\}\hspace*{-0.5mm}+\hspace*{-0.5mm} \mathbf{\hat l}_j^H\big(\overline{\mathbf{W}}\hspace*{-0.5mm}+\hspace*{-0.5mm} \overline{\mathbf{V}}\big)\mathbf{\hat l}_j\hspace*{-0.5mm}-\hspace*{-0.5mm}I_{j}^{\mathrm{PU}},\forall j,
\end{eqnarray}
respectively.

\newcounter{mytempeqncnt}
\begin{figure*}[!t]\setcounter{mytempeqncnt}{\value{equation}}
\setcounter{equation}{38}
  \begin{eqnarray}
\label{eqn:cross-layer-t31-LMI}&&\hspace*{3cm}\underset{\mathbf{\Theta}}
{\mino}\,\,\,\tau\nonumber\\
 &&\notag\hspace*{0.2cm} \mbox{s.t. }\,\mbox{\textoverline{C1}, \textoverline{C4} -- \textoverline{C7}},\,  \mbox{\textoverline{C8}: }\Rank(\overline{\mathbf{W}})=1,\,\mbox{\textoverline{C9}a},\,\mbox{\textoverline{C9}b},\,\mbox{\textoverline{C9}c},\notag\\
 &&\hspace*{-1.2cm}\notag\mbox{\textoverline{C2}: } \mathbf{S}_{\mathrm{\overline{C2}}_k}(\overline{\mathbf{W}},\overline{\mathbf{V}}, \xi,\delta_k)\succeq \mathbf{0},\forall k,\hspace*{1.1cm}\mbox{\textoverline{C3}: } \mathbf{S}_{\mathrm{\overline{C3}}_j}(\overline{\mathbf{W}},\overline{\mathbf{V}}, \xi,\gamma_j)\succeq \mathbf{0},\forall j,\\
&&\hspace*{-1.4cm}\mbox{\textoverline{C10}: }  \mathbf{S}_{\mathrm{\overline{C10}}_k}(\overline{\mathbf{W}},\overline{\mathbf{V}},E_{k}^{\mathrm{SU}}, \varphi_k)\succeq \mathbf{0},\forall k,\hspace*{3mm} \mbox{\textoverline{C11}: } \mathbf{S}_{\mathrm{\overline{C11}}_j}(\overline{\mathbf{W}},\overline{\mathbf{V}},I_{j}^{\mathrm{PU}},\omega_j)\succeq \mathbf{0},\forall j,\notag\\
&&\hspace*{-1.38cm}\mbox{\textoverline{C12}: } \delta_k\ge 0,\forall k,\,\, \,\,\,\, \,\,\,\, \,\,\mbox{\textoverline{C13}: } \gamma_j\ge 0,\forall j,\,\,\,\, \,\mbox{\textoverline{C14}: } \varphi_k\ge 0,\forall k,\,\,\,\, \,\, \,\, \mbox{\textoverline{C15}: } \omega_j\ge 0,\forall j,
\end{eqnarray}\addtocounter{mytempeqncnt}{1}
\setcounter{equation}{35}
\hrulefill
\end{figure*}
 By using Lemma 1, constraint $\mbox{\textoverline{C3}}$, $\mbox{\textoverline{C10}}$, and $\mbox{\textoverline{C11}}$ can be equivalently written as
\begin{eqnarray}\label{eqn:LMI_C3}
&&\mbox{\textoverline{C3}: } \mathbf{S}_{\mathrm{\overline{C3}}_j}(\overline{\mathbf{W}},\overline{\mathbf{V}}, \xi,\gamma_j)\notag\\
\hspace*{-2.5mm}&=&\hspace*{-2.5mm}
          \begin{bmatrix}
       \gamma_j\mathbf{I}_{N_{\mathrm{T}}}+\overline{\mathbf{V}} & \overline{\mathbf{V}}\mathbf{\hat l}_j          \\
       \mathbf{\hat l}_j^H \overline{\mathbf{V}}    & \hspace*{-0.5mm}-\hspace*{-0.5mm}\gamma_j\upsilon_j^2 +\xi\sigma_{\mathrm{PU}}^2+  \mathbf{\hat l}_j^H \overline{\mathbf{V}} \mathbf{\hat l}_j        \\
           \end{bmatrix}\notag\\
            \hspace*{-0.5mm}&-&\hspace*{-0.5mm}\frac{ \mathbf{U}_{\mathbf{l}_j}^H\overline{\mathbf{W}}\mathbf{U}_{\mathbf{l}_j}}{\Gamma_{\mathrm{tol}_j}^\mathrm{PU}} \hspace*{-0.5mm}\succeq \hspace*{-0.5mm}\mathbf{0}, \forall j,
\label{eqn:LMI_C10}\\
&&\mbox{\textoverline{C10}: }\mathbf{S}_{\mathrm{\overline{C10}}_k}(\overline{\mathbf{W}},\overline{\mathbf{V}},E_{k}^{\mathrm{SU}}, \varphi_k)\notag\\
\hspace*{-2.5mm}&=&\hspace*{-2.5mm}
         \begin{bmatrix}
       \varphi_k\mathbf{I}_{N_{\mathrm{T}}}\hspace*{-0.5mm}+\hspace*{-0.5mm}\overline{\mathbf{V}}& \hspace*{-0.5mm}\overline{\mathbf{V}}\mathbf{\hat g}_k          \\ \notag
       \mathbf{\hat g}_h^H \overline{\mathbf{V}}
        & \hspace*{-0.5mm}-\hspace*{-0.5mm}\varphi_k\varepsilon_k^2 \hspace*{-0.5mm}+\hspace*{-0.5mm}\frac{E_{k}^{\mathrm{SU}}}{\eta_k}\hspace*{-0.5mm} +\hspace*{-0.5mm} \mathbf{\hat g}_k^H \overline{\mathbf{V}} \mathbf{\hat g}_k        \\
           \end{bmatrix}\\
           \hspace*{-0.5mm}&+&\hspace*{-0.5mm} \mathbf{U}_{\mathbf{g}_k}^H\overline{\mathbf{W}}\mathbf{U}_{\mathbf{g}_k}\hspace*{-0.5mm}\succeq \hspace*{-0.5mm}\mathbf{0}, \forall k,
\\
\label{eqn:LMI_C11}&&\mbox{\textoverline{C11}: }\mathbf{S}_{\mathrm{\overline{C11}}_j}(\overline{\mathbf{W}},\overline{\mathbf{V}},I_{j}^{\mathrm{PU}},\omega_j)\notag\\
\hspace*{-2.5mm}&=&\hspace*{-2.5mm}
         \begin{bmatrix}
       \omega_j\mathbf{I}_{N_{\mathrm{T}}}-\overline{\mathbf{V}}& -\overline{\mathbf{V}}\mathbf{\hat l}_j          \\
       -\mathbf{\hat l}_j^H \overline{\mathbf{V}}
        & -\omega_j\upsilon_j^2 +I_{j}^{\mathrm{PU}} - \mathbf{\hat l}_j^H \overline{\mathbf{V}} \mathbf{\hat l}_j        \\
           \end{bmatrix}\notag\\
           \hspace*{-0.5mm}&-& \hspace*{-0.5mm} \mathbf{U}_{\mathbf{l}_j}^H\overline{\mathbf{W}}\mathbf{U}_{\mathbf{l}_j}\hspace*{-0.5mm}\succeq\hspace*{-0.5mm} \mathbf{0}, \forall j,
\end{eqnarray}
respectively, with $\mathbf{U}_{\mathbf{l}_j}=\Big[\mathbf{I}_{N_{\mathrm{T}}}\quad\mathbf{\hat l}_j\Big]$ and new auxiliary  optimization variables  $\gamma_j\ge 0, j\in\{1,\ldots,J\},$ $\varphi_k\ge 0, k\in\{1,\ldots,K-1\},$ and $\omega_j\ge 0, j\in\{1,\ldots,J\}$. We note that now constraints $\mbox{\textoverline{C2}}, \mbox{\textoverline{C3}}, \mbox{\textoverline{C10}}$, and $\mbox{\textoverline{C11}}$  involve only a finite number of convex constraints which facilitates an efficient resource allocation algorithm design. As a result, we obtain the following equivalent optimization problem on the top of this page in \eqref{eqn:cross-layer-t31-LMI},  where $\mathbf{\Theta}\triangleq\{ \mathbf{I}^{\mathrm{PU}},\mathbf{E}^{\mathrm{SU}}, \xi,\tau,\boldsymbol{\gamma},\boldsymbol{\delta},\boldsymbol{\varphi},\boldsymbol{\omega},\overline{\mathbf{V}}\in \mathbb{H}^{N_{\mathrm{T}}},\mathbf{\overline{W}}\in \mathbb{H}^{N_{\mathrm{T}}}\}$ denotes the set of optimization variables after transformation;  $\mathbf{I}^{\mathrm{PU}}$ and $\mathbf{E}^{\mathrm{SU}}$ are auxiliary variable vectors  with elements ${I}_j^{\mathrm{PU}},\forall j\in\{1,\ldots,J\},$ and ${E}_k^{\mathrm{SU}},\forall k\in\{1,\ldots,K-1\}$, respectively; $\boldsymbol{\delta},\boldsymbol{\gamma},\boldsymbol{\varphi}$, and $\boldsymbol{\omega}$ are auxiliary optimization variable vectors with elements $\delta_k,\gamma_j,\varphi_k$, and $\omega_j\ge 0$ connected to the constraints in (\ref{eqn:LMI_C2})--(\ref{eqn:LMI_C11}), respectively.

The remaining non-convexity of problem (\ref{eqn:cross-layer-t31-LMI}) is due to the combinatorial rank constraint in $\mbox{\textoverline{C8}}$ on the beamforming matrix $\overline{\mathbf{W}}$. In fact, by relaxing constraint $\mbox{\textoverline{C8}: }\Rank(\overline{\mathbf{W}})=1$, i.e., removing it from (\ref{eqn:cross-layer-t31-LMI}), the considered problem is a convex SDP and can be solved efficiently by numerical solvers such as SeDuMi \cite{JR:SeDumi} and CVX \cite{cvx}. Besides, if the obtained solution for the relaxed SDP  problem admits a rank-one matrix $\overline{\mathbf{W}}$, i.e.,  $\Rank(\overline{\mathbf{W}})=1$, then it is the optimal solution of the  original problem.  In general,  the adopted SDP relaxation  may not yield a rank-one solution and  the result of the relaxed problem serves as  a performance upper bound for the original
problem. Nevertheless, in the following,  we show that there always exists an optimal solution  for the relaxed problem with $\Rank(\mathbf{\overline{W}}) = 1$. In particular, the optimal solution of the relaxed version of  (\ref{eqn:cross-layer-t31-LMI}) with $\Rank(\mathbf{\overline{W}}) = 1$  can be obtained from the solution of the dual problem of the SDP relaxed version of (\ref{eqn:cross-layer-t31-LMI}). In other words, we can obtain the global optimal solutions of non-convex Problems 1, 2, 3, and 4. Furthermore, we propose two suboptimal resource allocation schemes which do not require the solution of the dual problem of the SDP relaxed problem.

 \subsection{Optimality Condition for SDP Relaxation}
In this subsection, we first reveal the tightness of the proposed SDP relaxation. The existence of a rank-one solution  matrix $\overline{\mathbf{W}}$ for the  relaxed SDP version of transformed Problem 4 is summarized in
the following theorem which is based on  \cite[Proposition 4.1]{JR:rui_zhang}\footnote{We note that \cite[Proposition 1]{JR:rui_zhang} was designed for a  communication system with SWIPT for  the case of perfect CSI and single objective optimization. The application of the results of \cite{JR:rui_zhang} to the  scenarios considered  in this paper is only possible after performing the steps and transformations introduced in Section \Rmnum{2} to Section \Rmnum{4}-A.}.

\begin{Thm}\label{prop1}Suppose the optimal solution for the SDP relaxed version of (\ref{eqn:cross-layer-t31-LMI}) is denoted as $\mathbf{\Lambda}^*\triangleq\{\mathbf{I}^{\mathrm{PU}*},\mathbf{E}^{\mathrm{SU}*}, \xi^*,\tau^*,\boldsymbol{\gamma^*}, \boldsymbol{\delta^*},\boldsymbol{\varphi^*},\boldsymbol{\omega^*},\overline{\mathbf{V}}^*,\mathbf{\overline{W}}^*\}$ and  $\Rank(\overline{\mathbf{W}}^*)>1$. Then, there exists a feasible solution for the SDP relaxed version of (\ref{eqn:cross-layer-t31-LMI}), denoted as  $\mathbf{\widetilde \Lambda}\triangleq\{\mathbf{\widetilde I}^{\mathrm{PU}},\mathbf{\widetilde E}^{\mathrm{SU}}, \widetilde \xi,\widetilde\tau,\boldsymbol{\widetilde \gamma},\boldsymbol{\widetilde\delta},\boldsymbol{\widetilde\varphi},\boldsymbol{\widetilde\omega},$ ${\mathbf{\widetilde V}},\mathbf{{ \widetilde W}}\}$, which not only achieves the same objective value as $\mathbf{\Lambda}^*$, but also admits a rank-one matrix $\mathbf{\widetilde W}$, i.e.,  $\Rank(\mathbf{\widetilde W})=1$. The solution $\mathbf{\widetilde \Lambda}$ can be constructed exploiting  $\mathbf{\Lambda}^*$ and  the solution of the dual problem of the relaxed version of  (\ref{eqn:cross-layer-t31-LMI}).
\end{Thm}

\begin{proof}
Please refer to Appendix B for the proof of Theorem \ref{prop1} and the method for constructing the optimal solution.
\end{proof}
Since there always exists  an achievable  optimal solution  with a rank-one beamforming matrix  $\mathbf{\widetilde W}$, the global optimum of (\ref{eqn:cross-layer-t31-LMI}) can be obtained  despite the SDP relaxation. By utilizing Theorem 1, we specify the optimal solution of transformed Problems $1$--$3$ in the following corollary.
\begin{figure*}[!t]\setcounter{mytempeqncnt}{\value{equation}}
\setcounter{equation}{40}
\begin{eqnarray}\label{eqn:suboptimal1}\notag
&& \hspace*{45mm}\underset{\mathbf{\Theta}_{\mathrm{sub}}}{\mino}\,\,  \tau\\
\notag&& \mbox{s.t.} \hspace*{25mm}\mbox{\textoverline{C9}a},\,\,\mbox{\textoverline{C9}b},\,\,\mbox{\textoverline{C9}c},\mbox{\textoverline{C12}}-\mbox{\textoverline{C15}},\notag\\
&&\hspace*{-5mm}\mbox{\textoverline{C1}: }\notag\frac{P_b\Tr(\mathbf{H}\overline{\mathbf{ W}}_{\mathrm{sub}})}{ \Tr(\mathbf{H}\overline{\mathbf{V}})+\sigma_{\mathrm{z}}^2}\hspace*{-0.5mm} \ge\hspace*{-0.5mm} \Gamma_{\mathrm{req}}, \,\,\mbox{\textoverline{C2}: } \mathbf{S}_{\mathrm{\overline{C2}}_k}(P_b \overline{\mathbf{W}}_{\mathrm{sub}},\overline{\mathbf{V}}, \xi,\delta_k)\succeq \mathbf{0},\forall k,\\
 &&\hspace*{-5mm}\notag\mbox{\textoverline{C3}: } \mathbf{S}_{\mathrm{\overline{C3}}_j}(P_b \overline{\mathbf{W}}_{\mathrm{sub}},\overline{\mathbf{V}}, \xi,\gamma_j)\succeq \mathbf{0},\forall j,\quad\mbox{\textoverline{C4}: }\notag  P_b\Tr(\overline{\mathbf{ W}}_{\mathrm{sub}})  +\Tr(\overline{\mathbf{V}})\le P_{\max}\xi,\\
 &&\hspace*{-5mm}\mbox{\textoverline{C5}:}\,\, P_b\ge 0,\mathbf{\overline{V}}\succeq \mathbf{0},\quad\mbox{\textoverline{C6}:}\,\, \xi \ge 0,\quad\mbox{\textoverline{C7}:}\,\, P_b\Tr(\overline{\mathbf{ W}}_{\mathrm{sub}})+\Tr(\overline{\mathbf{V}})= 1,\quad \notag\\
&&\hspace*{-7mm}\mbox{\textoverline{C10}: }  \mathbf{S}_{\mathrm{\overline{C10}}_k}(P_b \overline{\mathbf{W}}_{\mathrm{sub}},\overline{\mathbf{V}}, {E}^{\mathrm{SU}}_k,\varphi_k)\succeq \mathbf{0},\forall k,\quad \mbox{\textoverline{C11}: } \mathbf{S}_{\mathrm{\overline{C11}}_j}(P_b \overline{\mathbf{W}}_{\mathrm{sub}},\overline{\mathbf{V}}, {I}^{\mathrm{PU}}_j,\omega_j)\succeq \mathbf{0},\forall j,
\end{eqnarray}\hrulefill\setcounter{equation}{41}
\end{figure*}
\begin{Cor} Transformed Problems $1$--$3$ can be solved optimally by applying SDP relaxation and the solution of each problem can be obtained with the method provided in the proof of Theorem $1$.  In particular,  Problem $p$ can be solved by solving  Problem \ref{Prob:multi} with $\lambda_p=1$, $\lambda_i=0, \forall i\ne p$, $i\in\{1,2,3\}$, and setting $F_p^*,\forall p\in\{1,2,3\},$ to any non-negative and finite constant\footnote{$F_p^*,\forall p\in\{1,2,3\}$, are considered to be given constants in Problem $4$ for studying the trade-offs between objective functions $1$, $2$, and $3$. Setting $F_p^*=c$ where $0<c<\infty$ is a constant in Problem $4$ is used for recovering the solution of Problems $1$, $2$, and $3$. However, this does not imply that the optimal value of Problem $p$ is equal to $c$.}.
\end{Cor}

 \setcounter{equation}{39}
\begin{Remark}
 The computational complexity of the proposed optimal algorithm with respect to the numbers of secondary users $K$, the number of primary users $J$, and   transmit antennas $N_{\mathrm{T}}$ at the secondary transmitter can be characterized as
\begin{eqnarray}
&&\bigo\Bigg(\Big(\sqrt{2N_{\mathrm{T}}}\log\big(\frac{1}{\kappa}\big)\Big)
\Big((2K+2J)(2N_{\mathrm{T}})^3\notag\\
&&+(2N_{\mathrm{T}})^2(2K+2J)^2+(2K+2J)^3\Big)\Bigg)
\end{eqnarray}
for a given solution accuracy $\kappa>0$, where  $\bigo(\cdot)$ is the big-O notation. We note that polynomial time computational complexity algorithms are considered to be fast algorithms in the literature \cite[Chapter 34]{book:polynoimal} and are desirable for real time implementation. Besides, the computational complexity can be further reduced by adopting a tailor made interior point method \cite{JR:complexity1,JR:complexity2}. Also, we would like to emphasize that in practice, $\lambda_1,\lambda_2,$ and $\lambda_3$ are given parameters and thus we only need to compute one point of the trade-off region.
\end{Remark}

\subsection{Suboptimal Resource Allocation Schemes}
As discussed in Appendix B, constructing the optimal solution $\mathbf{\widetilde \Lambda}$ with $\Rank(\mathbf{\widetilde W})=1$ requires the solution of the  dual problem of problem (\ref{eqn:cross-layer-t31-LMI}) as the Lagrange multiplier matrix $\mathbf{Y}^*$ is needed  in (\ref{eqn:Y}). Nevertheless, $\mathbf{Y}^*$ may not be provided by some numerical solvers and thus the construction of a rank-one solution matrix $\mathbf{\widetilde W}$ may not be possible.  In the following, we propose two suboptimal resource allocation schemes based on the solution of the primal problem of the relaxed version of (\ref{eqn:cross-layer-t31-LMI}) which do not require knowledge of $\mathbf{Y}^*$ when $\Rank(\overline{\mathbf{W}}^*)>1$.

\subsubsection{Suboptimal Resource Allocation Scheme 1}
  The first proposed suboptimal resource allocation scheme is a hybrid scheme and is based on the  solution of the relaxed version of (\ref{eqn:cross-layer-t31-LMI}). We first solve (\ref{eqn:cross-layer-t31-LMI}) by SDP relaxation. If the solution admits a  rank-one $\overline{\mathbf{W}}$, then  the global optimal solution of  (\ref{eqn:cross-layer-t31-LMI}) is obtained. Otherwise, we  construct a suboptimal beamforming matrix $ \overline{\mathbf{ W}}_{\mathrm{sub}}$. Suppose $\overline{\mathbf{W}}$ is a matrix with rank $N$. Then $\overline{\mathbf{W}}$ can be written as  $\overline{\mathbf{W}}=\sum_{t=1}^N \vartheta_t \mathbf{e}_t\mathbf{e}^H_t $, where $\vartheta_t$ and $\mathbf{e}_t\in\mathbb{C}^{N_{\mathrm{T}}\times 1}$ are the descending eigenvalues, i.e.,  $\vartheta_1\ge\vartheta_2\ge,\ldots,\ge\vartheta_t,\ldots,\ge\vartheta_N$, and eigenvectors  associated with $\overline{\mathbf{W}}$, respectively.  Now, we introduce the suboptimal beamforming vector $\overline{\mathbf{ w}}_{\mathrm{sub}}=\mathbf{e}_1$   such that  $\overline{\mathbf{ W}}_{\mathrm{sub}} = \overline{\mathbf{ w}}_{\mathrm{sub}}\overline{\mathbf{ w}}_{\mathrm{sub}}^H$.  Then, we define a scalar optimization variable $P_b$ which controls the power of the suboptimal beamforming matrix. As a result, a new optimization problem is then given by \eqref{eqn:suboptimal1} on the top of this page, where $\mathbf{\Theta}_{\mathrm{sub}}\triangleq\{P_b,\mathbf{I}^{\mathrm{PU}},\mathbf{E}^{\mathrm{SU}}, \xi,\tau,\boldsymbol{\gamma},\boldsymbol{\delta},\boldsymbol{\varphi},\boldsymbol{\omega},\overline{\mathbf{V}}\in \mathbb{H}^{N_{\mathrm{T}}}\}$ is the new set of optimization variables for suboptimal resource allocation scheme 1. The problem formulation in (\ref{eqn:suboptimal1}) is jointly convex with respect to the optimization variables and can be solved by using efficient numerical solvers. Besides, the solution of (\ref{eqn:suboptimal1}) satisfies the constraints of (\ref{eqn:cross-layer-t31-LMI}), thus  the solution of (\ref{eqn:suboptimal1})  serves as a suboptimal solution for (\ref{eqn:cross-layer-t31-LMI}) since the beamforming matrix $\overline{\mathbf{ W}}_{\mathrm{sub}}$ is fixed which leads to reduced degrees of freedom for resource allocation.

\subsubsection{Suboptimal Resource  Allocation Scheme 2}
The second proposed suboptimal resource allocation scheme is also a hybrid scheme. It adopts a similar approach to solve the problem  as   suboptimal resource  allocation scheme $1$, except for the choice of the suboptimal beamforming matrix $\overline{\mathbf{ W}}_{\mathrm{sub}}$ when  $\Rank(\overline{\mathbf{W}}^*)>1$. Here, the choice of beamforming matrix  $\overline{\mathbf{ W}}_{\mathrm{sub}}$  is based on
the rank-one Gaussian randomization scheme  \cite{JR:Gaussian_randomization}.  Specifically, we calculate the eigenvalue decomposition of $\overline{\mathbf{W}}^*=\mathbf{U}\mathbf{\Sigma}\mathbf{U}^H$, where $\mathbf{U}=\Big[\mathbf{e}_1\ldots\mathbf{e}_N\Big]$ and $\mathbf{\Sigma}=\diag\big({\vartheta_1},\ldots,{\vartheta_N}\big)$  are an $N_\mathrm{T}\times N_\mathrm{T}$ unitary matrix and a diagonal matrix, respectively. Then, we adopt the suboptimal beamforming vector $\overline{\mathbf{ w}}_{\mathrm{sub}}=\mathbf{U}\mathbf{\Sigma}^{1/2}\mathbf{r},  \overline{\mathbf{ W}}_{\mathrm{sub}}=P_b\overline{\mathbf{ w}}_{\mathrm{sub}}\overline{\mathbf{ w}}_{\mathrm{sub}}^H$, where $\mathbf{r}\in {\mathbb C}^{N_{\mathrm{T}}\times 1}$ and $\mathbf{r}\sim {\cal CN}(\mathbf{0}, \mathbf{I}_{N_{\mathrm{T}}})$. Subsequently, we follow the same approach as in (\ref{eqn:suboptimal1}) for optimizing $\mathbf{\Theta}_{\mathrm{sub}}$ and obtain a suboptimal rank-one solution $P_b \overline{\mathbf{ W}}_{\mathrm{sub}}$.  We note that suboptimal resource  allocation scheme $2$ provides flexibility for trading computational complexity and system performance which is not offered by scheme $1$. In fact, by executing  scheme $2$ repeatedly for different  Gaussian distributed random vectors $\mathbf{r}$,  the performance of scheme $2$ can be improved by selecting the best $\overline{\mathbf{ w}}_{\mathrm{sub}}=\mathbf{U}\mathbf{\Sigma}^{1/2}\mathbf{r}$ over different trials.

\begin{Remark}\label{eqn:from_fracitonal_to_non_fractional}
We note  that
the solution of the total received interference power minimization  \Big($\underset{\mathbf{V}\in \mathbb{H}^{N_{\mathrm{T}}},\mathbf{w}
}{\mino}\,\underset{\Delta\mathbf{l}_j\in {\Psi}_j}{\max}$  $\mathrm{IP}(\mathbf{w},\mathbf{V})$\Big) and total harvested power maximization \Big($\underset{\mathbf{V}\in \mathbb{H}^{N_{\mathrm{T}}},\mathbf{w}}{\maxo}\,  \underset{{\Delta\mathbf{g}_k\in {\Omega}_k}}{\min} \mathrm{HP}(\mathbf{w},\mathbf{V})$\Big) problems can be obtained by applying Corollary 1 and solving Problem 4 after setting $\zeta=1$ and removing constraint $\mbox{\textoverline{C7}}$.
\end{Remark}

\section{Results}
\label{sect:result-discussion} We evaluate the
system performance of the proposed resource allocation schemes using simulations. The important simulation parameters are summarized in Table \ref{tab:parameters}.  A reference distance of $2$ meters  for the path loss model is selected.  There are $K$ receivers uniformly distributed between the reference distance and the maximum service distance of $20$ meters in the secondary network. Besides, we assume that the primary transmitter is  $40$ meters away from the secondary transmitter. In particular,  there are $J$ primary receivers  uniformly distributed between $20$ meters and $40$  from the secondary transmitter, cf. Figure
  \ref{fig:simulation_model}.     Because of  path loss and channel fading, different  secondary receivers experience different interference powers from the primary transmitter\footnote{ From Table \ref{tab:parameters}, we observe that the secondary transmitter, which needs to provide both information and energy, has a higher maximum transmit power budget compared to the primary transmitter which only provides information signals to the receivers in its networks. In fact, by exploiting the extra degrees of freedom offered by the multiple transmit antennas, the secondary transmitter can transmit a high power to the secondary receivers and cause a minimal interference to the primary network.  On the contrary, the primary transmitter is equipped with a single  antenna only and has to transmit with a relatively small power to avoid harmful interference.   }.    To facilitate the presentation, in the sequel, we define the normalized maximum  channel estimation errors of primary receiver $j$ and idle secondary receiver $k$  as  $\sigma_{\mathrm{PU}_j}^2=\frac{\upsilon^2_j}{\norm{\mathbf{l}_j}^2}$ and  $\sigma_{\mathrm{SU}_k}^2=\frac{\varepsilon^2_k}{\norm{\mathbf{g}_k}^2}$, respectively, with $\sigma_{\mathrm{PU}_a}^2=\sigma_{\mathrm{PU}_b}^2,\forall a, b\in\{1,\ldots,J\}$, for all primary receivers and $\sigma_{\mathrm{SU}_c}^2=\sigma_{\mathrm{SU}_d}^2,\forall c, d\in\{1,\ldots,K-1\}$, for all  secondary receivers, respectively.  Unless specified otherwise, we assume normalized maximum  channel estimation errors of  idle secondary receiver $k$ and primary receiver $j$ of  $\sigma_{\mathrm{SU}_k}^2=0.01,\sigma_{\mathrm{PU}_j}^2=0.05,\forall k,j$. Besides, we study the trade-off between the different objective functions via the solution of Problem 4 for two cases. In particular, in Case I, we study  the  trade-off between the objective functions for total harvested power maximization,   total interference power leakage minimization, and  total transmit power minimization, cf. Remark \ref{remark1} and Remark \ref{eqn:from_fracitonal_to_non_fractional}; in Case II, we study the trade-off between the objective functions for energy harvesting efficiency maximization, average IPTR minimization, and average total transmit power minimization.    The average system performance shown in the following sections is obtained by averaging over different realizations of  both path loss and multipath fading.

\begin{figure}[t]
\centering
\includegraphics[width=3.5in]{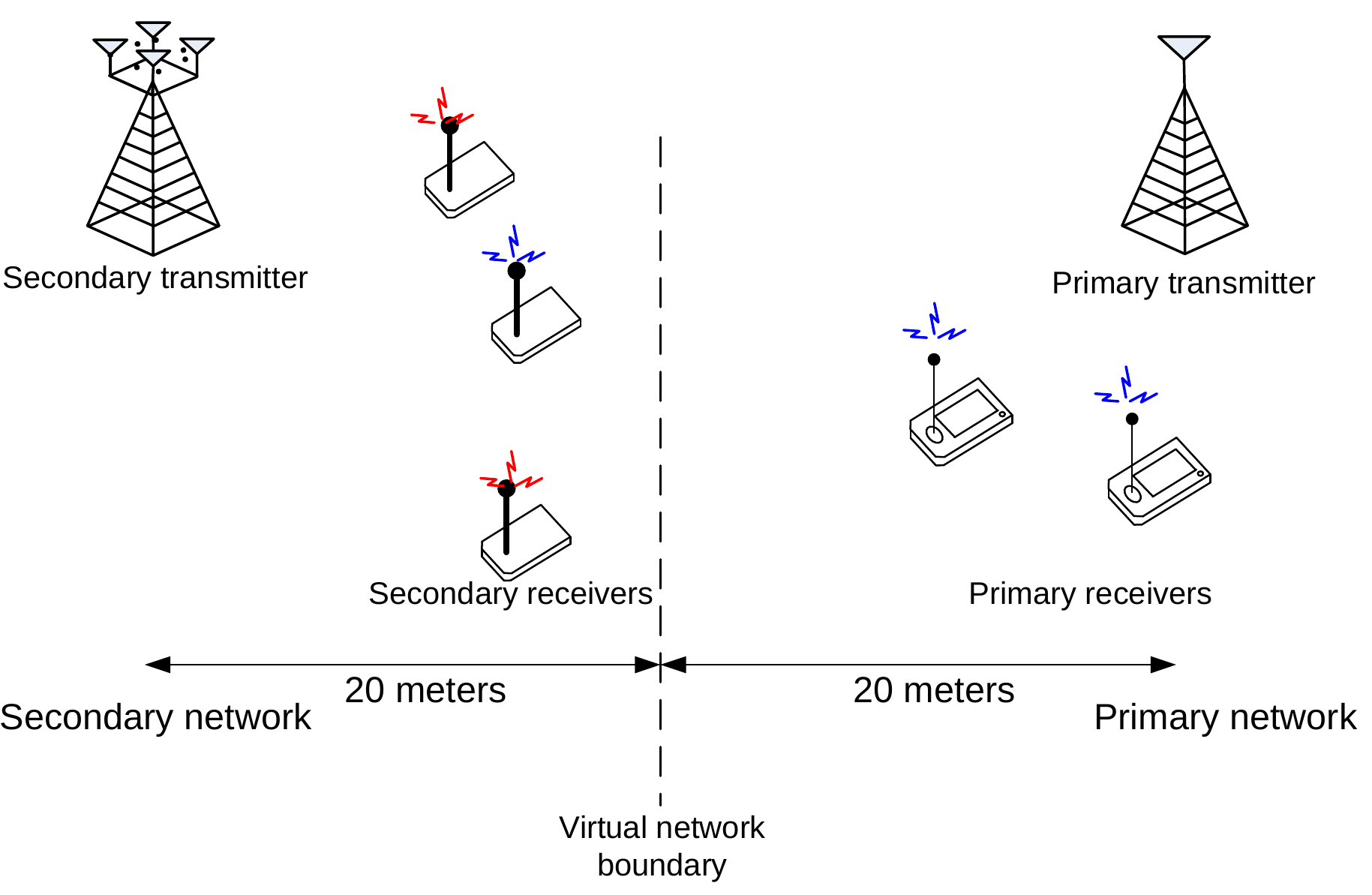}
\caption{CR SWIPT network simulation topology.}
\label{fig:simulation_model}
\end{figure}

\begin{table}[t]\caption{System  parameters.}\label{tab:parameters} \centering{
\begin{tabular}{|L|l|}\hline
\hspace*{-1mm}Carrier center frequency & 915 MHz  \\
\hline
\hspace*{-1mm}Path loss model&  TGn path loss model E \cite{report:tgn}  \\
\hline
\hspace*{-1mm}Multipath fading distribution & \mbox{Rician fading} with Rician factor $3$ dB  \\

\hline
\hspace*{-1mm}Total noise variance, $\sigma_{\mathrm{s}}^2$ &  \mbox{$-23$ dBm}   \\
\hline
\hspace*{-1mm}Transmit antenna gain &  $10$ dBi   \\
\hline
\hspace*{-1mm}Number of transmit antennas at the secondary transmitter, $N_{\mathrm{T}}$ &  $8$   \\
\hline
\hspace*{-1mm}Max. transmit power  allowance at the secondary transmitter, $P_l^{T_{\max}}$ & $30$ dBm \\
\hline
\hspace*{-1mm}Min. required SINR  of the desired secondary receiver, $\Gamma_{\mathrm{req}}$ & $20$ dB  \\
\hline
\hspace*{-1mm}Max. tolerable SINR at the potential eavesdroppers, $\Gamma_{\mathrm{tol}_k}=\Gamma_{\mathrm{tol}_j}^{\mathrm{PU}}$ & $0$ dB  \\
\hline
\hspace*{-1mm}Minimum required secrecy rate, $\log_2(1+\Gamma_{\mathrm{req}})-\log_2(1+\Gamma_{\mathrm{tol}_k})$  &  $5.6582$ bit/s/Hz   \\
           \hline
\hspace*{-1mm}Transmit power of primary transmitter   &  $5$ dBm   \\
           \hline
\end{tabular}\vspace*{-4mm}}
\end{table}

\subsection{Trade-off Regions for Case I and Case II}
Figures \ref{fig:pareto_users1} and \ref{fig:pareto_users2} depict the trade-off regions for the system objectives  for  Case I and Case II achieved by the proposed optimal resource allocation scheme, respectively. There are
 one active secondary receiver, $K-1=3$ idle secondary  receivers, and $J=2$ primary receivers.  The trade-off regions in  Figures \ref{fig:pareto_users1} and \ref{fig:pareto_users2} are obtained by solving Problem $4$
 via varying the values of $0\le \lambda_p\le 1,\forall p\in\{1,2,3\}$,   uniformly  for a step size of $0.04$ such that $\sum_p \lambda_p=1$. We use asterisk markers to denote the trade-off region achieved by the considered resource allocation scheme  and  colored circles to represent the Pareto frontier \cite{JR:MOOP}.
\begin{figure}[t]
\subfigure[System objective trade-off region for Case I.]{
\includegraphics[scale=0.45]{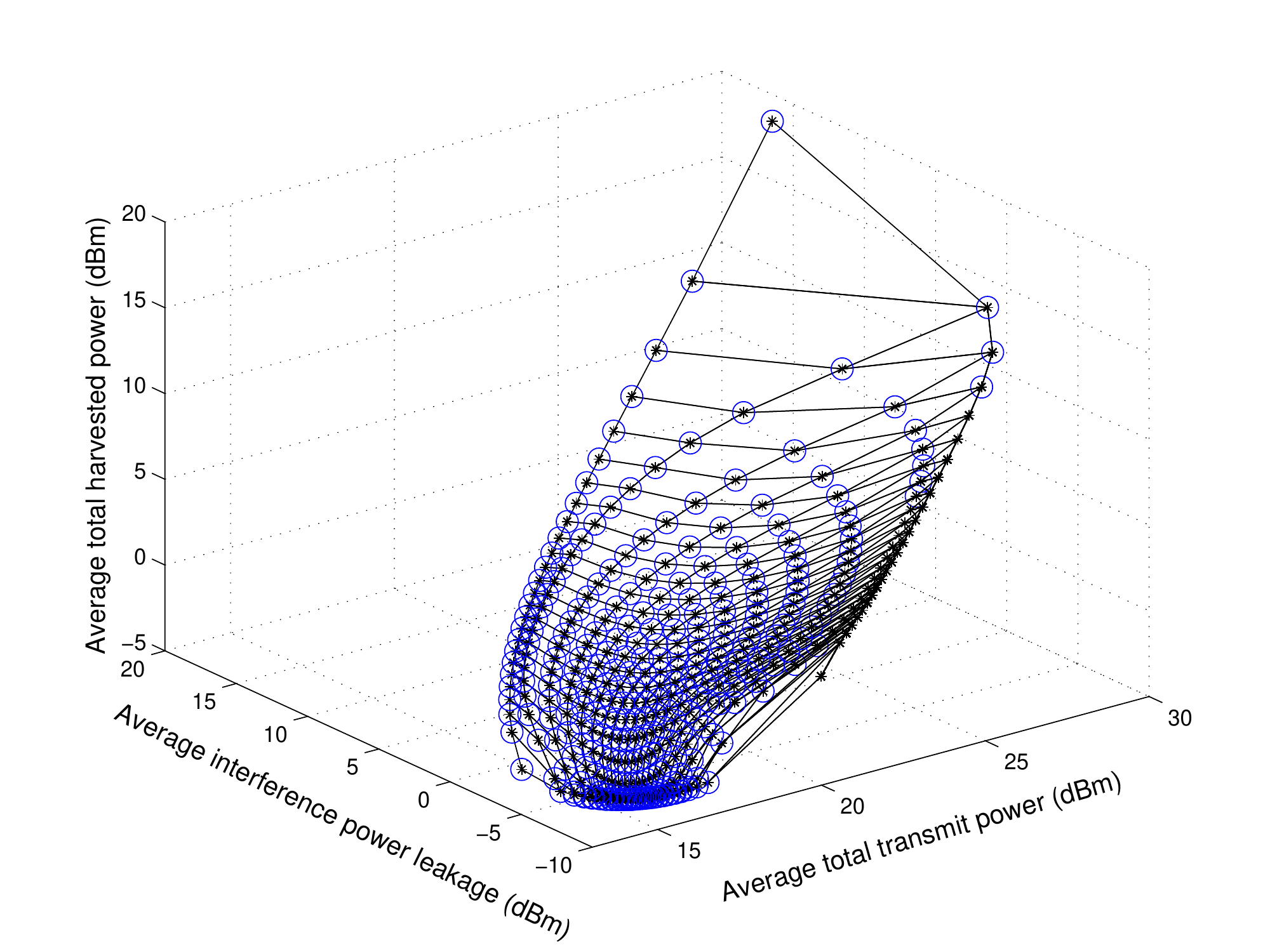}
\label{fig:pareto_users1}}  \subfigure[System objective trade-off region for Case II.]{
\includegraphics[scale=0.45]{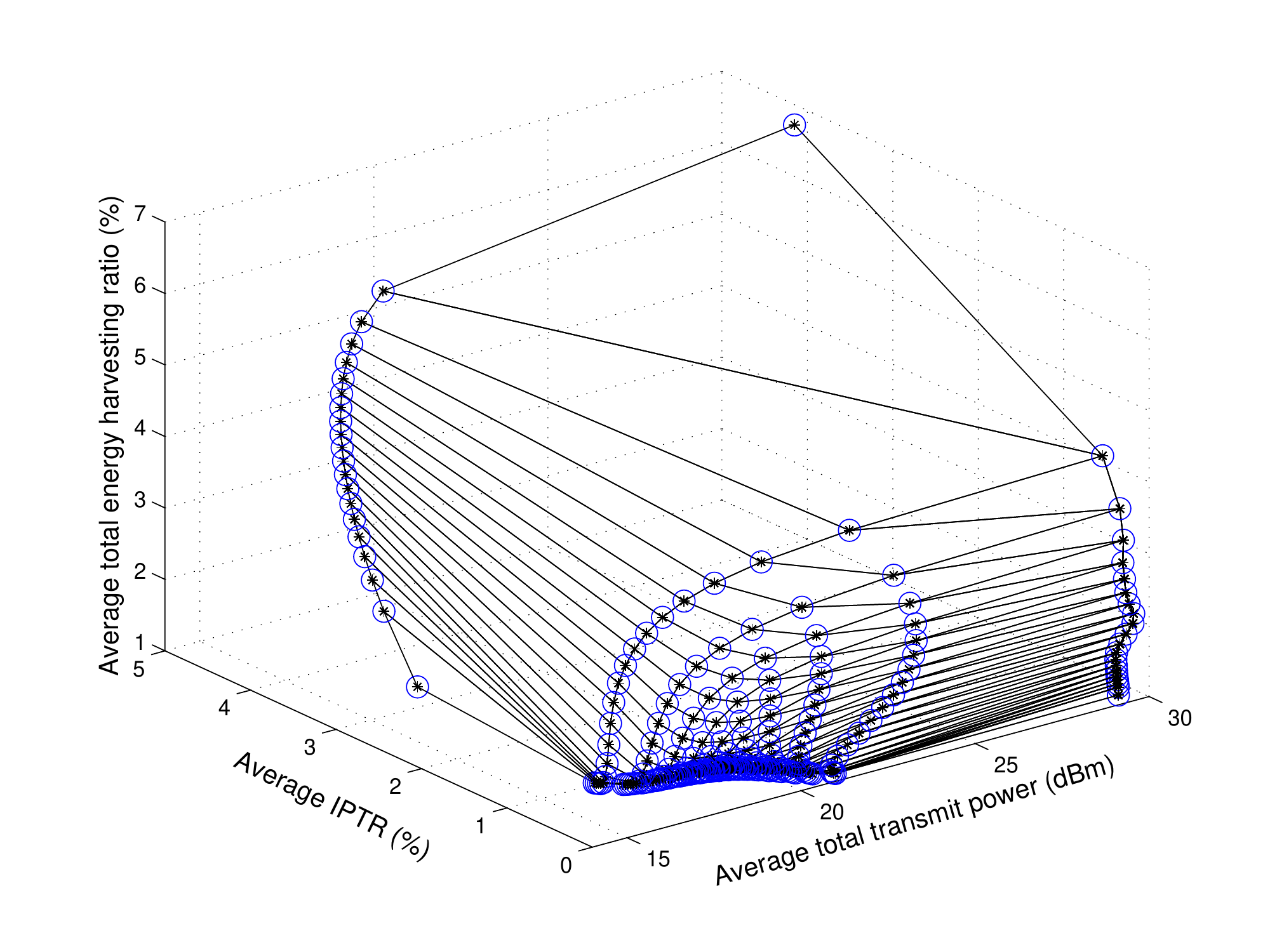}
\label{fig:pareto_users2} }\caption[System objective trade-off region for Case II]
{Three-dimensional system objective trade-off regions achieved by the proposed optimal resource allocation scheme. Asterisk markers denote the trade-off region achieved by the resource allocation scheme  and  colored circles represent the Pareto frontier. }\label{fig:pareto}
\end{figure}
 For the trade-off region for Case I in Figure \ref{fig:pareto_users1}, it can be observed that although the system design objectives of  total transmit power minimization and  total interference power leakage minimization do not share the same optimal solution (a single point which is the minimum of both objective functions), these two objectives  are  partially aligned with each other. Specifically, a large portion of  the trade-off region is concentrated at the bottom of the figure. In other words, a resource allocation policy which minimizes the total transmit power can also reduce the total interference power leakage effectively and vice versa.  On the contrary, the objective of total harvested power maximization conflicts with the other two objective functions. In particular, in order to maximize the total harvested power, the secondary transmitter has to transmit with full power in every time instant despite the imperfection of the  CSI. The associated  resource allocation policy  with full power transmission corresponds to the top corner point in Figure  \ref{fig:pareto_users1}. Besides, if the secondary transmitter employs a large  transmit power,  a  high average total interference power leakage at the primary receivers will result. Furthermore, the total harvested power in the system is in the order of milliwatt which  is sufficient to charge the sensor type idle secondary receivers,  despite the existence of the primary receivers, cf. footnote \ref{label:fn}.

For the trade-off region for Case II in Figure \ref{fig:pareto_users2}, it can be seen that   a significant portion of the trade-off region is concentrated near the bottom and the remaining parts spread over the entire space of the figure. The fact that the trade-off region is condensed near the bottom indicates that  resource allocation policies which minimize the total transmit power can also  reduce the IPTR   to a certain extent and vice versa. However, there also exist resource allocation policies that incur a high transmit power while achieving  a low  IPTR, i.e., the points located near an average total transmit power of $30$ dBm and average IPTR $=0.1\%$.   This can be explained by the fact that the objective functions for energy harvesting efficiency maximization and IPTR minimization are invariant to a  simultaneous positive scaling of both $\mathbf{W}$ and $\mathbf{V},$ e.g. $\frac{\mathrm{HP}(c\mathbf{W},c\mathbf{V})}{\mathrm{TP}(c\mathbf{W},c\mathbf{V})}=\frac{\mathrm{HP}(\mathbf{W},\mathbf{V})}{\mathrm{TP}(\mathbf{W},\mathbf{V})} $ and $\frac{\mathrm{IP}(c\mathbf{W},c\mathbf{V})}{\mathrm{TP}(c\mathbf{W},c\mathbf{V})}=\frac{\mathrm{IP}(\mathbf{W},\mathbf{V})}{\mathrm{TP}(\mathbf{W},\mathbf{V})}$ for $c>0$. As a result, if total transmit power minimization is not a system design objective in Problem $4$, i.e., $\lambda_2=0$, optimal solutions of Problem $4$ in the trade-off region may exist such that the secondary transmitter transmits with a high power while still satisfying all constraints.  On the other hand, to achieve the maximum energy harvesting efficiency in the secondary network, i.e., the top corner point in Figure  \ref{fig:pareto_users2},  the secondary transmitter has to transmit with maximum power which leads to  a high average IPTR.

\begin{figure}[t]
\subfigure[System objective trade-off region for Case I.]{
\includegraphics[scale=0.45]{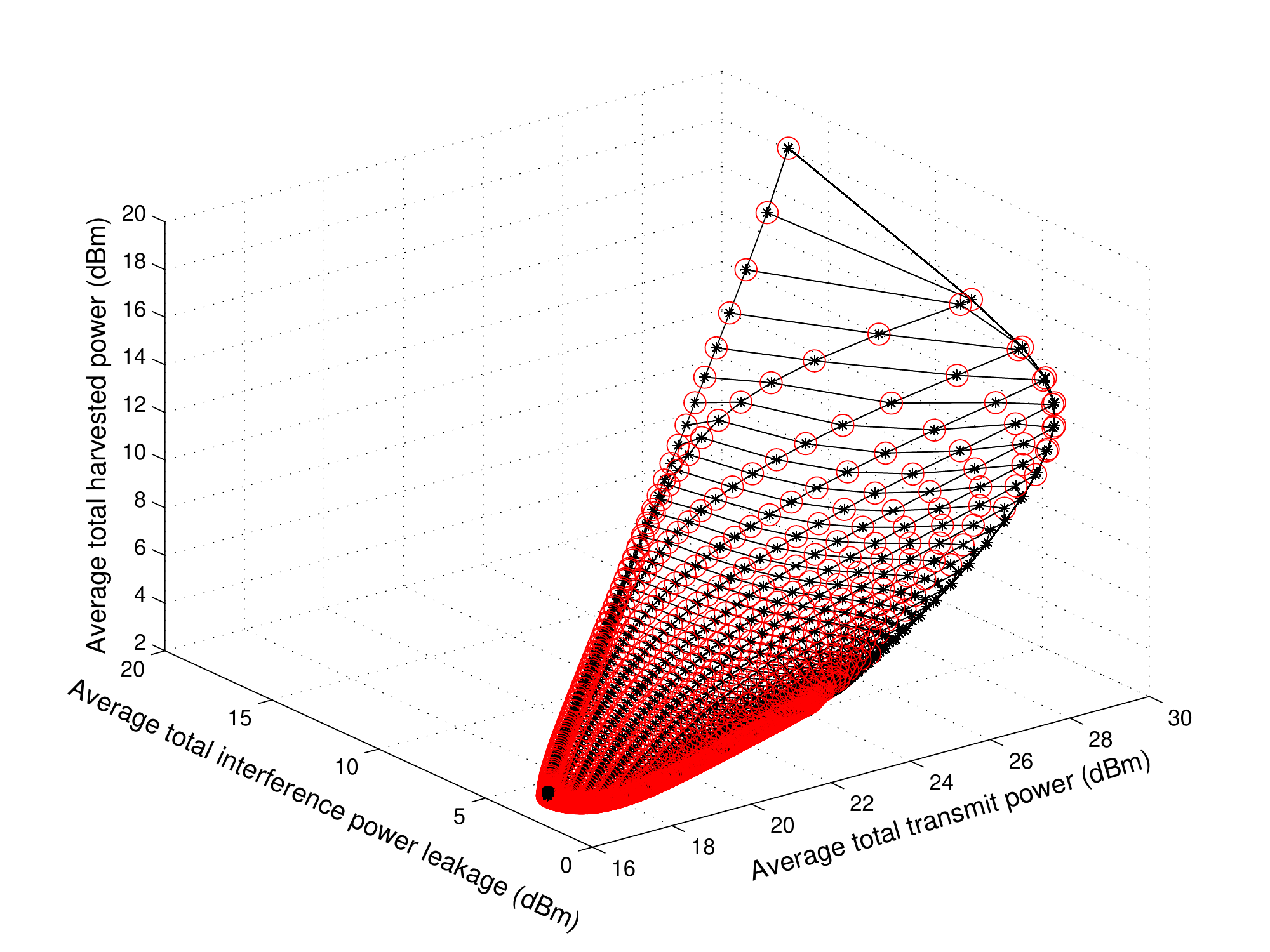}
\label{fig:MRT_pareto_users1}}   \subfigure[System objectives trade-off region for Case II.]{
\includegraphics[scale=0.45]{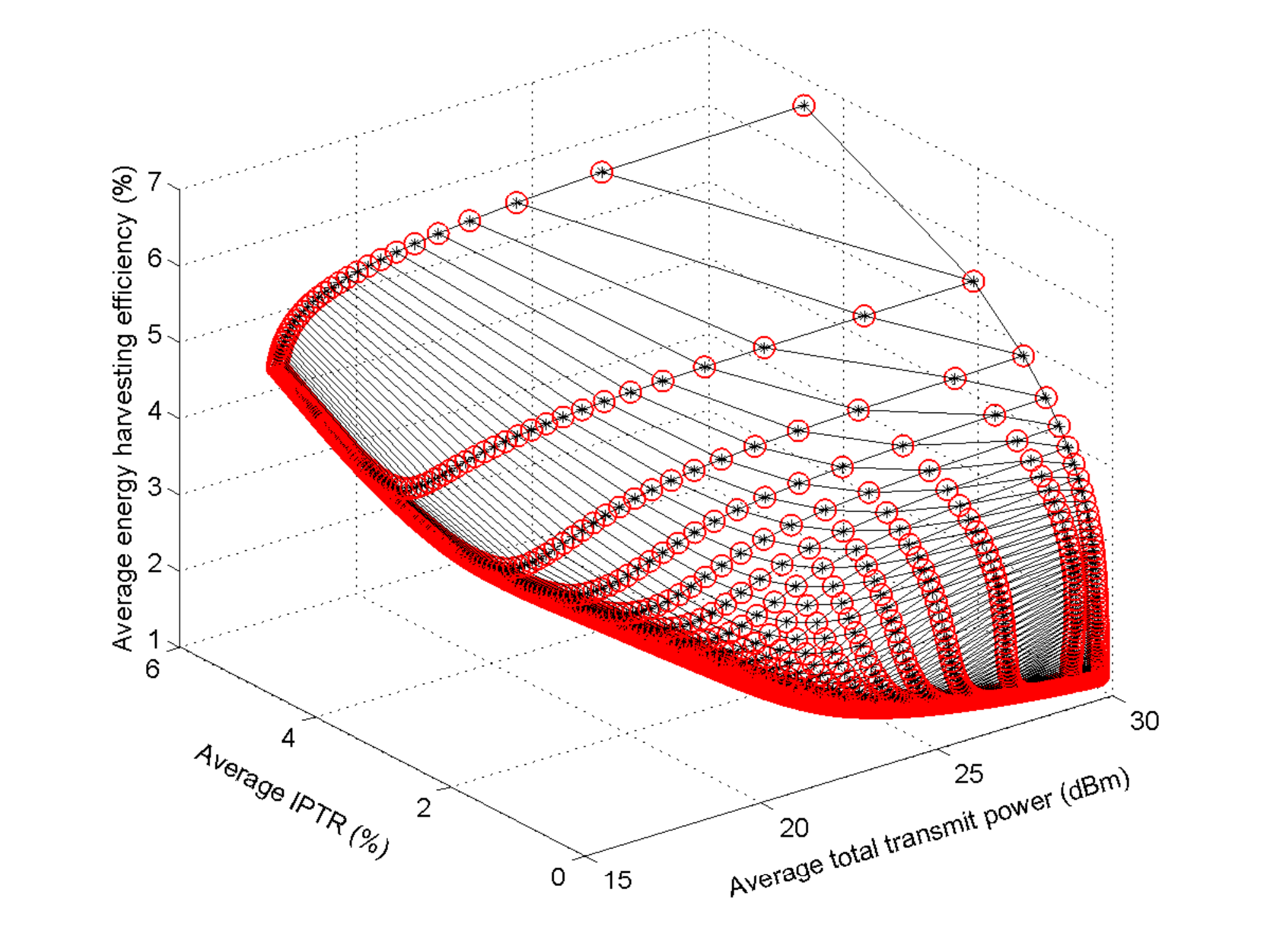}
\label{fig:MRT_pareto_users2} }\caption[System objective trade-off region for Case II]
{Three-dimensional system objective trade-off regions achieved by the baseline scheme. Asterisk markers denote the trade-off region achieved by the baseline resource allocation scheme  and  colored circles represent the Pareto frontier with respect to the baseline resource allocation scheme. }\label{fig:p_MRT}
\end{figure}

For comparison, in Figure \ref{fig:p_MRT}, we plot the trade-off regions achieved by a baseline
resource  allocation scheme for Case I and Case II in Figure
3. For the baseline scheme,  we adopt  maximum ratio transmission (MRT) with respect to the desired  secondary receiver for information beamforming matrix ${\mathbf{W}}$. In other words, the beamforming direction of matrix ${\mathbf{W}}$ is fixed and it has  a rank-one structure.  Then, we optimize the artificial noise covariance matrix $\mathbf{{V}}$ and the power of  ${\mathbf{W}}$ in Problem $4$  via varying the values of $0\le \lambda_p\le 1,\forall p\in\{1,2,3\}$. We note that  the proposed baseline scheme requires the same amount of CSI as the proposed optimal scheme. However, the baseline scheme does not fully exploit the available CSI for resource allocation optimization and, as a result, the required computational complexity is reduced roughly by half compared to the proposed optimal scheme.

 As can be observed from Figures \ref{fig:pareto_users1}, \ref{fig:pareto_users2}, \ref{fig:MRT_pareto_users1}, and \ref{fig:MRT_pareto_users2}, the baseline scheme  is  effective  in maximizing the
energy harvesting efficiency and the total harvested power in the high
transmit power regime and is able to approach the optimal trade-off region achieved by the proposed optimal SDP based resource allocation scheme. This can be explained by the fact that both the optimal scheme and the baseline scheme optimize the covariance matrix of the artificial noise  which contributes most of the  power transferred to the idle secondary receivers. In particular,  by exploiting the spatial degrees of freedom offered by the multiple antennas, multiple narrow energy beams can be created via the proposed optimization framework for transfer of the artificial noise. The narrow energy beams help in focusing energy   on the idle secondary receivers which  increases the energy transfer efficiency.  Nevertheless, when the total transmit power budget of the secondary transmitter is small, the baseline scheme may not be able to satisfy  the QoS constraints which leads to a smaller trade-off region compared to the proposed optimal scheme.

\begin{figure}[t]
\subfigure[Average total harvested power.  ]{
\includegraphics[scale=0.45]{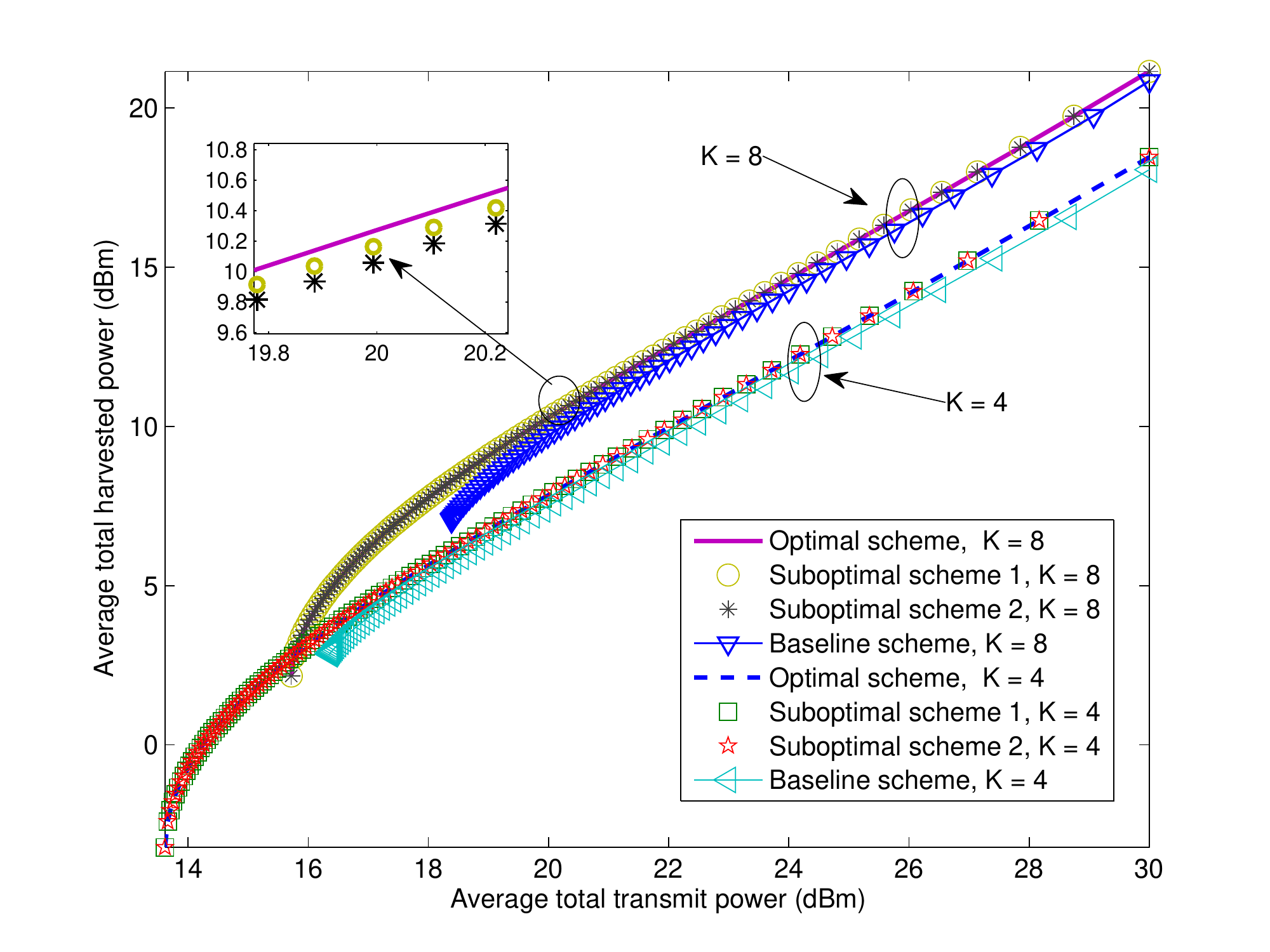}
\label{fig:hp_pt}}  \subfigure[Average energy harvesting efficiency.]{
\includegraphics[scale=0.45]{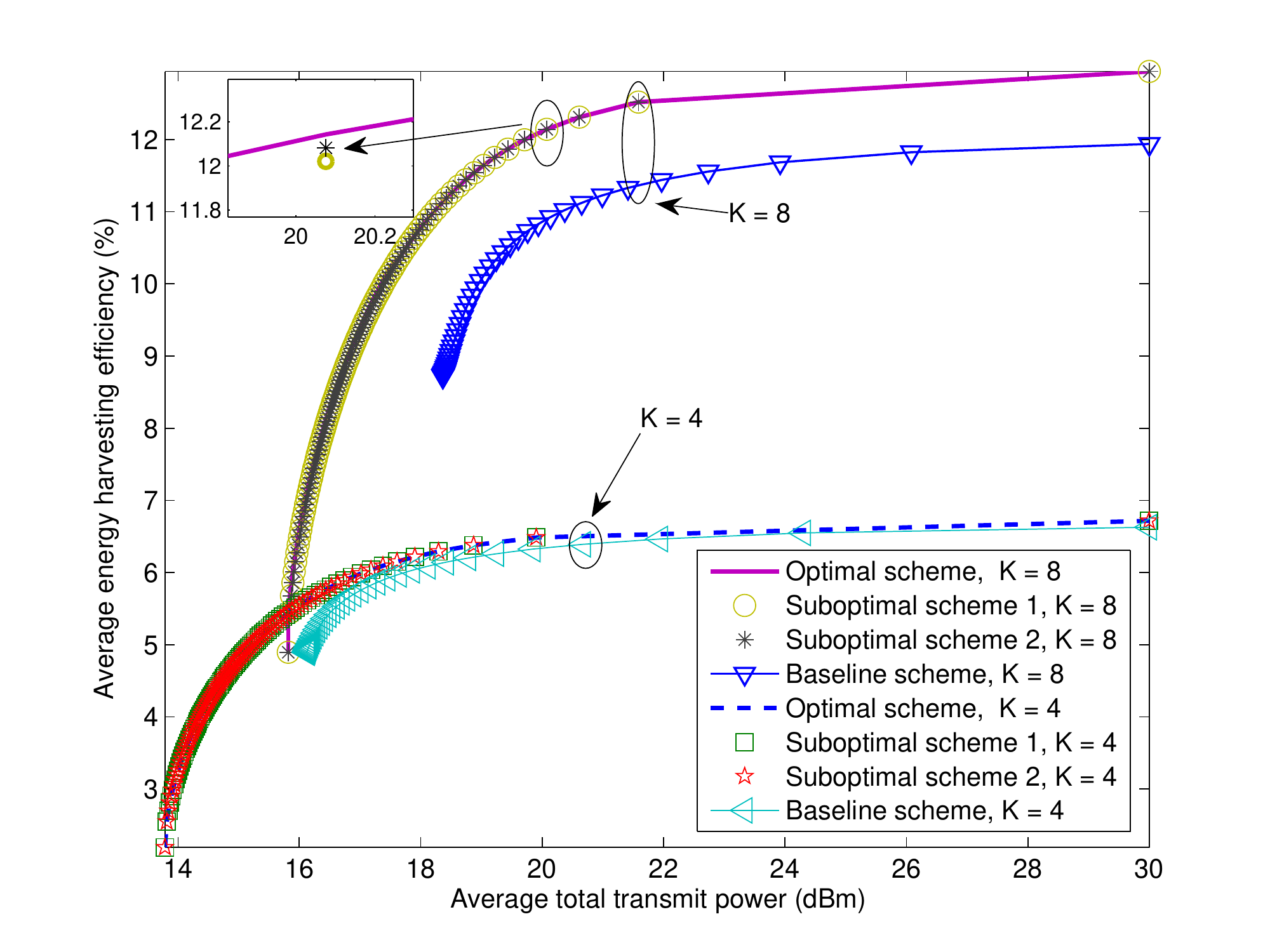}
\label{fig:ee_pt} }\caption[Average total harvested power.]
{Average total harvested power and  average energy harvesting efficiency versus the average total transmit power of the secondary system  for different resource allocation schemes and different numbers of secondary receivers, $K$.  }
\end{figure}

\subsection{Average Total Harvested Power and Average Energy Harvesting Efficiency}
Figures  \ref{fig:hp_pt} and \ref{fig:ee_pt}  depict the average total harvested power and the average energy harvesting efficiency of the secondary system   versus the average total transmit power for different numbers of secondary receivers, $K$, respectively. The curves in  Figures  \ref{fig:hp_pt} and \ref{fig:ee_pt} are obtained for Case I and Case II, respectively. Specifically, for each  case, we solve Problem 4 for  $\lambda_3=0$ and  $0\le \lambda_p\le 1,\forall p\in\{1,2\}$, where the values of $\lambda_p,p\in\{1,2\},$ are  uniformly varied  for a step size of $0.01$  such that $\sum_p \lambda_p=1$. It can be observed from Figures \ref{fig:hp_pt} and \ref{fig:ee_pt}  that the average total harvested power and the average energy harvesting efficiency  are monotonically increasing functions
 with respect to  the total transmit power.  In other words,  total harvested power/energy harvesting efficiency maximization and  total transmit power minimization are conflicting system design objectives.  Also, the proposed optimal scheme outperforms the baseline scheme. In particular, the  proposed optimal scheme fully utilizes the available CSI  and  provides a larger trade-off region, e.g.  $2.7$ dB less transmit power in the considered scenario  compared to the baseline scheme in Figure \ref{fig:hp_pt}.  Besides, the two proposed suboptimal schemes perform close to the trade-off region achieved by the optimal SDP resource allocation scheme.  Furthermore, all trade-off curves are shifted  in the upper-right
direction  if the number of secondary receivers is increased. This is due to the fact that for a larger number of secondary users, there are more idle secondary receivers in the system
harvesting the power radiated by the transmitter which improves the energy harvesting efficiency and the total harvested power. Also, having additional idle secondary receivers means that there are more potential eavesdroppers in the system. Thus, more artificial noise generation is required  for neutralizing information leakage. We note that in all the considered scenarios, the proposed resource allocation schemes are able to guarantee the minimum secrecy data rate requirement of $C_\mathrm{sec}\ge 5.6582$ bit/s/Hz despite the imperfectness of the CSI.

\begin{figure}[t]
\subfigure[Average total harvested power.]{
\includegraphics[scale=0.45]{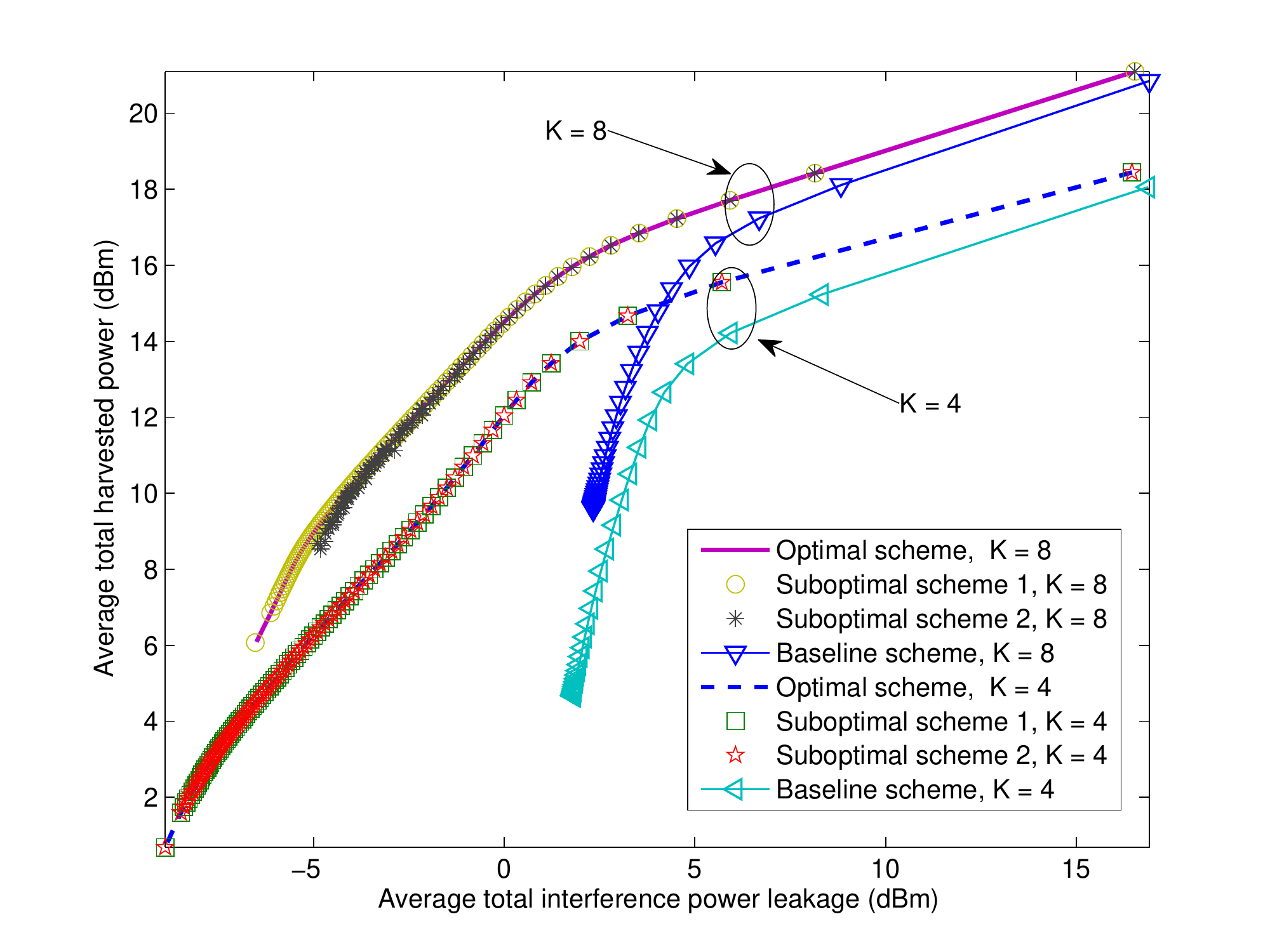}
\label{fig:hp_inteferece}} \subfigure[Average energy harvesting efficiency.]{
\includegraphics[scale=0.45]{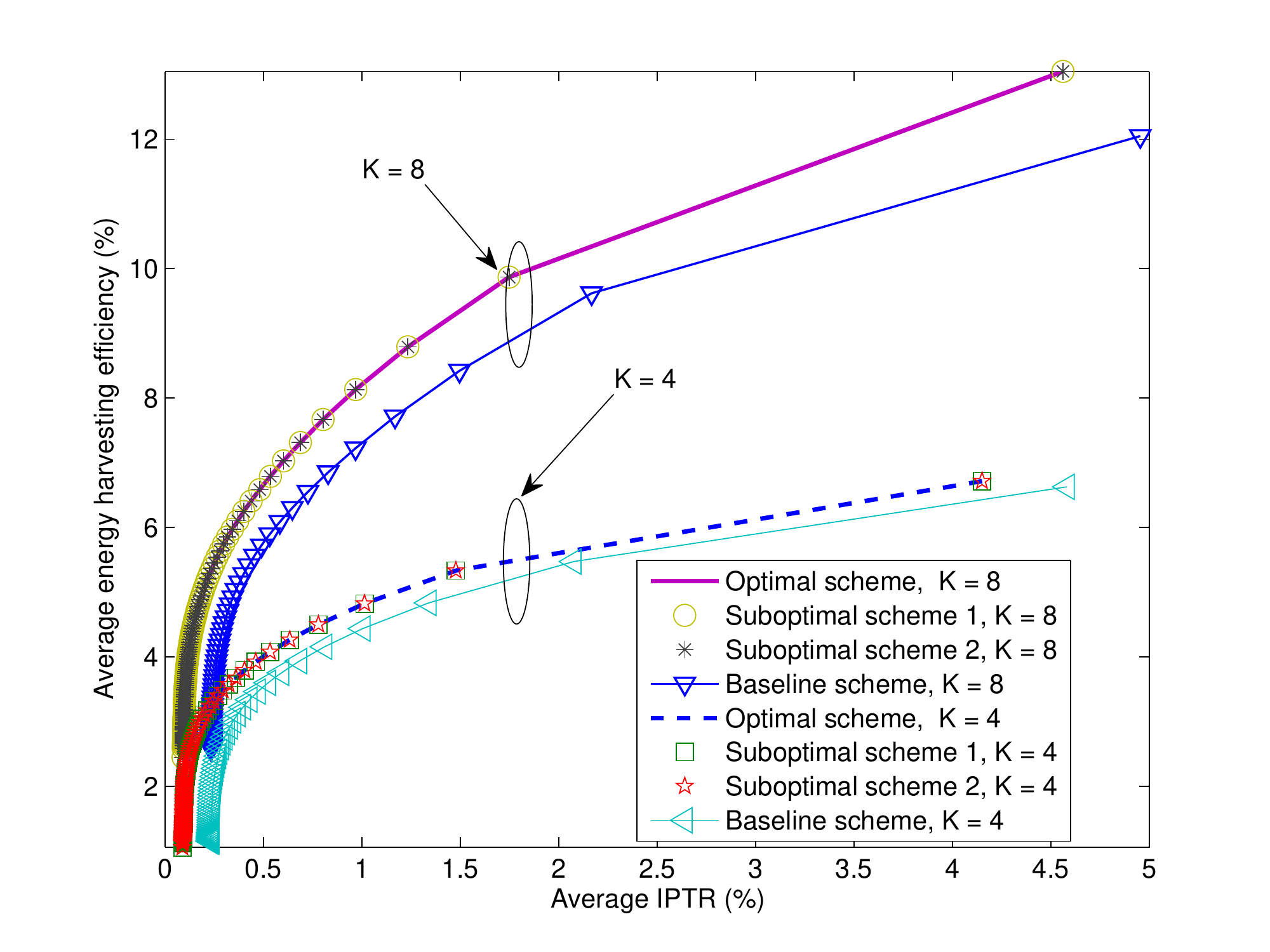}
\label{fig:ee_inteferece} }\caption[System objectives trade-off region for Case II]
{ Average total harvested power and  average energy harvesting efficiency of the secondary system versus the average total interference power leakage and the average  IPTR, respectively,  for different resource allocation schemes and different numbers of secondary receivers, $K$. }\label{fig:ee-inteference-big}
\end{figure}

Figures  \ref{fig:hp_inteferece} and  \ref{fig:ee_inteferece} show the average total harvested power and the average energy harvesting efficiency  of the secondary system versus the average total interference power leakage and the average IPTR, respectively, for different numbers of desired secondary receivers, $K$. The curves in  Figures \ref{fig:hp_inteferece} and  \ref{fig:ee_inteferece} are obtained
by solving Problem 4 for $\lambda_2=0$ and varying the values of $0\le \lambda_p\le 1,\forall p\in\{1,3\}$,  uniformly for a step size of $0.01$ such that $\sum_p \lambda_p=1$ for Case I and Case II, respectively. The average total harvested power and the average energy harvesting efficiency  increase with increasing  average total interference power leakage and   increasing average IPTR, respectively.  This observation indicates that  total harvested power maximization and energy harvesting efficiency maximization are conflicting with total interference power leakage minimization and IPTR minimization, respectively. Besides, the two proposed suboptimal schemes perform close to the trade-off curve achieved by the optimal resource allocation scheme.  Furthermore, all the trade-off curves are shifted in the upper-right direction  as the number of secondary receivers is increased. In fact,  there are more potential eavesdroppers in the system when the number of idle secondary receivers increases. Thus, more artificial noise has to be radiated by the secondary transmitter for guaranteeing communication security which leads to a higher  IPTR and a higher interference power leakage. On the other hand,   the baseline scheme achieves a  smaller trade-off region in both Figures \ref{fig:hp_inteferece} and \ref{fig:ee_inteferece} compared to the  proposed optimal and suboptimal schemes. This performance gap reveals the importance  of the optimization of beamforming matrix $\overline{\mathbf{W}}$  for minimizing the total interference power leakage and the IPTR.

\begin{figure}[t]
\subfigure[Average total interference power leakage.]{
\includegraphics[scale=0.45]{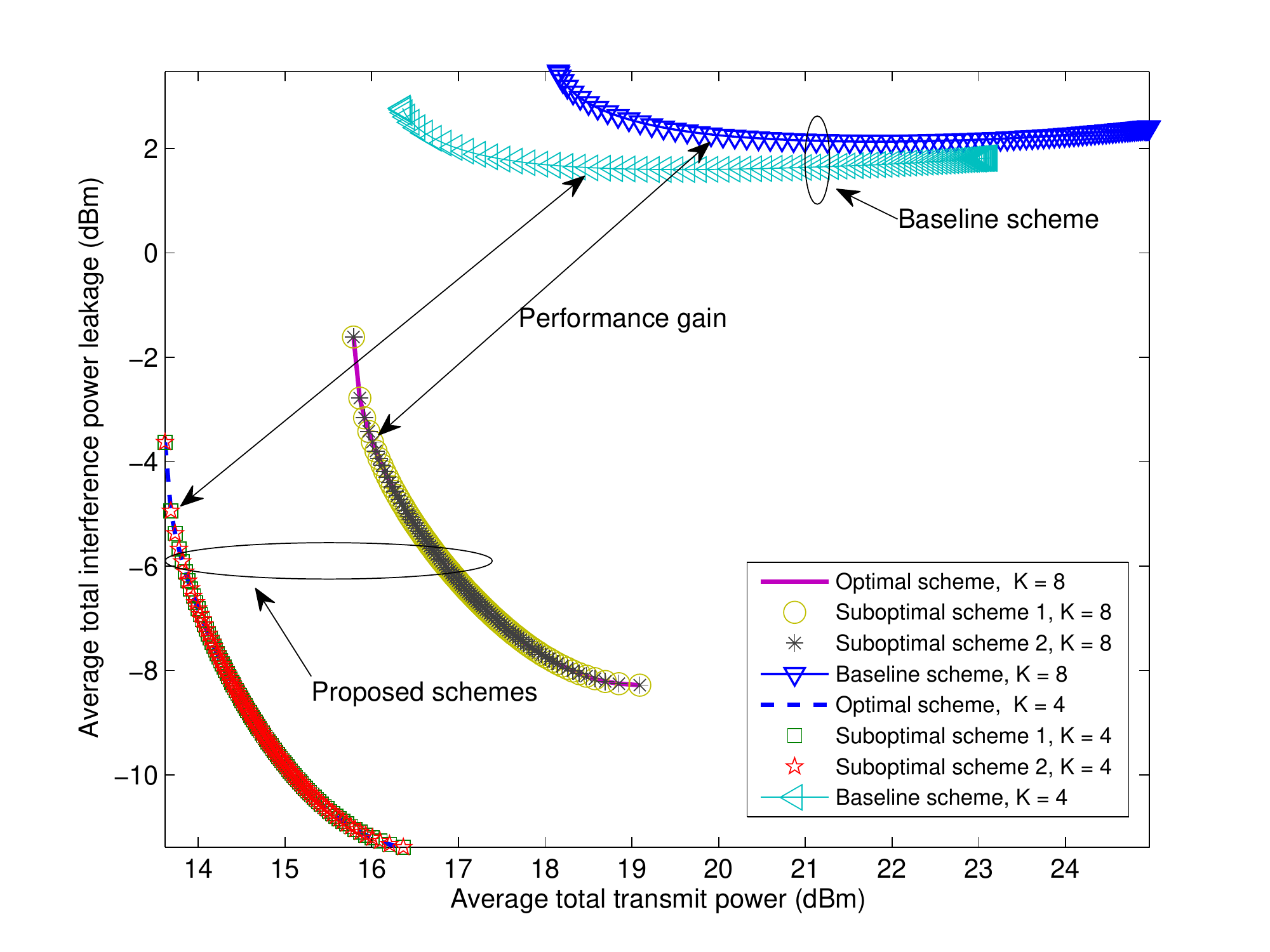}
\label{fig:IP_pt}} \subfigure[Average IPTR.]{
\includegraphics[scale=0.45]{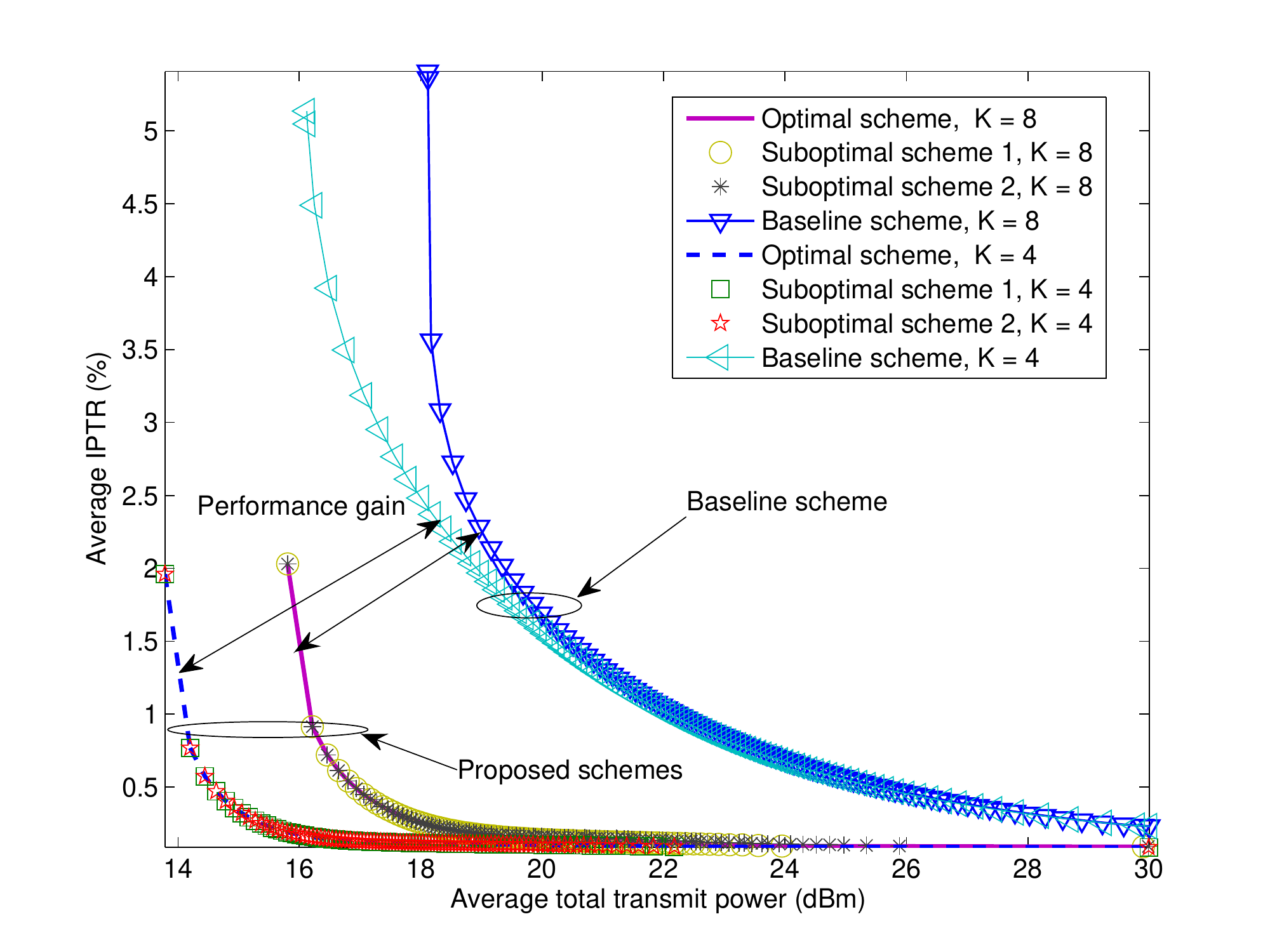}
\label{fig:IPR_pt} }\caption[Average IPTR.]
{Average total interference power leakage and  average IPTR of the secondary system versus the average total transmit power for different resource allocation schemes and different numbers of secondary receivers, $K$.  The double-sided arrows indicate the performance gain achieved by the proposed schemes compared to the baseline scheme. }\label{fig:inteference-pt}
\end{figure}

\subsection{Average Total Interference Power Leakage and Average IPTR}
Figures  \ref{fig:IP_pt} and \ref{fig:IPR_pt}   depict the average total interference power leakage and  the average IPTR  of the secondary system versus the average total transmit power for different numbers of secondary receivers, $K$, respectively. The curves in  Figures \ref{fig:IP_pt} and \ref{fig:IPR_pt}  are obtained by solving Problem 4 for Case I and Case II, respectively,
by setting $\lambda_1=0$ and varying the values of $0\le \lambda_p\le 1,\forall p\in\{2,3\}$,   uniformly  for a step size of $0.01$  such that $\sum_p \lambda_p=1$.  Interestingly, we observe from Figures  \ref{fig:IP_pt} and \ref{fig:IPR_pt} that a higher transmit power may not correspond
to a stronger interference leakage to the primary system or a higher IPTR, if the degrees of freedom offered by the multiple antennas are properly exploited.  Furthermore,  the baseline scheme achieves a significantly worse trade-off  compared to the  proposed optimal and suboptimal schemes, e.g.  $10$ dB more interference leakage in   Figure \ref{fig:IP_pt}. Also, a resource allocation policy that minimizes the total transmit power can only minimize the total interference power leakage simultaneously to a certain extent or vice versa in general. For minimizing  the total interference power leakage, the secondary transmitter sacrifices some degrees of freedom to reduce the  received strengths of both information signal and artificial noise at the primary receivers. Thus, fewer degrees of freedom are available for providing reliable and secure communication to the secondary receivers such that a  higher transmit power is required. In fact, in the proposed optimal scheme,  both  the  beamforming matrix $\mathbf{\overline{W}}$ and the  artificial noise covariance matrix $\mathbf{\overline{V}}$ are jointly optimized for performing resource allocation  based on the CSI of all receivers.  In contrast,  in the baseline scheme, the direction of the beamforming matrix  is fixed which leads to fewer degrees of freedom  for resource allocation. Thus,  the baseline scheme performs worse than the proposed schemes.  On the other hand, for a given required average interference leakage power or average IPTR, increasing the number of secondary receivers induces a higher transmit power in both cases. Indeed,  constraint C2 on communication secrecy becomes more stringent for an increasing number of secondary receivers. In other words, it leads to a smaller feasible solution set for resource allocation optimization. As a result, the  efficiency of the resource allocation schemes in jointly optimizing  the multiple objective functions decreases  for  a larger number of secondary receivers $K$.

\section{Conclusions}\label{sect:conclusion}
In this paper,  we studied the resource allocation algorithm design  for  CR secondary networks with simultaneous wireless power transfer and secure  communication based on a multi-objective optimization framework. We focused on three system design objectives: transmit power minimization, energy harvesting efficiency maximization, and IPTR minimization.  Besides, the proposed multi-objective problem formulation includes total harvested power
maximization and interference power leakage minimization as special cases.  In addition, our problem formulation takes into account the imperfectness of the CSI of the idle secondary receivers and the primary receivers  at the secondary transmitter. By utilizing  the primal and dual optimal solutions of the SDP relaxed problem,  the global optimal solution of the original problem can be constructed. Furthermore, two suboptimal resource allocation schemes were proposed for the case when the solution of the dual problem is unavailable. Simulation results illustrated the performance gains of the proposed  schemes compared to a baseline scheme, and unveiled
the trade-off between the considered system design objectives: (1) A resource allocation policy  minimizing the total transmit power also leads to a low total interference power leakage in general;  (2)  energy harvesting efficiency maximization and transmit power minimization  are conflicting system design objectives; (3)  maximum energy harvesting efficiency is achieved at the expense of high interference power leakage and high transmit power.

\section*{Appendix}
\subsection{Proof of Proposition 1}\setcounter{equation}{41}
The proof is based on the Charnes-Cooper transformation \cite{JR:ken_artifical_noise,JR:linear_fractional}. By applying the change of variables in (\ref{eqn:change_of_variables}) to (\ref{eqn:cross-layer1-flip}), Problem $1$ in (\ref{eqn:cross-layer1-flip}) can be equivalently transformed to
\begin{eqnarray}\label{eqn:proof1}
&&\hspace*{-10mm} \underset{\overline{\mathbf{W}},\overline{\mathbf{V}}\in \mathbb{H}^{N_{\mathrm{T}}},\xi}{\mino}\,\, \,\, \frac{-\sum_{k=1}^{K-1}\eta_k\Tr(\mathbf{G}_k(\overline{\mathbf{W}}+\overline{\mathbf{V}}))}{\Tr(\overline{\mathbf{W}})+\Tr(\overline{\mathbf{V}})}\\
\mbox{s.t.} &&\hspace*{0mm}\mbox{\textoverline{C1}} -\mbox{\textoverline{C5}},\,\, \mbox{\textoverline{C6}:}\,\, \xi > 0,\mbox{\textoverline{C7}:}\,\, \Tr(\overline{\mathbf{W}})+\Tr(\overline{\mathbf{V}})= 1,\mbox{\textoverline{C8}}.\notag
\end{eqnarray}
Now, we show that (\ref{eqn:proof1}) is equivalent to
\begin{eqnarray}\label{eqn:proof2}
&&\hspace*{-12mm} \underset{\overline{\mathbf{W}},\overline{\mathbf{V}}\in \mathbb{H}^{N_{\mathrm{T}}},\xi}{\mino}\,\, \,\, -\sum_{k=1}^{K-1}\eta_k\Tr(\mathbf{G}_k(\overline{\mathbf{W}}+\overline{\mathbf{V}}))\\
 \mbox{s.t.} &&\hspace*{0mm}\mbox{\textoverline{C1}} -\mbox{\textoverline{C5}},\,\,\mbox{\textoverline{C6}:}\,\, \xi \ge 0,\,\, \mbox{\textoverline{C7}:}\,\, \Tr(\overline{\mathbf{W}})+\Tr(\overline{\mathbf{V}})= 1,\mbox{ \textoverline{C8}}.\nonumber
\end{eqnarray}
We denote the optimal solution of (\ref{eqn:proof2}) as $(\overline{\mathbf{W}}^*,\overline{\mathbf{V}}^*,\xi^*)$. If $\xi^*=0$, then $\overline{\mathbf{W}}=\overline{\mathbf{V}}=\mathbf{0}$ according to $\mbox{\textoverline{C3}}$. Yet, this solution cannot satisfy $\mbox{\textoverline{C1}}$ for $\Gamma_{\mathrm{req}}>0$. As a result, without loss of generality and optimality, constraint $\xi>0$ can be replaced by $\xi\ge 0$. The equivalence between transformed Problems 2, 3, and 4 and their original problem formulations can be proved by following a similar approach as above.
%Besides, it can be deduced that $\mbox{\textoverline{C6}}$ is satisfied with equality for the optimal solution, i.e.,
%\begin{eqnarray}
%\Tr(\overline{\mathbf{W}}^*)+\Tr(\overline{\mathbf{V}}^*)=1.
%\end{eqnarray}
%We prove the above by contradiction. Suppose that $\mbox{\textoverline{C6}}$ is satisfied with strict inequality for the optimal solution, i.e.,  $\Tr(\overline{\mathbf{W}}^*)+\Tr(\overline{\mathbf{V}}^*)<1$. Then, we construct a new feasible solution $(\mathbf{A},\mathbf{B},c)=(\delta\overline{\mathbf{W}}^*,\delta\overline{\mathbf{V}}^*,\delta\xi^*)$ where $\delta>1$ such that $\Tr(\overline{\mathbf{W}}^*)+\Tr(\overline{\mathbf{V}}^*)=1$. It can be verified that the point $(\mathbf{A},\mathbf{B},c)$ achieves a lower objective value in (\ref{eqn:proof2}) than $(\overline{\mathbf{W}}^*,\overline{\mathbf{V}}^*,\xi^*)$. Thus, $(\overline{\mathbf{W}}^*,\overline{\mathbf{V}}^*,\xi^*)$ cannot be the optimal solution. As a result, (\ref{eqn:cross-layer-t1}) and  (\ref{eqn:proof2})  are equivalent.

\subsection{Proof of Theorem 1}
 The proof is divided into two parts. In the first part, we investigate the structure of the optimal solution $\mathbf{\overline W}^*$ of the relaxed version of problem (\ref{eqn:cross-layer-t31-LMI}). Then, in the second part, we propose a method to construct a solution $\mathbf{\widetilde \Lambda}\triangleq\{\mathbf{\widetilde I}^{\mathrm{PU}},\mathbf{\widetilde E}^{\mathrm{SU}}, \widetilde \xi,\widetilde\tau,\boldsymbol{\widetilde \gamma},\boldsymbol{\widetilde\delta},\boldsymbol{\widetilde\varphi},\boldsymbol{\widetilde\omega},{\mathbf{\widetilde V}},\mathbf{{ \widetilde W}}\}$ that achieves the same objective value as $\mathbf{\Lambda}^*\triangleq\{\mathbf{I}^{\mathrm{PU}*},\mathbf{E}^{\mathrm{SU}*}, \xi^*,\tau^*,$ $\boldsymbol{\gamma^*}, \boldsymbol{\delta^*},\boldsymbol{\varphi^*},\boldsymbol{\omega^*},\overline{\mathbf{V}}^*,\mathbf{\overline{W}}^*\}$ but admits a rank-one $\mathbf{\widetilde W}$.

It can be shown that the relaxed version of problem (\ref{eqn:cross-layer-t31-LMI}) is jointly convex with respect to the optimization variables and satisfies Slater's constraint qualification. As a result, the Karush-Kuhn-Tucker (KKT) conditions are necessary and sufficient conditions \cite{book:convex} for the optimal solution of the relaxed version of problem (\ref{eqn:cross-layer-t31-LMI}).   The Lagrangian function  of the relaxed version of problem (\ref{eqn:cross-layer-t31-LMI}) is given by
\begin{eqnarray}{\cal
L}&=&\hspace*{-3mm}  \Tr\Big(\big(\mathbf{I}_{N_{\mathrm{T}}}(\alpha+\mu)-\mathbf{Y}-\beta\mathbf{H}\big)\overline{\mathbf{W}}\Big)\notag\\
&-&\hspace*{-3mm}   \sum_{k=1}^{K-1}\Tr\Big( \mathbf{S}_{\mathrm{\overline{C2}}_k}\big(\overline{\mathbf{W}},
\overline{\mathbf{V}},\delta_k\big)\mathbf{D}_{\mathrm{\overline{C2}}_k}\Big) \notag\\
&-&\hspace*{-3mm}\sum_{j=1}^{J}\Tr\Big( \mathbf{S}_{\mathrm{\overline{C3}}_j}\big(\overline{\mathbf{W}},\overline{\mathbf{V}}, \xi,\gamma_j\big) \mathbf{D}_{\mathrm{\overline{C3}}_j}\Big)\notag\\
&-& \hspace*{-3mm}\sum_{j=1}^{J}\Tr\Big( \mathbf{S}_{\mathrm{\overline{C11}}_j}\big(\overline{\mathbf{W}},\overline{\mathbf{V}}, \xi,\omega_j\big)
 \mathbf{D}_{\mathrm{\overline{C11}}_j}\Big)\notag\\
\hspace*{-5mm}&-&\hspace*{-3mm}\sum_{k=1}^{K-1}\Tr\Big( \mathbf{S}_{\mathrm{\overline{C10}}_k}\big(\overline{\mathbf{W}},\overline{\mathbf{V}}, \xi,\varphi_k\big) \mathbf{D}_{\mathrm{\overline{C10}}_k}\Big)+\Omega,
\label{eqn:Lagrangian}
\end{eqnarray}
where $\Omega$ denotes the collection of the terms that only involve variables that are not relevant for the proof.  $\beta,\alpha
\ge 0$, and $\mu$ are the Lagrange multipliers associated with constraints $\mbox{\textoverline{C1}, \textoverline{C4}}$, and $\mbox{\textoverline{C7}}$, respectively. Matrix $\mathbf{Y}\succeq \mathbf{0}$ is the Lagrange multiplier matrix for the semidefinite constraint on matrix $\overline{\mathbf{W}}$ in $\mbox{\textoverline{C4}}$. $\mathbf{D}_{\mathrm{\overline{C2}}_k}\succeq \mathbf{0},\forall k\in\{1,\,\ldots,\,K-1\},$ and $\mathbf{D}_{\mathrm{\overline{C3}}_j}\succeq \mathbf{0},\forall j\in\{1,\,\ldots,\,J\}$, are
the Lagrange multiplier matrices for the maximum tolerable SINRs of the idle secondary receivers and the primary receivers  in $\mbox{\textoverline{C2}}$ and $\mbox{\textoverline{C3}}$, respectively.  $\mathbf{D}_{\mathrm{\overline{C10}}_k}\succeq \mathbf{0},\forall k\in\{1,\,\ldots,\,K-1\},$ and $\mathbf{D}_{\mathrm{\overline{C11}}_j}\succeq \mathbf{0},\forall j\in\{1,\,\ldots,\,J\}$, are
the Lagrange multiplier matrices associated with constraints $\mbox{\textoverline{C10}}$ and $\mbox{\textoverline{C11}}$, respectively.
In the following, we focus on the KKT conditions related to the optimal $\mathbf{\overline{W}}^*$:
\begin{eqnarray}\label{eqn:KKT}\mathbf{Y}^*,\mathbf{D}^*_{\mathrm{\overline{C2}}_k},\mathbf{D}^*_{\mathrm{\overline{C3}}_j},\mathbf{D}^*_{\mathrm{\overline{C10}}_k},
\mathbf{D}^*_{\mathrm{\overline{C11}}_j}
\hspace*{-2mm}&\succeq&\hspace*{-2mm} \mathbf{0},\alpha^*,\beta^*\ge 0,\mu^*,\\
 \mathbf{Y^*\overline{\mathbf{W}}^*}\hspace*{-2mm}&=&\hspace*{-2mm}\mathbf{0}, \label{eqn:KKT-complementarity}\\
\nabla_{\overline{\mathbf{W}}^*}{\cal
L}\hspace*{-2mm}&=&\hspace*{-2mm}\mathbf{0}, \label{eqn:KKT-gradient}
\end{eqnarray}
where $\mathbf{Y}^*,\mathbf{D}^*_{\mathrm{\overline{C2}}_k},\mathbf{D}^*_{\mathrm{\overline{C3}}_j},\mathbf{D}^*_{\mathrm{\overline{C10}}_k},
\mathbf{D}^*_{\mathrm{\overline{C11}}_j},\mu^*,\beta^*$, and $\alpha^*$ are the optimal Lagrange multipliers for the dual problem of (\ref{eqn:cross-layer-t31-LMI}). From the complementary slackness condition in (\ref{eqn:KKT-complementarity}),  we observe that the columns of $\overline{\mathbf{W}}^*$ are required to lie in the null space of $\mathbf{Y}^*$ for $\overline{\mathbf{W}}^*\ne \mathbf{0}$.  Thus, we study  the  composition of $\mathbf{Y}^*$ to obtain the structure of  $\overline{\mathbf{W}}^*$. The KKT condition in (\ref{eqn:KKT-gradient}) can be expressed as
\begin{eqnarray}\label{eqn:KKT-gradient-equivalent}
&&\mathbf{Y}^*+\beta^*\mathbf{H}\notag\\
&=&\mathbf{I}_{N_{\mathrm{T}}}(\mu^*+\alpha^*)+\sum_{k=1}^{K-1} \mathbf{U}_{\mathbf{g}_k}\Big(\frac{\mathbf{D}^*_{\mathrm{\overline{C2}}_k}}{\Gamma_{\mathrm{tol}_k}}-
\mathbf{D}_{\mathrm{\overline{C10}}_k}^*
\Big)\mathbf{U}_{\mathbf{g}_k}^H \notag\\
&+&\sum_{j=1}^J\mathbf{U}_{\mathbf{l}_j}\Big(\frac{\mathbf{D}^*_{\mathrm{\overline{C3}}_j}}
{\Gamma_{\mathrm{tol}_j}^{\mathrm{PU}}}+\mathbf{D}_{\mathrm{\overline{C11}}_j}^*
\Big)\mathbf{U}_{\mathbf{l}_j}^H.
\end{eqnarray}
For notational simplicity, we define
\begin{eqnarray}\label{eqn:A_matrix}
\mathbf{A}^*&=&\mathbf{I}_{N_{\mathrm{T}}}(\mu^*+\alpha^*)+\sum_{k=1}^{K-1} \mathbf{U}_{\mathbf{g}_k}\Big(\frac{\mathbf{D}^*_{\mathrm{\overline{C2}}_k}}{\Gamma_{\mathrm{tol}_k}}-\mathbf{D}_{\mathrm{\overline{C10}}_k}^*
\Big)\mathbf{U}_{\mathbf{g}_k}^H \notag\\ &+&\sum_{j=1}^J\mathbf{U}_{\mathbf{l}_j}\Big(\frac{\mathbf{D}^*_{\mathrm{\overline{C3}}_j}}{\Gamma_{\mathrm{tol}_j}^{\mathrm{PU}}}+\mathbf{D}_{\mathrm{\overline{C11}}_j}^*
\Big)\mathbf{U}_{\mathbf{l}_j}^H.
\end{eqnarray}
%Then, we post-multiply both sides of (\ref{eqn:KKT-gradient-equivalent}) by $\overline{\mathbf{W}}^*$ which yields
%\begin{eqnarray}\label{eqn:pre-rank_equality}
%\hspace*{-6mm}&&
%\mathbf{A}^*\overline{\mathbf{W}}^*=\Big(\mathbf{Y}^*+\beta^*\mathbf{H}\Big)\overline{\mathbf{W}}^*\stackrel{(a)}{=}\beta^*\mathbf{H} \overline{\mathbf{W}}^*,
%\end{eqnarray}
%where $(a)$ is due to the complementary slackness condition in (\ref{eqn:KKT-complementarity}).
Besides,  there exists at least one optimal solution with $\beta^*>0$, i.e., constraint $\mbox{\textoverline{C1}}$ is satisfied with equality. Suppose that for the optimal solution,  constraint $\mbox{\textoverline{C1}}$ is satisfied with strict inequality, i.e., $\frac{\Tr(\mathbf{H}\overline{\mathbf{W}}^*)}{\Tr(\mathbf{H}
\overline{\mathbf{V}}^*)+\sigma_\mathrm{z}^2\xi^*} > \Gamma_{\mathrm{req}}$. Then, we can replace $\xi^*$ with  $\overline\xi^*=\xi^* c$ for $c>1$ such that $\mbox{\textoverline{C1}}$ is satisfied with equality. We note that the new solution $\overline\xi^*$ not only satisfies all the constraints, but also provides a larger feasible solution set for minimizing $\tau$, cf. constraints $\mbox{\textoverline{C4}}$ and $\mbox{\textoverline{C9}b}$. As a result, there always exist at least one optimal solution such that constraint $\mbox{\textoverline{C1}}$ is satisfied with equality. In order to obtain the optimal solution in practice, we can replace the inequality ``$\ge$"with equality ``$=$"  in $\mbox{\textoverline{C2}}$ without loss of optimality.
 From (\ref{eqn:KKT-gradient-equivalent}) and (\ref{eqn:A_matrix}), we can express the Lagrange multiplier matrix $\mathbf{Y}^*$ as
\begin{eqnarray}\label{eqn:Y}
\mathbf{Y}^*=\mathbf{A}^*-\beta^*\mathbf{H},
\end{eqnarray}
where $\beta^*\mathbf{H}$ is a rank-one matrix since $\beta^*>0$. Without loss of generality, we define $r=\Rank(\mathbf{A}^*)$ and the orthonormal basis of the null space of $\mathbf{A}^*$ as $\mathbf{\mathbf{\Upsilon}}\in\mathbb{C}^{N_{\mathrm{T}}\times (N_{\mathrm{T}}-r)}$ such that $\mathbf{A}^*\mathbf{\Upsilon}=\mathbf{0}$ and $\Rank(\mathbf{\Upsilon})=N_{\mathrm{T}}-r$. Let ${\boldsymbol \phi}_t\in \mathbb{C}^{N_{\mathrm{T}}\times 1}$, $1\le t\le N_{\mathrm{T}}-r$, denote the $t$-th column  of $\mathbf{\Upsilon}$. By exploiting \cite[Proposition 4.1]{JR:rui_zhang},  it can be shown that $\mathbf{H}\mathbf{\Upsilon}=\mathbf{0}$ and we can express the optimal solution of  $\overline{\mathbf{W}}^*$  as
\begin{eqnarray}\label{eqn:general_structure}
\overline{\mathbf{W}}^*=\sum_{t=1}^{N_{\mathrm{T}}-r} \psi_t {\boldsymbol \phi}_t {\boldsymbol \phi}_t^H  + \underbrace{f\mathbf{u}\mathbf{u}^H}_{\mbox{Rank-one}},
\end{eqnarray}
where $\psi_t\ge0, \forall t\in\{1,\ldots,N_{\mathrm{T}}-r\},$ and $f>0$ are positive scalars and $\mathbf{u}\in \mathbb{C}^{N_{\mathrm{T}}\times 1}$, $\norm{\mathbf{u}}=1$, satisfies $\mathbf{u}^H\mathbf{\Upsilon}=\mathbf{0}$.

In the second part of the proof, for $\Rank(\overline{\mathbf{W}}^*)>1$,  we construct another solution $\mathbf{\widetilde \Lambda}\triangleq\{\mathbf{\widetilde I}^{\mathrm{PU}},\mathbf{\widetilde E}^{\mathrm{SU}}, \widetilde \xi,\widetilde\tau,$ $\boldsymbol{\widetilde \gamma,\widetilde\delta,\widetilde\varphi,\widetilde\omega},{\mathbf{\widetilde V}},\mathbf{{ \widetilde W}}\}$ based on  (\ref{eqn:general_structure}).
Let
\begin{eqnarray}\label{eqn:rank-one-structure}\mathbf{\widetilde W}\hspace*{-2mm}&=&\hspace*{-2mm}f\mathbf{u}\mathbf{u}^H=\overline{\mathbf{W}}^*-\sum_{t=1}^{N_{\mathrm{T}}-r} \psi_t {\boldsymbol \phi}_t {\boldsymbol \phi}_t^H,\\
 \mathbf{\widetilde V}&=&\overline{\mathbf{ V}}^*+\sum_{t=1}^{N_{\mathrm{T}}-r} \psi_t {\boldsymbol \phi}_t {\boldsymbol \phi}_t^H,\\
\mathbf{\widetilde I}^{\mathrm{PU}}\hspace*{-2mm}&=&\hspace*{-2mm}\mathbf{ I}^{\mathrm{PU}*},\,\,\mathbf{\widetilde E}^{\mathrm{SU}}=\mathbf{ E}^{\mathrm{SU}*},\,\, \widetilde\xi=\xi^*,\,\,\\
\widetilde\tau\hspace*{-2mm}&=&\hspace*{-2mm}\tau^*,\,\,\boldsymbol{\widetilde\gamma}=\boldsymbol{\gamma}^*,\,\,
\boldsymbol{\widetilde\delta}=\boldsymbol{\delta}^*,\,\,
\boldsymbol{\widetilde\varphi}=\boldsymbol{\widetilde\varphi}^*,\,\,\boldsymbol{\widetilde\omega}=\boldsymbol{\omega}^*.
\end{eqnarray}
Then, we substitute the constructed solution $\mathbf{\widetilde \Lambda}$ into the objective function and the constraints in (\ref{eqn:cross-layer-t31-LMI}) which yields  the equations in \eqref{eqn:equivalent_objective} on the top of next page.

\begin{figure*}[ht]\setcounter{mytempeqncnt}{\value{equation}}
\setcounter{equation}{55}
\begin{eqnarray}\label{eqn:equivalent_objective}
 &&\hspace*{-5mm}\mbox{Objective value:}\quad \widetilde\tau=\tau^* \\
 \label{eqn:C1_new_solution}
 \mbox{\textoverline{C1}:}&&\hspace*{-5mm} \frac{\Tr(\mathbf{\mathbf{H}\widetilde W})}{\Tr(\mathbf{H}\mathbf{\widetilde V})+\widetilde \xi\sigma_{\mathrm{z}}^2}= \frac{\Tr\big(\mathbf{H}(\overline{\mathbf{W}}^*-\sum_{t=1}^{N_{\mathrm{T}}-r} \psi_t {\boldsymbol \phi}_t {\boldsymbol \phi}_t^H)\big)}{\Tr(\mathbf{H}(\overline{\mathbf{ V}}^*+\sum_{t=1}^{N_{\mathrm{T}}-r} \psi_t {\boldsymbol \phi}_t {\boldsymbol \phi}_t^H))+ \xi^*\sigma_{\mathrm{z}}^2}=  \frac{\Tr(\overline{\mathbf{W}}^*\mathbf{H})}{\Tr(\mathbf{H}\overline{\mathbf{ V}}^*)+\xi^*\sigma_{\mathrm{z}}^2}\ge  \Gamma_{\mathrm{req}},\notag\\
  \mbox{\textoverline{C2}:}\notag&&\hspace*{-5mm} \mathbf{S}_{\mathrm{\overline{C2}}_k}(\mathbf{\widetilde W},\mathbf{\widetilde V},\widetilde\xi,\widetilde\delta_k)\succeq\mathbf{S}_{\mathrm{\overline{C2}}_k}(\overline{\mathbf{W}}^*,\overline{\mathbf{V}}^*, \xi^*,\delta_k^*)\notag\\
+&& \notag \hspace*{-5mm}\mathbf{U}_{\mathbf{g}_k}^H\Big[\sum_{t=1}^{N_{\mathrm{T}}-r} \psi_t {\boldsymbol \phi}_t {\boldsymbol \phi}_t^H \Big] \mathbf{U}_{\mathbf{g}_k}\Big(1+\frac{1}{\Gamma_{\mathrm{tol}_k}}\Big) \succeq \mathbf{0},\forall k\in\{1,\ldots,K-1\},\\
 \mbox{\textoverline{C3}:}\notag&&\hspace*{-5mm} \mathbf{S}_{\mathrm{\overline{C3}}_j}(\mathbf{\widetilde W},\mathbf{\widetilde V},\widetilde\xi,\widetilde\gamma_j)\succeq\mathbf{S}_{\mathrm{\overline{C3}}_j}(\overline{\mathbf{W}}^*,\overline{\mathbf{V}}^*, \xi^*,\gamma_j^*)\\
+&&  \hspace*{-5mm}\mathbf{U}_{\mathbf{l}_j}^H\Big[\sum_{t=1}^{N_{\mathrm{T}}-r} \psi_t{\boldsymbol \phi}_t {\boldsymbol \phi}_t^H \Big] \mathbf{U}_{\mathbf{l}_j}\Big(1+\frac{1}{\Gamma_{\mathrm{tol}_j}^{\mathrm{PU}}}\Big) \succeq \mathbf{0},\forall j\in\{1,\ldots,J\},\notag\\
\mbox{\textoverline{C4}:}&&\hspace*{-5mm}\Tr({\mathbf{\widetilde W}})  +\Tr({\mathbf{\widetilde V}})=\Tr(\overline{\mathbf{W}}^*)  +\Tr(\overline{\mathbf{V}}^*)\le P_{\max}\widetilde\xi,\notag \\
\mbox{\textoverline{C5}:}&&\hspace*{-5mm} \mathbf{\widetilde{W}},\mathbf{\widetilde{V}}\succeq \mathbf{0},\quad\mbox{\textoverline{C6}:}\,\, \widetilde\xi \ge 0, \notag\quad
\mbox{\textoverline{C7}:}\Tr({\mathbf{\widetilde W}})  +\Tr({\mathbf{\widetilde V}})=\Tr(\overline{\mathbf{W}}^*)  +\Tr(\overline{\mathbf{V}}^*)= 1,\notag \\
\mbox{\textoverline{C9}a:}&&\hspace*{-5mm}\frac{\lambda_1}{\abs{F_1^*}} \Big(\sum_{k=1}^{K-1}\widetilde E_{k}^{\mathrm{SU}}\hspace*{-0.5mm}-\hspace*{-0.5mm}F_1^*\Big)\le \widetilde\tau,\,\,\mbox{\textoverline{C9}b: }\frac{\lambda_2} {\abs{F_2^*}} \Big(\frac{1}{\widetilde\xi}\hspace*{-0.5mm}-\hspace*{-0.5mm}F_2^*\Big)\le \widetilde\tau,\,\, \mbox{\textoverline{C9}c: }\frac{\lambda_3}{\abs{F_3^*}}  \Big(\sum_{j=1}^{J}\widetilde I_{j}^{\mathrm{PU}}\hspace*{-0.5mm}-\hspace*{-0.5mm}F_3^*\Big)\le \widetilde\tau,\notag\\
\mbox{\textoverline{C10}:}&&\hspace*{-5mm}\mathbf{S}_{\mathrm{\overline{C10}}_k}({\mathbf{\widetilde W}},{\mathbf{\widetilde V}}, \widetilde E_{k}^{\mathrm{SU}},\widetilde \varphi_k)=\mathbf{S}_{\mathrm{\overline{C10}}_k}(\overline{\mathbf{W}}^*,\overline{\mathbf{V}}^*,  E_{k}^{\mathrm{SU}*}, \varphi_k^*)\succeq \mathbf{0},\forall k,\notag\\
\mbox{\textoverline{C11}:}&&\hspace*{-5mm}\mathbf{S}_{\mathrm{\overline{C11}}_j}({\mathbf{\widetilde W}},{\mathbf{\widetilde V}}, \widetilde I_{j}^{\mathrm{PU}},\widetilde\omega_j)=\mathbf{S}_{\mathrm{\overline{C11}}_j}(\overline{\mathbf{W}}^*,\overline{\mathbf{V}}^*, I_{j}^{\mathrm{PU}*},\omega_j^*)\succeq \mathbf{0},\forall j,\notag\\
&&\hspace*{-1.3cm}\mbox{\textoverline{C12}: } \widetilde\delta_k\ge 0,\forall k,\,\, \,\,\mbox{\textoverline{C13}: } \widetilde\gamma_j\ge 0,\forall j,\,\, \,\,\mbox{\textoverline{C14}: } \widetilde\varphi_k\ge 0,\forall k,\,\, \,\, \mbox{\textoverline{C15}: } \widetilde\omega_j\ge 0,\forall j.\label{eqn:C10_new_solution}\notag
\end{eqnarray}\hrulefill
\end{figure*}
It can be seen from (\ref{eqn:equivalent_objective}) that the constructed solution set $\mathbf{\widetilde \Lambda}$ achieves the same optimal value as the optimal solution  $\mathbf{\Lambda}^*$ while satisfying all the constraints. Thus, $\mathbf{\widetilde \Lambda}$ is also an optimal solution of (\ref{eqn:cross-layer-t31-LMI}). Besides,  the constructed beamforming matrix $\mathbf{\widetilde W}$ is a rank-one matrix, i.e.,     $\Rank(\mathbf{\widetilde W})=1$. On the other hand, we can obtain the values of $f$ and $\psi_t$ in (\ref{eqn:rank-one-structure}) by substituting the variables in (\ref{eqn:rank-one-structure}) into  the relaxed version of (\ref{eqn:cross-layer-t31-LMI}) and solving the resulting convex optimization problem for $f$ and $\psi_t$.


\begin{thebibliography}{10}
\providecommand{\url}[1]{#1}
\csname url@samestyle\endcsname
\providecommand{\newblock}{\relax}
\providecommand{\bibinfo}[2]{#2}
\providecommand{\BIBentrySTDinterwordspacing}{\spaceskip=0pt\relax}
\providecommand{\BIBentryALTinterwordstretchfactor}{4}
\providecommand{\BIBentryALTinterwordspacing}{\spaceskip=\fontdimen2\font plus
\BIBentryALTinterwordstretchfactor\fontdimen3\font minus
  \fontdimen4\font\relax}
\providecommand{\BIBforeignlanguage}[2]{{%
\expandafter\ifx\csname l@#1\endcsname\relax
\typeout{** WARNING: IEEEtran.bst: No hyphenation pattern has been}%
\typeout{** loaded for the language `#1'. Using the pattern for}%
\typeout{** the default language instead.}%
\else
\language=\csname l@#1\endcsname
\fi
#2}}
\providecommand{\BIBdecl}{\relax}
\BIBdecl

\bibitem{CN:Kwan_PIMRC2013}
D.~W.~K. Ng, L.~Xiang, and R.~Schober, ``{Multi-Objective Beamforming for
  Secure Communication in Systems with Wireless Information and Power
  Transfer},'' in \emph{Proc. IEEE Personal, Indoor and Mobile Radio Commun.
  Sympos.}, Sep. 2013, pp. 7 -- 12.

\bibitem{Report:CR}
``{Facilitating Opportunities for Flexible, Efficient, and Reliable Spectrum
  Use Employing Cognitive Radio Technologies},'' Federal Communications
  Commission, Tech. Rep., 2002, {FCC 02-155}.

\bibitem{JR:CR_overview}
Y.-C. Liang, K.-C. Chen, G.~Li, and P.~Mahonen, ``{Cognitive Radio Networking
  and Communications: An Overview},'' \emph{IEEE Trans. Veh. Technol.},
  vol.~60, pp. 3386--3407, Sep. 2011.

\bibitem{JR:CR_tutorials}
X.~Chen, H.-H. Chen, and W.~Meng, ``{Cooperative Communications for Cognitive
  Radio Networks -- 2014; From Theory to Applications},'' \emph{IEEE Commun.
  Surveys Tuts.}, vol.~16, pp. 1180--1192, Mar. 2014.

\bibitem{JR:CR_WSN}
O.~Akan, O.~Karli, and O.~Ergul, ``{Cognitive Radio Sensor Networks},''
  \emph{IEEE Netw.}, vol.~23, pp. 34--40, Jul. 2009.

\bibitem{JR:CR_WSN2}
X.~Lu, P.~Wang, D.~Niyato, and E.~Hossain, ``{Dynamic Spectrum Access in
  Cognitive Radio Networks with RF Energy Harvesting},'' \emph{IEEE Wireless
  Commun.}, vol.~21, pp. 102--110, Jun. 2014.

\bibitem{JR:CR_sensing1}
G.~Ganesan and Y.~Li, ``{Cooperative Spectrum Sensing in Cognitive Radio, Part
  I: Two User Networks},'' \emph{IEEE Trans. Wireless Commun.}, vol.~6, pp.
  2204--2213, Jun. 2007.

\bibitem{JR:CR_sensing2}
Y.-C. Liang, Y.~Zeng, E.~Peh, and A.~T. Hoang, ``{Sensing-Throughput Tradeoff
  for Cognitive Radio Networks},'' \emph{IEEE Trans. Wireless Commun.}, vol.~7,
  pp. 1326--1337, Apr. 2008.

\bibitem{CN:CR_beamforming_1}
H.~Islam, Y.-C. Liang, and A.~T. Hoang, ``{Joint Beamforming and Power Control
  in the Downlink of Cognitive Radio Networks},'' in \emph{Proc. IEEE Wireless
  Commun. and Networking Conf.}, Apr. 2007, pp. 21--26.

\bibitem{JR:CR_beamforming_1}
L.~Zhang, Y.-C. Liang, Y.~Xin, and H.~Poor, ``{Robust Cognitive Beamforming
  with Partial Channel State Information},'' \emph{IEEE Trans. Wireless
  Commun.}, vol.~8, pp. 4143--4153, Aug. 2009.

\bibitem{JR:CR_beamforming_2}
G.~Zheng, K.-K. Wong, and B.~Ottersten, ``{Robust Cognitive Beamforming With
  Bounded Channel Uncertainties},'' \emph{IEEE Trans. Signal Process.},
  vol.~57, pp. 4871--4881, Dec. 2009.

\bibitem{JR:antenna_selection}
P.~Yeoh, M.~Elkashlan, N.~Yang, D.~da~Costa, and T.~Duong, ``{Unified Analysis
  of Transmit Antenna Selection in MIMO Multirelay Networks},'' \emph{IEEE
  Trans. Veh. Technol.}, vol.~62, pp. 933--939, Feb 2013.

\bibitem{JR:antenna_selection2}
P.~L. Yeoh, M.~Elkashlan, T.~Duong, N.~Yang, and D.~da~Costa, ``{Transmit
  Antenna Selection for Interference Management in Cognitive Relay Networks},''
  \emph{IEEE Trans. Veh. Technol.}, vol.~63, pp. 3250--3262, Sep. 2014.

\bibitem{Powercast}
{Powercast Coporation}, ``{RF Energy Harvesting and Wireless Power for
  Low-Power Applications}, year = {2011}, url =
  {http://www.mouser.com/pdfdocs/Powercast-Overview-2011-01-25.pdf},.''

\bibitem{JR:harvesting_single_user}
J.~Yang and S.~Ulukus, ``{Optimal Packet Scheduling in an Energy Harvesting
  Communication System},'' \emph{IEEE Trans. Commun.}, vol.~60, pp. 220--230,
  Jan. 2012.

\bibitem{Krikidis2014}
I.~Krikidis, S.~Timotheou, S.~Nikolaou, G.~Zheng, D.~W.~K. Ng, and R.~Schober,
  ``Simultaneous {W}ireless {I}nformation and {P}ower {T}ransfer in {M}odern
  {C}ommunication {S}ystems,'' \emph{IEEE Commun. Mag.}, vol.~52, no.~11, pp.
  104--110, Nov. 2014.

\bibitem{Ding2014}
Z.~Ding, C.~Zhong, D.~W.~K. Ng, M.~Peng, H.~A. Suraweera, R.~Schober, and H.~V.
  Poor, ``{Application of Smart Antenna Technologies in Simultaneous Wireless
  Information and Power Transfer},'' \emph{IEEE Commun. Magazine}, vol.~53,
  no.~4, pp. 86--93, Apr. 2015.

\bibitem{JR:SWIPT_mag}
X.~Chen, Z.~Zhang, H.-H. Chen, and H.~Zhang, ``{Enhancing Wireless Information
  and Power Transfer by Exploiting Multi-Antenna Techniques},'' \emph{IEEE
  Commun. Magazine}, no.~4, pp. 133--141, Apr. 2015.

\bibitem{CN:WIPT_fundamental}
L.~Varshney, ``{Transporting Information and Energy Simultaneously},'' in
  \emph{Proc. IEEE Intern. Sympos. on Inf. Theory}, Jul. 2008, pp. 1612 --1616.

\bibitem{CN:Shannon_meets_tesla}
P.~Grover and A.~Sahai, ``{Shannon Meets Tesla: Wireless Information and Power
  Transfer},'' in \emph{Proc. IEEE Intern. Sympos. on Inf. Theory}, Jun. 2010,
  pp. 2363 --2367.

\bibitem{JR:WIPT_bruno}
J.~Park and B.~Clerckx, ``{Joint Wireless Information and Energy Transfer in a
  Two-User MIMO Interference Channel},'' \emph{IEEE Trans. Wireless Commun.},
  vol.~12, pp. 4210--4221, Aug. 2013.

\bibitem{JR:WIP_receiver}
X.~Zhou, R.~Zhang, and C.~Ho, ``{Wireless Information and Power Transfer:
  Architecture Design and Rate-Energy Tradeoff},'' \emph{IEEE Trans. Commun.},
  vol.~61, pp. 4754 -- 4767, Nov. 2013.

\bibitem{JR:Kai_bin}
K.~Huang and E.~Larsson, ``{Simultaneous Information and Power Transfer for
  Broadband Wireless Systems},'' \emph{IEEE Trans. Signal Process.}, vol.~61,
  pp. 5972--5986, Dec. 2013.

\bibitem{JR:WIPT_fullpaper}
D.~W.~K. Ng, E.~S. Lo, and R.~Schober, ``{Wireless Information and Power
  Transfer: Energy Efficiency Optimization in OFDMA Systems},'' \emph{IEEE
  Trans. Wireless Commun.}, vol.~12, pp. 6352 -- 6370, Dec. 2013.

\bibitem{JR:EE-sec}
X.~Chen and L.~Lei, ``{Energy-Efficient Optimization for Physical Layer
  Security in Multi-Antenna Downlink Networks with QoS Guarantee},'' \emph{IEEE
  Commun. Lett.}, vol.~17, pp. 637--640, Apr. 2013.

\bibitem{JR:EE_secrecy}
D.~W.~K. Ng, E.~S. Lo, and R.~Schober, ``{Energy-Efficient Resource Allocation
  for Secure OFDMA Systems},'' \emph{IEEE Trans. Veh. Technol.}, vol.~6, pp.
  2572--2585, Jul. 2012.

\bibitem{JR:ken_artifical_noise}
W.-C. Liao, T.-H. Chang, W.-K. Ma, and C.-Y. Chi, ``{QoS-Based Transmit
  Beamforming in the Presence of Eavesdroppers: An Optimized
  Artificial-Noise-Aided Approach},'' \emph{IEEE Trans. Signal Process.},
  vol.~59, pp. 1202--1216, Mar. 2011.

\bibitem{JR:CR_phy}
Y.~Pei, Y.-C. Liang, L.~Zhang, K.~Teh, and K.~H. Li, ``{Secure Communication
  Over MISO Cognitive Radio Channels},'' \emph{IEEE Trans. Wireless Commun.},
  vol.~9, pp. 1494--1502, Apr. 2010.

\bibitem{JR:Kwan_secure_imperfect}
D.~W.~K. Ng, E.~S. Lo, and R.~Schober, ``{Robust Beamforming for Secure
  Communication in Systems with Wireless Information and Power Transfer},''
  \emph{IEEE Trans. Wireless Commun.}, vol.~13, pp. 4599--4615, Aug. 2014.

\bibitem{JR:rui_zhang}
L.~Liu, R.~Zhang, and K.-C. Chua, ``{Secrecy Wireless Information and Power
  Transfer with MISO Beamforming},'' \emph{IEEE Trans. Signal Process.},
  vol.~62, pp. 1850--1863, Apr. 2014.

\bibitem{JR:Kwan_SEC_DAS}
\BIBentryALTinterwordspacing
D.~W.~K. Ng and R.~Schober, ``{Secure and Green {SWIPT} in Distributed Antenna
  Networks with Limited Backhaul Capacity},'' \emph{\emph{submitted to IEEE
  Trans. Wireless Commun.}}, 2014. [Online]. Available:
  \url{http://arxiv.org/abs/1410.3065}
\BIBentrySTDinterwordspacing

\bibitem{CN:kwan_vicky}
S.~Leng, D.~W.~K. Ng, and R.~Schober, ``{Power Efficient and Secure Multiuser
  Communication Systems with Wireless Information and Power Transfer},'' in
  \emph{Proc. IEEE Intern. Commun. Conf.}, Jun. 2014.

\bibitem{CN:massive_MIMO_security_SWIPT}
X.~Chen, J.~Chen, and T.~Liu, ``{Secure Wireless Information and Power Transfer
  in Large-Scale MIMO Relaying Systems with Imperfect CSI},'' in \emph{Proc.
  IEEE Global Telecommun. Conf.}, Dec. 2014, pp. 4131--4136.

\bibitem{JR:PHY_CR}
M.~Elkashlan, L.~Wang, T.~Duong, G.~Karagiannidis, and A.~Nallanathan, ``{On
  the Security of Cognitive Radio Networks},'' \emph{IEEE Trans. Veh.
  Technol.}, vol.~PP, no.~99, 2014.

\bibitem{MCU}
\BIBentryALTinterwordspacing
{Texas Instruments Incorporated}. (2012) Msp430f2274. [Online]. Available:
  \url{http://www.ti.com/product/msp430f2274}
\BIBentrySTDinterwordspacing

\bibitem{CN:karama}
K.~Hamdi, W.~Zhang, and K.~Letaief, ``{Joint Beamforming and Scheduling in
  Cognitive Radio Networks},'' in \emph{Proc. IEEE Global Telecommun. Conf.},
  Nov. 2007, pp. 2977--2981.

\bibitem{Report:Wire_tap}
A.~D. Wyner, ``{The Wire-Tap Channel},'' Tech. Rep., Oct 1975.

\bibitem{CN:hybrid_energy_source}
F.~Philipp, P.~Zhao, F.~Samman, M.~Glesner, K.~Dassanayake, S.~Maheswararajah,
  and S.~Halgamuge, ``{Adaptive Wireless Sensor Networks Powered by Hybrid
  Energy Harvesting for Environmental Monitoring},'' in \emph{IEEE Intern.
  Conf. on Inf. and Autom. for Sustainability}, Sep. 2012, pp. 285 --289.

\bibitem{JR:hybrid_energy_source}
Y.~Tan and S.~Panda, ``{Energy Harvesting From Hybrid Indoor Ambient Light and
  Thermal Energy Sources for Enhanced Performance of Wireless Sensor Nodes},''
  \emph{IEEE Trans. on Ind. Electron.}, vol.~58, pp. 4424--4435, Sep. 2011.

\bibitem{JR:Robust_error_models1}
G.~Zheng, K.~K. Wong, and T.~S. Ng, ``{Robust Linear MIMO in the Downlink: A
  Worst-Case Optimization with Ellipsoidal Uncertainty Regions},''
  \emph{EURASIP J. Adv. Signal Process.}, vol. 2008, 2008, {Article ID 609028}.

\bibitem{JR:Robust_error_models2}
C.~Shen, T.-H. Chang, K.-Y. Wang, Z.~Qiu, and C.-Y. Chi, ``{Distributed Robust
  Multicell Coordinated Beamforming With Imperfect CSI: An ADMM Approach},''
  \emph{IEEE Trans. Signal Process.}, vol.~60, pp. 2988--3003, Jun. 2012.

\bibitem{JR:CSI-determinisitic-model}
N.~Vucic and H.~Boche, ``{Robust QoS-Constrained Optimization of Downlink
  Multiuser MISO Systems},'' \emph{IEEE Trans. Signal Process.}, vol.~57, pp.
  714--725, Feb. 2009.

\bibitem{JR:CSI-determinisitic-model2}
J.~Wang and D.~Palomar, ``{Worst-Case Robust MIMO Transmission With Imperfect
  Channel Knowledge},'' \emph{IEEE Trans. Signal Process.}, vol.~57, pp.
  3086--3100, Aug. 2009.

\bibitem{book:convex}
S.~Boyd and L.~Vandenberghe, \emph{{Convex Optimization}}.\hskip 1em plus 0.5em
  minus 0.4em\relax {Cambridge University Press}, 2004.

\bibitem{JR:power_minimziation_beamforming}
Q.~Li and W.-K. Ma, ``{Optimal and Robust Transmit Designs for MISO Channel
  Secrecy by Semidefinite Programming},'' \emph{IEEE Trans. Signal Process.},
  vol.~59, pp. 3799--3812, Aug. 2011.

\bibitem{JR:downlink_beamforming_CR}
I.~Wajid, M.~Pesavento, Y.~Eldar, and D.~Ciochina, ``{Robust Downlink
  Beamforming With Partial Channel State Information for Conventional and
  Cognitive Radio Networks},'' \emph{IEEE Trans. Signal Process.}, vol.~61, pp.
  3656--3670, Jul. 2013.

\bibitem{JR:MOOP}
R.~T. Marler and J.~S. Arora, ``{Survey of Multi-objective Optimization Methods
  for Engineering},'' \emph{{Structural and Multidisciplinary Optimization}},
  vol.~26, pp. 369--395, Apr. 2004.

\bibitem{book:MOOP2}
E.~Bj\"ornson and E.~Jorswieck, \emph{{Optimal Resource Allocation in
  Coordinated Multi-Cell Systems}}.\hskip 1em plus 0.5em minus 0.4em\relax {Now
  Publishers Inc.}, 2013.

\bibitem{CN:SDP_relaxation1}
M.~Bengtsson and B.~Ottersten, ``{Optimal Downlink Beamforming using
  Semidefinite Optimization},'' in \emph{{Proc. Annual Allerton Conf. on
  Commun., Control and Computing}}, Sep. 1999, pp. 987–--996.

\bibitem{JR:SDP_relaxation1}
Z.-Q. Luo, W.-K. Ma, A.-C. So, Y.~Ye, and S.~Zhang, ``{Semidefinite Relaxation
  of Quadratic Optimization Problems},'' \emph{IEEE Trans. Signal Process.},
  vol.~27, pp. 20--34, May 2010.

\bibitem{JR:SeDumi}
J.~F. Sturm, ``{Using SeDuMi 1.02, A MATLAB Toolbox for Optimization over
  Symmetric Cones},'' \emph{{Optimiz. Methods and Software}}, vol. 11-12, pp.
  625--653, Sep. 1999.

\bibitem{cvx}
M.~Grant and S.~Boyd, ``{{CVX}: Matlab Software for Disciplined Convex
  Programming, version 2.0 beta},'' \url{http://cvxr.com/cvx}, Sep. 2013.

\bibitem{book:polynoimal}
T.~H. Cormen, C.~E. Leiserson, and R.~L.~R. amd Clifford~Stein,
  \emph{{Introduction to Algorithms}}, 3rd~ed.\hskip 1em plus 0.5em minus
  0.4em\relax {The MIT Press}, 2009.

\bibitem{JR:complexity1}
B.~Choi and G.~Lee, ``{New Complexity Analysis for Primal-Dual Interior-Point
  Methods for Self-Scaled Optimization Problems},'' \emph{Fixed Point Theory
  and Applications}, no.~1, p. 213, Dec. 2012.

\bibitem{JR:complexity2}
G.~Wang and Y.~Bai, ``{A New Primal-Dual Path-Following Interior-Point
  Algorithm for Semidefinite Optimization},'' \emph{{Journal of Mathematical
  Analysis and Applications}}, vol. 353, pp. 339 -- 349, 2009.

\bibitem{JR:Gaussian_randomization}
N.~Sidiropoulos, T.~Davidson, and Z.-Q. Luo, ``{Transmit Beamforming for
  Physical-Layer Multicasting},'' \emph{IEEE Trans. Signal Process.}, vol.~54,
  pp. 2239--2251, Jun. 2006.

\bibitem{report:tgn}
{IEEE P802.11 Wireless LANs, ``TGn Channel Models", IEEE 802.11-03/940r4},
  Tech. Rep., May 2004.

\bibitem{JR:linear_fractional}
A.~Charnes and W.~W. Cooper, ``{Programming with Linear Fractional
  Functions},'' \emph{{Naval Res. Logist. Quart.}}, vol.~9, pp. 181--186, Apr.
  1962.

\end{thebibliography}
\end{document}